\documentclass[11pt]{article}

\usepackage[utf8]{inputenc}
\usepackage[T1]{fontenc}
\usepackage{multicol}
\usepackage{amsthm}
\usepackage{amsmath}
\usepackage{amssymb}
\usepackage{mathtools}
\usepackage{dsfont}
\usepackage{bm}
\usepackage{bbm}
\usepackage{xparse}
\usepackage{physics}
\usepackage{empheq}
\usepackage{url}
\usepackage{hyperref}
\usepackage[affil-it]{authblk}
\usepackage{tikz}
\usetikzlibrary{calc}
\usetikzlibrary{shapes,arrows}
\usetikzlibrary{spy}

\tikzstyle{block} = [draw, fill=white, rectangle, 
  minimum height=2em, minimum width=3em]
\tikzstyle{bigblock} = [draw, fill=white, rectangle, 
  minimum height=3em, minimum width=4em]
\tikzstyle{Bigblock} = [draw, fill=white, rectangle, 
    minimum height=5em, minimum width=6em]
\tikzstyle{input} = [coordinate]
\tikzstyle{output} = [coordinate]
\tikzstyle{pinstyle} = [pin edge={to-,thin,black}]

\usepackage{rotating}
\usepackage{graphicx}
\usepackage[linesnumbered,ruled,vlined]{algorithm2e}

\usepackage{pgfplots}
\pgfplotsset{compat = newest}

\usepackage[top=3cm,bottom=2cm,right=2cm,left=2cm]{geometry}

\theoremstyle{definition}
\newtheorem{theo}{Theorem}[section]
\newtheorem{lem}[theo]{Lemma}
\newtheorem{cor}[theo]{Corollary}
\newtheorem{prop}[theo]{Proposition}
\newtheorem{defi}[theo]{Definition}

\newtheorem{theo*}{Theorem}

\theoremstyle{remark}
\newtheorem*{rk}{Remark}

\DeclareMathOperator{\maxi}{\text{maximize}}
\DeclareMathOperator{\mini}{\text{minimize}}
\DeclareMathOperator{\st}{\text{subject to}}
\DeclarePairedDelimiter\ceil{\lceil}{\rceil}

\DeclarePairedDelimiterX\set[1]\lbrace\rbrace{#1}

\title{\bfseries Multiple-Access Channel Coding \\
with Non-Signaling Correlations}
\author{Omar Fawzi\footnote{Univ Lyon, ENS Lyon, UCBL, Inria,  LIP, F-69342, Lyon Cedex 07, France. \href{mailto:omar.fawzi@ens-lyon.fr}{\texttt{omar.fawzi@ens-lyon.fr}}} \qquad Paul Fermé\footnote{Univ Lyon, ENS Lyon, UCBL, Inria, LIP, F-69342, Lyon Cedex 07, France. \href{mailto:paul.ferme@ens-lyon.fr}{\texttt{paul.ferme@ens-lyon.fr}}} \footnote{This paper was presented in part at ISIT 2022, see~\cite{FF22}}}
\date{}

\begin{document}
\maketitle

\begin{abstract}
We address the problem of coding for classical multiple-access channels (MACs) with the assistance of non-signaling correlations between parties. It is well-known that non-signaling assistance does not change the capacity of classical point-to-point channels. However, it was recently observed that one can construct MACs from two-player non-local games while relating the winning probability of the game to the capacity of the MAC. By considering games for which entanglement (a special kind of non-signaling correlation) increases the winning probability (e.g., the Magic Square game), this shows that for some specific kinds of channels, entanglement between the senders can increase the capacity.

In this work, we make several contributions towards understanding the capacity region for MACs with the assistance of non-signaling correlations between the parties. We develop a linear program computing the optimal success probability for coding over $n$ copies of a MAC $W$ with size growing polynomially in $n$. Solving this linear program allows us to achieve inner bounds for MACs. Applying this method to the binary adder channel, we show that using non-signaling assistance, the sum-rate $\frac{\log_2(72)}{4} \simeq 1.5425$ can be reached even with zero error, which beats the maximum sum-rate capacity of $1.5$ in the unassisted case. For noisy channels, where the zero-error non-signaling assisted capacity region is trivial, we can use concatenated codes to obtain achievable points in the capacity region. Applied to a noisy version of the binary adder channel, we show that non-signaling assistance still improves the sum-rate capacity. Complementing these achievability results, we give an outer bound on the non-signaling assisted capacity region that has the same expression as the unassisted region except that the channel inputs are not required to be independent. Finally, we show that the capacity region with non-signaling assistance shared only between each sender and the receiver independently is the same as without assistance.
\end{abstract}

\section{Introduction}
Multiple-access channels (MACs for short) are one of the simplest models of network communication settings, where two senders aim to transmit individual messages to one receiver. The capacity of such channels has been entirely characterized by the seminal works by Liao~\cite{Liao73} and Ahlswede~\cite{Ahlswede73} in terms of a simple single-letter formula. From the point of view of quantum information, it is natural to ask whether additional resources, such as quantum entanglement or more generally non-signaling correlations between the parties, change the capacity region. A non-signaling correlation is a multipartite input-output box shared between parties that, as the name suggests, cannot by itself be used to send information between parties. However, non-signaling correlations such as the ones generated by measurements of entangled quantum particles, can provide an advantage for various information processing tasks and nonlocal games. The study of such correlations has given rise to the quantum information area known as nonlocality~\cite{BCPSW14}. For example, in the context of channel coding, there exists classical point-to-point channels for which quantum entanglement between the sender and the receiver can increase the optimal success probability for sending one bit of information with a single use of the channel~\cite{PLMK11,BF18}. However, a well-known result~\cite{BSST99} states that for classical point-to-point channels, entanglement and even more generally non-signaling correlations do not change the capacity of the channel; see also~\cite{Matthews12,BF18}.

In the network setting, behavior is different. Quek and Shor showed in~\cite{QS17} the existence of two-sender two-receiver interference channels with gaps between their classical, quantum-entanglement assisted and non-signaling assisted capacity regions. Following this result, Leditzky et al.~\cite{LALS20} (see also~\cite{SLSS22}) showed that quantum entanglement shared between the two senders of a MAC can strictly enlarge the capacity region. This has been demonstrated through channels that are constructed from two-player non-local games, such as the Magic Square game~\cite{Mermin90,Peres90,Aravind02,BBT05}, by translating known gaps between classical and quantum values of games into MAC capacity gaps. Other instances of network channels for which entanglement increases the capacity region were studied in~\cite{Noetzel20,ND20}. This raises the following natural question: Can non-signaling correlations lead to significant gains in capacity for natural MACs? Can we find a characterization of the capacity region of the MAC when non-signaling resources between the parties are allowed? 

\paragraph{Our Results} We focus here on the MAC with two senders and we allow arbitrary tripartite non-signaling correlations between the two senders and the receiver. This is the most optimistic setting, in the sense that we only enforce the non-signaling constraints between the parties, and also the mathematically simplest setting. Even if not all non-signaling correlations are feasible within quantum theory, the setting we study here can be seen as a tractable and physically motivated outer approximation of what can be achieved with quantum theory. In fact, the quantum set is notoriously complicated and deciding membership in this set is not computable~\cite{JNWY20}. We note that very recently, Pereg et al.~\cite{PDB23} found a multi-letter formula for the capacity of MACs with quantum entanglement shared between the two senders. Unfortunately, this characterization is very difficult to evaluate for any fixed channel.

We denote by $\mathrm{S}^{\mathrm{NS}}(W,k_1,k_2)$ the success probability of the best non-signaling assisted $(k_1,k_2)$-code for the MAC $W$. Contrary to the unassisted value that we denote $\mathrm{S}(W,k_1,k_2)$,  $\mathrm{S}^{\mathrm{NS}}(W,k_1,k_2)$ can be formulated as a linear program; see Proposition~\ref{prop:NSLP}. Furthermore, using symmetries, we have developed a linear program computing $\mathrm{S}^{\mathrm{NS}}$ for a finite number of copies of a MAC $W$ with a size growing polynomially in the number of copies; see Theorem~\ref{theo:polyLP} and Corollary~\ref{cor:poly}. Using this result, we describe a method to derive inner bounds on the non-signaling assisted capacity region achievable with zero error; see Proposition~\ref{prop:IB}. Applied to the binary adder channel, which maps $(x_1, x_2) \in \{0,1\}^2$ to $x_1 + x_2 \in \{0,1,2\}$, we show that the sum-rate $\frac{\log_2(72)}{4} \simeq 1.5425$ can be reached with zero error, which beats the maximum classical sum-rate capacity of $\frac{3}{2}$; see Theorem~\ref{theo:BAC}. For noisy channels, where the zero-error non-signaling assisted capacity region is trivial, we can use concatenated codes to obtain achievable points in the capacity region; see Proposition~\ref{prop:NumericalMethod}. Applied to a noisy version of the binary adder channel, we show that non-signaling assistance still improves the sum-rate capacity. 

In order to find outer bounds, we define a relaxed notion of non-signaling assistance and characterize its capacity region by a single-letter expression, which is the same as the well-known expression for the capacity of the MAC (see Theorem~\ref{theo:capacity}) except that the inputs $X_1$ and $X_2$ are not required to be independent; see Theorem~\ref{theo:CharaNSrelaxed}. This gives in particular an outer bound on the non-signaling assisted capacity region; see Corollary~\ref{cor:OB}. The main open problem that we leave is whether this outer bound on the non-signaling capacity region is tight. We give an example of a channel for which the relaxed notion of non-signaling assistance gives a strictly larger success probability than non-signaling assistance but we do not know if such a gap can persist for the capacity region.

We also study the case where non-signaling assistance is shared only between each sender and the receiver independently. Note that no assistance is shared between the senders. We show that this capacity region is the same as the capacity region without any assistance; see Theorem~\ref{theo:NSsr} and Corollary~\ref{cor:NSsr}. We note that a similar setting with independent entangled states between each sender and the receiver was studied by Hsieh et al.~\cite{HDW08}: a regularized characterization of the capacity region is obtained for any quantum MAC in this setting. It is simple to show using their result that for a classical MAC, this type of entanglement does not change the capacity region given in Theorem~\ref{theo:capacity}.

\paragraph{Organization} In Section~\ref{section:capacity}, we define precisely the different notions of MAC capacities: the classical capacity (i.e. without any assistance) as well as the non-signaling assisted capacity. In Section~\ref{section:complexity}, we address computational complexity questions concerning the probability of success of the best classical coding strategy and the best non-signaling strategy for a MAC. In Section~\ref{section:IB}, we develop numerical methods to find inner bounds on non-signaling assisted capacity regions, and apply those to the binary adder channel and a noisy variant. In Section~\ref{section:OB}, we define our relaxation of non-signaling assistance, we characterize its capacity region by a single-letter formula, and apply those to the binary adder channel. Finally, in Section~\ref{section:NSsr}, we show that the capacity region with non-signaling assistance shared only between each sender and the receiver independently is the same as without assistance.

\section{Multiple Access Channels Capacities}
\label{section:capacity}
\subsection{Classical Capacities}
Formally, a MAC $W$ is a conditional probability distribution depending on two inputs in $\mathcal{X}_1$ and $\mathcal{X}_2$, and an output in $\mathcal{Y}$, so $W := \left(W(y|x_1x_2)\right)_{x_1 \in \mathcal{X}_1,x_2  \in \mathcal{X}_2, y \in \mathcal{Y}}$ such that $W(y|x_1x_2) \geq 0$ and $\sum_{y \in \mathcal{Y}} W(y|x_1x_2) = 1$. We will denote such a MAC by $W : \mathcal{X}_1 \times \mathcal{X}_2 \rightarrow \mathcal{Y}$. The tensor product of two MACs $W: \mathcal{X}_1 \times \mathcal{X}_2 \rightarrow \mathcal{Y}$ and $W': \mathcal{X}'_1 \times \mathcal{X}'_2 \rightarrow \mathcal{Y}'$ is denoted by $W \otimes W' : (\mathcal{X}_1 \times \mathcal{X}'_1) \times (\mathcal{X}_2 \times \mathcal{X}'_2) \to \mathcal{Y} \times \mathcal{Y}'$ and defined by $(W \otimes W')(yy'|x_1x_1'x_2x_2') := W(y|x_1x_2) \cdot W'(y'|x_1'x_2')$. We denote by $W^{\otimes n}(y^n|x_1^nx_2^n) := \prod_{i=1}^nW(y_i|x_{1,i}x_{2,i})$, for $y^n := y_1 \ldots y_n \in \mathcal{Y}^n$, $x_1^n := x_{1,1} \ldots x_{1,n} \in \mathcal{X}_1^n$ and $x_2^n := x_{2,1} \ldots x_{2,n} \in \mathcal{X}_2^n$. We will use the notation $[k]:=\{1,\ldots,k\}$.

The coding problem for a MAC $W : \mathcal{X}_1 \times \mathcal{X}_2 \rightarrow \mathcal{Y}$ is the following: one wants to encode messages in $[k_1]$ into $\mathcal{X}_1$ and messages in $[k_2]$ into $\mathcal{X}_2$ independently, which will be given as input to the channel $W$. This results in a random output in $\mathcal{Y}$, which one needs to decode back into the corresponding messages in $[k_1]$ and $[k_2]$. We will call $e_1 : [k_1] \rightarrow \mathcal{X}_1$ the first encoder, $e_2 : [k_2] \rightarrow \mathcal{X}_2$ the second encoder and $d : \mathcal{Y} \rightarrow [k_1] \times [k_2]$ the decoder. This is depicted in Figure~\ref{fig:MACcoding}.

\begin{figure}[!h]
\begin{center}
  \begin{tikzpicture}[auto, node distance=2cm,>=latex']
    \node [input, name=i1] {};
    \node [block, right of=i1] (e1) {$e_1$};
    \node [bigblock, below right of=e1] (W) {$W$};
    \node [block, below left of=W] (e2) {$e_2$};
    \node [input, left of=e2, name=i2] {};

    \node [block, right of=W] (d) {$d$};

    \draw [->] (e1) -- node[name=x1] {$x_1$} (W);
    \draw [->] (e2) -- node[name=x2] {$x_2$} (W);
    \draw [->] (W) -- node[name=y] {$y$} (d);
    \node [output, right of=d] (j1j2) {};

    \draw [draw,->] (i1) -- node {$i_1$} (e1);
    \draw [draw,->] (i2) -- node {$i_2$} (e2);
    \draw [draw,->] (d) -- node {$(j_1,j_2)$} (j1j2);
  \end{tikzpicture}
\end{center}
\label{fig:MACcoding}
\caption{Coding for a MAC $W$.}
\end{figure}
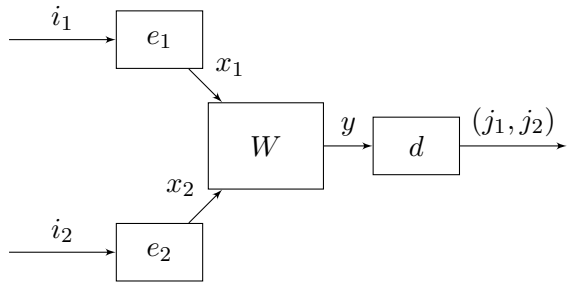

We want to maximize over all encoders $e_1,e_2$ and decoders $d$ the probability of successfully encoding and decoding the messages through $W$, i.e. the probability that $j_1 = i_1$ and  $j_2 = i_2$, given that the input messages are taken uniformly in $[k_1]$ and $[k_2]$. We call this quantity $\mathrm{S}(W,k_1,k_2)$, which is characterized by the following optimization program:

\begin{equation}
\label{eq:def-succ-prob-mac}
  \begin{aligned}
    \mathrm{S}(W,k_1,k_2) := &&\underset{e_1,e_2,d}{\maxi} &&& \frac{1}{k_1k_2} \sum_{i_1,i_2,x_1,x_2,y} W(y|x_1x_2)e_1(x_1|i_1)e_2(x_2|i_2)d(i_1i_2|y)\\
    &&\st &&& \sum_{x_1 \in \mathcal{X}_1} e_1(x_1|i_1) = 1, \forall i_1 \in [k_1]\\
    &&&&&  \sum_{x_2 \in \mathcal{X}_2} e_2(x_2|i_2) = 1, \forall i_2 \in [k_2]\\
    &&&&& \sum_{j_1 \in [k_1],j_2 \in [k_2]} d(j_1j_2|y) = 1, \forall y \in \mathcal{Y}\\
    &&&&& e_1(x_1|i_1), e_2(x_2|i_2), d(j_1j_2|y) \geq 0
  \end{aligned}
\end{equation}

\begin{proof}
One should note that we allow randomized encoders and decoders for generality reasons, although the value of the program is not changed as it is convex. Besides that remark, let us name $I_1,I_2,J_1,J_2,X_1,X_2,Y$ the random variables corresponding to $i_1,i_2,j_1,j_2,x_1,x_2,y$ in the coding and decoding process. Then, for given $e_1,e_2,d$ and $W$, the objective value of the previous program is:

\begin{equation*}
  \begin{aligned}
    &\mathbb{P}\left(J_1 = I_1, J_2 = I_2\right) = \frac{1}{k_1k_2}\sum_{i_1,i_2} \mathbb{P}\left(J_1 = I_1, J_2 = I_2|I_1=i_1,I_2=i_2\right)\\
    &= \frac{1}{k_1k_2}\sum_{i_1,i_2,x_1,x_2}e_1(x_1|i_1)e_2(x_2|i_2) \mathbb{P}\left(J_1 = i_1, J_2 = i_2|I_1=i_1,I_2=i_2,X_1=x_1,X_2=x_2\right)\\
    &= \frac{1}{k_1k_2}\sum_{i_1,i_2,x_1,x_2,y}W(y|x_1x_2)e_1(x_1|i_1)e_2(x_2|i_2) \mathbb{P}\left(J_1 = i_1,J_2=i_2|I_1=i_1,I_2=i_2,X_1=x_1,X_2=x_2,Y=y\right)\\
    &= \frac{1}{k_1k_2}\sum_{i_1,i_2,x_1,x_2,y}W(y|x_1x_2)e_1(x_1|i_1)e_2(x_2|i_2)d(i_1,i_2|y) \ .\\
  \end{aligned}
\end{equation*}
\end{proof}

Since MACs are more general than point-to-point channels (by defining $W(y|x_1x_2):=\hat{W}(y|x_1)$ for $\hat{W}$ a point-to-point channel and looking only at its first input), computing $\mathrm{S}(W,k_1,k_2)$ is \textrm{NP}-hard, and it is even \textrm{NP}-hard to approximate $\mathrm{S}(W,k_1,k_2)$ within a better ratio than $\left(1-e^{-1}\right)$, as a consequence of the hardness result on $\mathrm{S}(W,k)$ shown in~\cite{BF18}.

The (classical) capacity of a MAC, as defined for example in~\cite{CT01}, can be reformulated in the following way:
\begin{defi}[Capacity Region $\mathcal{C}(W)$ of a MAC $W$]
  A rate pair $(R_1,R_2)$ is achievable if:
  \[ \underset{n \rightarrow +\infty}{\lim} \mathrm{S}(W^{\otimes n},\ceil{2^{R_1n}},\ceil{2^{R_2n}}) = 1 \ . \]
  We define the (classical) capacity region $\mathcal{C}(W)$ as the closure of the set of all achievable rate pairs.
\end{defi}

The capacity region $\mathcal{C}(W)$ is characterized by a single-letter formula:

\begin{theo}[Liao~\cite{Liao73} and Ahlswede~\cite{Ahlswede73}]
  \label{theo:capacity}
  $\mathcal{C}(W)$ is the closure of the convex hull of all rate pairs $(R_1,R_2)$ satisfying:
  \[ R_1 < I(X_1:Y|X_2)\ ,\ R_2 < I(X_2:Y|X_1)\ ,\ R_1+R_2 < I((X_1,X_2):Y) \ ,\]
  for $(X_1,X_2) \in \mathcal{X}_1 \times \mathcal{X}_2$ following a product law $P_{X_1} \times P_{X_2}$, and $Y \in \mathcal{Y}$ the outcome of $W$ on inputs $X_1,X_2$.
\end{theo}

For the zero-error (classical) capacity, this leads to the following definition:
\begin{defi}[Zero-Error Capacity Region $\mathcal{C}_0(W)$ of a MAC $W$]
  A rate pair $(R_1,R_2)$ is achievable with zero-error if:
  \[ \exists n_0  \in \mathbb{N}^*, \forall n \geq n_0, \mathrm{S}(W^{\otimes n},\ceil{2^{R_1n}},\ceil{2^{R_2n}}) = 1 \ . \]
  We define the zero-error (classical) capacity region $\mathcal{C}_0(W)$ as the closure of the set of all achievable rate pairs with zero-error.
\end{defi}

We will also consider what we call the sum success probability $\mathrm{S}_{\text{sum}}(W,k_1,k_2)$, defined using $\frac{\mathbb{P}\left(J_1=I_1\right)+\mathbb{P}\left(J_2=I_2\right)}{2}$ rather than $\mathbb{P}\left(J_1=I_1,J_2=I_2\right)$ as an objective value, which leads to the following optimization program:
\begin{equation}
  \begin{aligned}
    \mathrm{S}_{\text{sum}}(W,k_1,k_2) := &&\underset{e_1,e_2,d_1,d_2}{\maxi} &&& \frac{1}{2k_1k_2} \sum_{i_1,i_2,x_1,x_2,y} W(y|x_1x_2)e_1(x_1|i_1)e_2(x_2|i_2)d_1(i_1|y)\\
    &&+&&& \frac{1}{2k_1k_2} \sum_{i_1,i_2,x_1,x_2,y} W(y|x_1x_2)e_1(x_1|i_1)e_2(x_2|i_2)d_2(i_2|y)\\
    &&\st &&& \sum_{x_1 \in \mathcal{X}_1} e_1(x_1|i_1) = 1, \forall i_1 \in [k_1]\\
    &&&&&  \sum_{x_2 \in \mathcal{X}_2} e_2(x_2|i_2) = 1, \forall i_2 \in [k_2]\\
    &&&&& \sum_{j_1 \in [k_1]} d_1(j_1|y) = 1, \forall y \in \mathcal{Y}\\
    &&&&& \sum_{j_2 \in [k_2]} d_2(j_2|y) = 1, \forall y \in \mathcal{Y}\\
    &&&&& e_1(x_1|i_1), e_2(x_2|i_2), d_1(j_1|y), d_2(j_2|y) \geq 0
  \end{aligned}
\end{equation}

Note that we used independent decoders $d_1(j_1|y),d_2(j_2|y)$ rather than a global $d(j_1j_2|y)$ here. This does not change the value of the optimization program. Indeed, since the program is convex, an optimal solution can be found on the extremal points of the search space. Thus, if we had used the variable $d(j_1j_2|y)$, we could always take it to be a function $d$ from $\mathcal{Y}$ to $[k_1]\times[k_2]$. Taking $d_1,d_2$ as the first and second coordinates of that function satisfies the equality $d(j_1j_2|y) = d_1(j_1|y)d_2(j_2|y)$, and therefore, the value of the program is the same in both cases. Note that it is also true for the program computing $\mathrm{S}(W,k_1,k_2)$.

As for the usual (joint) success probability, we can define its capacity region:

\begin{defi}[Sum-Capacity Region $\mathcal{C}_{\text{sum}}(W)$ of a MAC $W$]
  A rate pair $(R_1,R_2)$ is sum-achievable if:
  \[ \underset{n \rightarrow +\infty}{\lim} \mathrm{S}_{\text{sum}}(W^{\otimes n},\ceil{2^{R_1n}},\ceil{2^{R_2n}}) = 1 \ . \]
  We define the sum-capacity region $\mathcal{C}_{\text{sum}}(W)$ as the closure of the set of all sum-achievable rate pairs.
\end{defi}

However, it turns out those two notions of success define the same capacity region:
\begin{prop}
  \label{prop:CapaSumJoint}
  $\mathcal{C}(W) = \mathcal{C}_{\text{sum}}(W)$
\end{prop}
\begin{proof}
  Let us focus on error probabilities rather than success ones. Call them respectively $\mathrm{E}(W,k_1,k_2) := 1-\mathrm{S}(W,k_1,k_2)$ and $\mathrm{E}_{\text{sum}}(W,k_1,k_2) := 1-\mathrm{S}_{\text{sum}}(W,k_1,k_2)$. Let us fix a solution $e_1,d_1,e_2,d_2$ of the optimization program computing $\mathrm{S}(W,k_1,k_2)$. Let us remark first that:

  \[ \sum_{i_1,i_2,x_1,x_2,y} W(y|x_1x_2)e_1(x_1|i_1)e_2(x_2|i_2) = k_1k_2\ , \]

  thus, the value of the maximum error for those encoders and decoders is:

\begin{equation}
  \begin{aligned}
    &\mathrm{E}(W,k_1,k_2,e_1,d_1,e_2,d_2) := 1 -  \frac{1}{k_1k_2}\left(\sum_{i_1,i_2,x_1,x_2,y} W(y|x_1x_2)e_1(x_1|i_1)e_2(x_2|i_2)d_1(i_1|y)d_2(i_2|y)\right)\\
    &=\frac{1}{k_1k_2}\left(\sum_{i_1,i_2,x_1,x_2,y} W(y|x_1x_2)e_1(x_1|i_1)e_2(x_2|i_2)-\sum_{i_1,i_2,x_1,x_2,y} W(y|x_1x_2)e_1(x_1|i_1)e_2(x_2|i_2)d_1(i_1|y)d_2(i_2|y)\right)\\
    &=\frac{1}{k_1k_2}\left(\sum_{i_1,i_2,x_1,x_2,y} W(y|x_1x_2)e_1(x_1|i_1)e_2(x_2|i_2)\left[1-d_1(i_1|y)d_2(i_2|y)\right]\right) \ .\\
  \end{aligned}
\end{equation}

Similarly, the value of the sum error for those encoder and decoders is:
\begin{equation}
  \begin{aligned}
    \mathrm{E}_{\text{sum}}(W,k_1,k_2,e_1,d_1,e_2,d_2) &:= 1 -  \frac{1}{k_1k_2}\left(\sum_{i_1,i_2,x_1,x_2,y} W(y|x_1x_2)e_1(x_1|i_1)e_2(x_2|i_2)\frac{d_1(i_1|y)+d_2(i_2|y)}{2}\right)\\
    &= \frac{1}{k_1k_2}\left(\sum_{i_1,i_2,x_1,x_2,y} W(y|x_1x_2)e_1(x_1|i_1)e_2(x_2|i_2)\left[1-\frac{d_1(i_1|y)+d_2(i_2|y)}{2}\right]\right) \ .
  \end{aligned}
\end{equation}

However, for $x,y \in [0,1]$, we have that:

\[1-xy \geq \max\left(1-x,1-y\right) \geq 1-\frac{x+y}{2} \ , \]
and:
\[1-xy \leq (1-x) + (1-y) = 2  \left(1-\frac{x+y}{2}\right) \ . \]

This means that, for all $e_1,d_1,e_2,d_2$, we have:
\[ \mathrm{E}_{\text{sum}}(W,k_1,k_2,e_1,d_1,e_2,d_2) \leq \mathrm{E}(W,k_1,k_2,e_1,d_1,e_2,d_2) \leq 2\mathrm{E}_{\text{sum}}(W,k_1,k_2,e_1,d_1,e_2,d_2) \ ,\]
so, maximizing over all $e_1,d_1,e_2,d_2$, we get:
\[ \mathrm{E}_{\text{sum}}(W,k_1,k_2) \leq \mathrm{E}(W,k_1,k_2) \leq 2\mathrm{E}_{\text{sum}}(W,k_1,k_2) \ .\]

Thus, up to a multiplicative factor $2$, the error is the same. In particular, when one of those errors tends to zero, the other one tends to zero as well. This implies that the capacity regions are the same.
\end{proof}

\subsection{Non-Signaling Assisted Capacities}
\paragraph{Three-party non-signaling assistance} We now consider the case where the senders and the receiver are given non-signaling assistance. This resource, which is a theoretical but easier to manipulate generalization of quantum entanglement, can be characterized as follows. A tripartite non-signaling box is described by a joint conditional probability distribution $P(x_1x_2(j_1j_2)|i_1i_2y)$ such that the marginal from any two parties is independent from the removed party's input, i.e., we have:
\begin{equation}
  \begin{aligned}
    &&\forall x_2,j_1,j_2,i_1,i_2,y,i_1', &&&\sum_{x_1} P(x_1x_2(j_1j_2)|i_1i_2y) = \sum_{x_1} P(x_1x_2(j_1j_2)|i_1'i_2y) \ ,\\
    &&\forall x_1,j_1,j_2,i_1,i_2,y,i_2', &&&\sum_{x_2} P(x_1x_2(j_1j_2)|i_1i_2y) = \sum_{x_2} P(x_1x_2(j_1j_2)|i_1i_2'y) \ ,\\
    &&\forall x_1,x_2,i_1,i_2,y,y', &&&\sum_{j_1,j_2} P(x_1x_2(j_1j_2)|i_1i_2y) = \sum_{j_1,j_2} P(x_1x_2(j_1j_2)|i_1i_2y') \ .\\
  \end{aligned}
\end{equation}

This implies that one can consider for example $P(x_1x_2|i_1i_2)$ since it does not depend on $y$, or even $P(x_1|i_1)$ since it does not depend on $i_2,y$. Then, in our coding scenario, when the senders and the receiver are given non-signaling assistance, it means that they share a tripartite non-signaling box, the behavior of which is described by $P$. In this case, the expression $e_1(x_1|i_1)e_2(x_2|i_2)d(j_1j_2|y)$ in \eqref{eq:def-succ-prob-mac} is replaced by $P(x_1x_2(j_1j_2)|i_1i_2y)$, as depicted in Figure~\ref{fig:MACNScoding}.

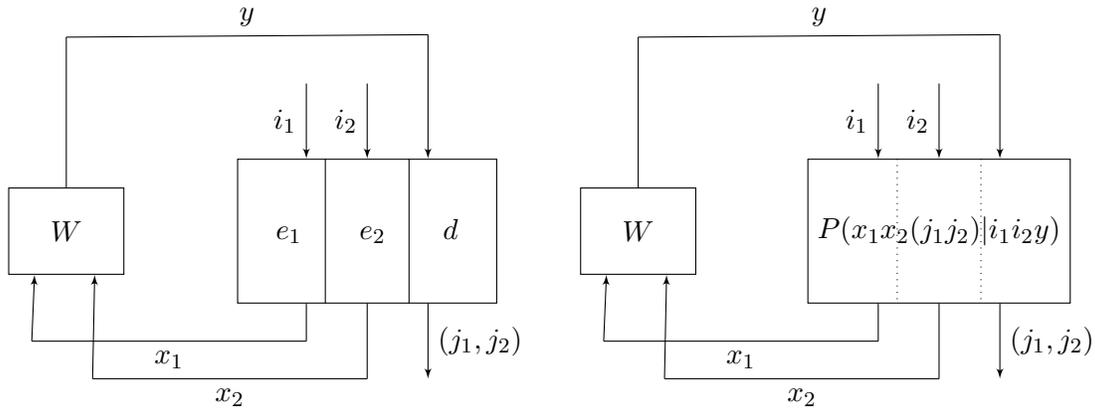
\begin{figure}[!h]
  \begin{center}
      \begin{tikzpicture}[auto, node distance=2cm,>=latex']
    \node [input, name=i1] {};
    \node [input, name=i2] {};
    \node [Bigblock, below of=i2] (P) {$\ \ \ e_1\ \ \ \ \ \ e_2\ \ \ \ \ \ d\ \ \ $};

    \draw (P.120) -- (P.240);
    \draw (P.60) -- (P.300);

    \draw [<-] (P.130) -- node {$i_1$} +(0pt,1cm);
    \draw [<-] (P.90) -- node {$i_2$} +(0pt,1cm);
    \coordinate (ybis) at ($ (P.50) + (0pt,1.65cm) $);
    \draw [<-] (P.50) -- (ybis);
    \coordinate (x1) at ($ (P.230)+(0pt,-0.5cm) $);
    \draw [-] (P.230) -- (x1);
    \coordinate (x2) at ($ (P.270)+(0pt,-1cm) $);
    \draw [-] (P.270) -- (x2);
    \draw [->] (P.310) -- node {$(j_1,j_2)$} +(0pt,-1cm);

    \node [left of=P] (A) {};
    \node [bigblock, left of=A] (W) {$W$};
    \coordinate (x1bis) at ($ (x1)+(-3.65cm,0pt) $);
    \coordinate (x2bis) at ($ (x2)+(-3.65cm,0pt) $);
    \draw [-] (x1) -- node {$x_1$} (x1bis);
    \draw [-] (x2) -- node {$x_2$} (x2bis);
    \draw [->] (x1bis) -- (W.234);
    \draw [->] (x2bis) -- (W.303);
    \coordinate (y) at ($ (W.north)+(0pt,2cm) $) ;
    \draw (W.north) -- (y);
    \draw (y) -- node {$y$} (ybis);
      \end{tikzpicture}
      \ \ \ \
\begin{tikzpicture}[auto, node distance=2cm,>=latex']
    \node [input, name=i1] {};
    \node [input, name=i2] {};
    \node [Bigblock, below of=i2] (P) {$P(x_1x_2(j_1j_2)|i_1i_2y)$};
    
    \draw[dotted] (P.120) -- (P.240);
    \draw[dotted] (P.60) -- (P.300);

    \draw [<-] (P.130) -- node {$i_1$} +(0pt,1cm);
    \draw [<-] (P.90) -- node {$i_2$} +(0pt,1cm);
    \coordinate (ybis) at ($ (P.50) + (0pt,1.65cm) $);
    \draw [<-] (P.50) -- (ybis);
    \coordinate (x1) at ($ (P.230)+(0pt,-0.5cm) $);
    \draw [-] (P.230) -- (x1);
    \coordinate (x2) at ($ (P.270)+(0pt,-1cm) $);
    \draw [-] (P.270) -- (x2);
    \draw [->] (P.310) -- node {$(j_1,j_2)$} +(0pt,-1cm);

    \node [left of=P] (A) {};
    \node [bigblock, left of=A] (W) {$W$};
    \coordinate (x1bis) at ($ (x1)+(-3.65cm,0pt) $);
    \coordinate (x2bis) at ($ (x2)+(-3.65cm,0pt) $);
    \draw [-] (x1) -- node {$x_1$} (x1bis);
    \draw [-] (x2) -- node {$x_2$} (x2bis);
    \draw [->] (x1bis) -- (W.234);
    \draw [->] (x2bis) -- (W.303);
    \coordinate (y) at ($ (W.north)+(0pt,2cm) $) ;
    \draw (W.north) -- (y);
    \draw (y) -- node {$y$} (ybis);
\end{tikzpicture}
\end{center}
\caption{A non-signaling box $P$ replacing $e_1,e_2$ and $d$ in the coding problem for the MAC $W$.}
\label{fig:MACNScoding}
\end{figure}

The cyclicity of Figure~\ref{fig:MACNScoding} is at first sight counter-intuitive. Note first that $P$ being a non-signaling box is completely independent from $W$: in particular, the variable $y$ does not need to follow any law in the definition of $P$ being a non-signaling box. Therefore, the remaining ambiguity is the apparent need to encode and decode at the same time. However, since $P$ is a non-signaling box, we do not need to do both at the same time. Indeed, $\forall y, P(x_1x_2|i_1i_2) = P(x_1x_2|i_1i_2y)$ by the non-signaling property of $P$. Thus, one can get the outputs $x_1,x_2$ on inputs $i_1,i_2$ without access to $y$, as that knowledge won't affect the laws of $x_1,x_2$. Then $y$ follows the law given by $W$ given those $x_1,x_2$. Finally, given access to $y$, the decoding process is described by:
          \[ P((j_1 j_2)|i_1 i_2 y x_1 x_2) = \frac{P(x_1 x_2 (j_1 j_2)|i_1 i_2 y)}{P(x_1 x_2 |i_1 i_2 y)} = \frac{P(x_1 x_2 (j_1 j_2)|i_1 i_2 y)}{P(x_1 x_2 |i_1 i_2)} \ , \]
          so we recover globally $P((j_1 j_2)|i_1 i_2 y x_1 x_2) \times P(x_1 x_2 |i_1 i_2) = P(x_1 x_2 (j_1 j_2)|i_1 i_2 y)$ the prescribed conditional probability. 
          The non-signaling condition ensures that it is possible to consider the conditional probabilities of each party independently. This clarifies how one can effectively encode and then decode messages through a non-signaling box.

As in the unassisted case, we want to maximize over all non-signaling box $P$ the probability of successfully encoding and decoding the messages through $W$, i.e. the probability that $j_1 = i_1$ and  $j_2 = i_2$, given that the input messages are taken uniformly in $[k_1]$ and $[k_2]$. We call this quantity $\mathrm{S^{\mathrm{NS}}}(W,k_1,k_2)$, which is characterized by the following optimization program:

\begin{equation}
  \begin{aligned}
    \mathrm{S}^{\mathrm{NS}}(W,k_1,k_2) := &&\underset{P}{\maxi} &&& \frac{1}{k_1k_2} \sum_{i_1,i_2,x_1,x_2,y} W(y|x_1x_2)P(x_1x_2(i_1i_2)|i_1i_2y)\\
    &&\st &&& \sum_{x_1} P(x_1x_2(j_1j_2)|i_1i_2y) = \sum_{x_1} P(x_1x_2(j_1j_2)|i_1'i_2y)\\
    &&&&& \sum_{x_2} P(x_1x_2(j_1j_2)|i_1i_2y) = \sum_{x_2} P(x_1x_2(j_1j_2)|i_1i_2'y)\\
    &&&&& \sum_{j_1,j_2} P(x_1x_2(j_1j_2)|i_1i_2y) = \sum_{j_1,j_2} P(x_1x_2(j_1j_2)|i_1i_2y')\\
    &&&&& \sum_{x_1,x_2,j_1,j_2} P(x_1x_2(j_1j_2)|i_1i_2y) = 1\\
    &&&&& P(x_1x_2(j_1j_2)|i_1i_2y) \geq 0
  \end{aligned}
\end{equation}

Since it is given as a linear program, the complexity of computing $\mathrm{S}^{\mathrm{NS}}(W,k_1,k_2)$ is polynomial in the number of variables and constraints (see for instance Section 7.1 of~\cite{LinearProgramming}), which is a polynomial in $|\mathcal{X}_1|,|\mathcal{X}_2|,|\mathcal{Y}|,k_1$ and $k_2$. Also, as it is easy to check that a classical strategy is a particular case of a non-signaling assisted strategy, we have that $\mathrm{S}^{\mathrm{NS}}(W,k_1,k_2) \geq \mathrm{S}(W,k_1,k_2)$.

We have then the same definitions of capacity and zero-error capacity:
\begin{defi}[Non-Signaling Assisted Capacity Region $\mathcal{C}^{\mathrm{NS}}(W)$ of a MAC $W$]
  A rate pair $(R_1,R_2)$ is achievable with non-signaling assistance if:
  \[ \underset{n \rightarrow +\infty}{\lim} \mathrm{S}^{\mathrm{NS}}(W^{\otimes n},\ceil{2^{R_1n}},\ceil{2^{R_2n}}) = 1 \ . \]
  We define the non-signaling assisted capacity region $\mathcal{C}^{\mathrm{NS}}(W)$ as the closure of the set of all achievable rate pairs with non-signaling assistance.
\end{defi}

\begin{defi}[Zero-Error Non-Signaling Assisted Capacity Region $\mathcal{C}^{\mathrm{NS}}_0(W)$ of a MAC $W$]
  A rate pair $(R_1,R_2)$ is achievable with zero-error and non-signaling assistance if:
  \[ \exists n_0 \in \mathbb{N}^*, \forall n \geq n_0, \mathrm{S}^{\mathrm{NS}}(W^{\otimes n},\ceil{2^{R_1n}},\ceil{2^{R_2n}}) = 1 \ . \]
  We define the zero-error non-signaling assisted capacity region $\mathcal{C}^{\mathrm{NS}}_0(W)$ as the closure of the set of all achievable rate pairs with zero-error and non-signaling assistance.
\end{defi}

\paragraph{Independent non-signaling assistance} One can also consider the case where non-signaling assistance is shared independently between the first sender and the receiver as well as between the second encoder and the receiver, which we call independent non-signaling assistance. The precise scenario is depicted in Figure~\ref{fig:MACNSsrcoding}:

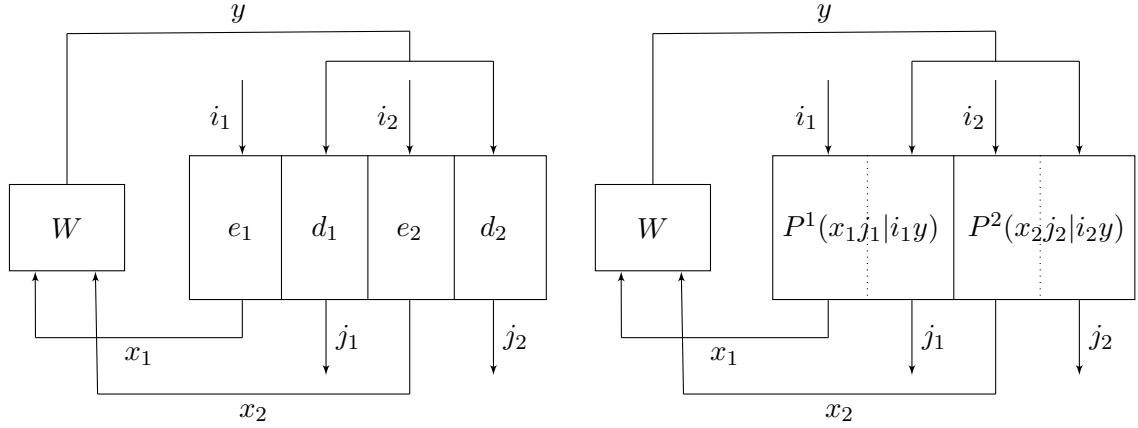
\begin{figure}[!h]
  \begin{center}
      \begin{tikzpicture}[auto, node distance=2cm,>=latex']
    \node [input, name=i1] {};
    \node [input, name=i2] {};
    \node [Bigblock, below of=i2] (P) {$\ \ \ e_1\ \ \ \ \ \ d_1\ \ \ \ \ \ e_2\ \ \ \ \ \ d_2\ \ \ $};

    \draw (P.140) -- (P.220);
    \draw (P.90) -- (P.270);
    \draw (P.40) -- (P.320);

    \draw [<-] (P.150) -- node {$i_1$} +(0pt,1cm);
    \draw [<-] (P.60) -- node {$i_2$} +(0pt,1cm);
    \coordinate (y1bis) at ($ (P.120) + (0pt,1.25cm) $);
    \draw [<-] (P.120) -- (y1bis);
    \coordinate (y2bis) at ($ (P.30) + (0pt,1.25cm) $);
    \draw [<-] (P.30) -- (y2bis);
    \draw [-] (y1bis) -- (y2bis);
    \coordinate (ybis) at ($ (P.60) + (0pt,1.65cm) $);
    \draw [-] (ybis) --  +(0pt,-0.4cm);
    
    \coordinate (x1) at ($ (P.210)+(0pt,-0.5cm) $);
    \draw [-] (P.210) -- (x1);
    \coordinate (x2) at ($ (P.300)+(0pt,-1.25cm) $);
    \draw [-] (P.300) -- (x2);
    \draw [->] (P.240) -- node {$j_1$} +(0pt,-1cm);
    \draw [->] (P.330) -- node {$j_2$} +(0pt,-1cm);
    
    \node [left of=P] (A) {};
    \node [bigblock, left of=A] (W) {$W$};
    \coordinate (x1bis) at ($ (x1)+(-2.75cm,0pt) $);
    \coordinate (x2bis) at ($ (x2)+(-4.15cm,0pt) $);
    \draw [-] (x1) -- node {$x_1$} (x1bis);
    \draw [-] (x2) -- node {$x_2$} (x2bis);
    \draw [->] (x1bis) -- (W.234);
    \draw [->] (x2bis) -- (W.303);
    \coordinate (y) at ($ (W.north)+(0pt,2cm) $) ;
    \draw (W.north) -- (y);
    \draw (y) -- node {$y$} (ybis);
      \end{tikzpicture}
      \ \ \ \
\begin{tikzpicture}[auto, node distance=2cm,>=latex']
    \node [input, name=i1] {};
    \node [input, name=i2] {};
    \node [Bigblock, below of=i2] (P) {$P^1(x_1j_1|i_1y)\ \ \ P^2(x_2j_2|i_2y)$};;
    
    \draw[dotted] (P.140) -- (P.220);
    \draw (P.90) -- (P.270);
    \draw[dotted] (P.40) -- (P.320);

    \draw [<-] (P.150) -- node {$i_1$} +(0pt,1cm);
    \draw [<-] (P.60) -- node {$i_2$} +(0pt,1cm);
    \coordinate (y1bis) at ($ (P.120) + (0pt,1.25cm) $);
    \draw [<-] (P.120) -- (y1bis);
    \coordinate (y2bis) at ($ (P.30) + (0pt,1.25cm) $);
    \draw [<-] (P.30) -- (y2bis);
    \draw [-] (y1bis) -- (y2bis);
    \coordinate (ybis) at ($ (P.60) + (0pt,1.65cm) $);
    \draw [-] (ybis) --  +(0pt,-0.4cm);
    
    \coordinate (x1) at ($ (P.210)+(0pt,-0.5cm) $);
    \draw [-] (P.210) -- (x1);
    \coordinate (x2) at ($ (P.300)+(0pt,-1.25cm) $);
    \draw [-] (P.300) -- (x2);
    \draw [->] (P.240) -- node {$j_1$} +(0pt,-1cm);
    \draw [->] (P.330) -- node {$j_2$} +(0pt,-1cm);
    
    \node [left of=P] (A) {};
    \node [bigblock, left of=A] (W) {$W$};
    \coordinate (x1bis) at ($ (x1)+(-2.75cm,0pt) $);
    \coordinate (x2bis) at ($ (x2)+(-4.15cm,0pt) $);
    \draw [-] (x1) -- node {$x_1$} (x1bis);
    \draw [-] (x2) -- node {$x_2$} (x2bis);
    \draw [->] (x1bis) -- (W.234);
    \draw [->] (x2bis) -- (W.303);
    \coordinate (y) at ($ (W.north)+(0pt,2cm) $) ;
    \draw (W.north) -- (y);
    \draw (y) -- node {$y$} (ybis);
\end{tikzpicture}
\end{center}
\caption{Non-signaling boxes $P^1,P^2$ replacing $e_1,d_1$ and $e_2,d_2$ in the coding problem for the MAC $W$.}
\label{fig:MACNSsrcoding}
\end{figure}

This leads to the following definition of the success probability $\mathrm{S}^{\mathrm{NS}_{\mathrm{SR}}}(W,k_1,k_2)$:
\begin{equation}
  \begin{aligned}
    \mathrm{S}^{\mathrm{NS}_{\mathrm{SR}}}(W,k_1,k_2) := &&\underset{P^1, P^2}{\maxi} &&& \frac{1}{k_1k_2} \sum_{i_1,i_2,x_1,x_2,y} W(y|x_1x_2)P^1(x_1i_1|i_1y)P^2(x_2i_2|i_2y)\\
    &&\st &&& \sum_{x_1} P^1(x_1j_1|i_1y) = \sum_{x_1} P^1(x_1j_1|i_1'y)\\
    &&&&& \sum_{j_1} P^1(x_1j_1|i_1y) = \sum_{j_1} P^1(x_1j_1|i_1y')\\
    &&&&& \sum_{x_1,j_1} P^1(x_1j_1|i_1y) = 1\\
    &&&&& \sum_{x_2} P^2(x_2j_2|i_2y) = \sum_{x_2} P^2(x_2j_2|i_2'y)\\
    &&&&& \sum_{j_2} P^2(x_2j_2|i_2y) = \sum_{j_2} P^2(x_2j_2|i_2y')\\
    &&&&& \sum_{x_2,j_2} P^2(x_2j_2|i_2y) = 1\\
    &&&&& P^1(x_1j_1|i_1y), P^2(x_2j_2|i_2y) \geq 0
  \end{aligned}
  \end{equation}

As before, one can also consider the sum-success probability $\mathrm{S}_{\text{sum}}^{\mathrm{NS}_{\mathrm{SR}}}(W,k_1,k_2)$:
\begin{equation}
  \begin{aligned}
    \mathrm{S}_{\text{sum}}^{\mathrm{NS}_{\mathrm{SR}}}(W,k_1,k_2) := &&\underset{P^1, P^2}{\maxi} &&& \frac{1}{2k_1k_2} \sum_{i_1,i_2,x_1,x_2,y} W(y|x_1x_2)P^1(x_1i_1|i_1y)\sum_{j_2}P^2(x_2j_2|i_2y)\\
    &&+&&& \frac{1}{2k_1k_2} \sum_{i_1,i_2,x_1,x_2,y} W(y|x_1x_2)P^2(x_2i_2|i_2y)\sum_{j_1}P^1(x_1j_1|i_1y)\\
    &&\st &&& \sum_{x_1} P^1(x_1j_1|i_1y) = \sum_{x_1} P^1(x_1j_1|i_1'y)\\
    &&&&& \sum_{j_1} P^1(x_1j_1|i_1y) = \sum_{j_1} P^1(x_1j_1|i_1y')\\
    &&&&& \sum_{x_1,j_1} P^1(x_1j_1|i_1y) = 1\\
    &&&&& \sum_{x_2} P^2(x_2j_2|i_2y) = \sum_{x_2} P^2(x_2j_2|i_2'y)\\
    &&&&& \sum_{j_2} P^2(x_2j_2|i_2y) = \sum_{j_2} P^2(x_2j_2|i_2y')\\
    &&&&& \sum_{x_2,j_2} P^2(x_2j_2|i_2y) = 1\\
    &&&&& P^1(x_1j_1|i_1y), P^2(x_2j_2|i_2y) \geq 0
  \end{aligned}
\end{equation}

\begin{defi}[Independent Non-Signaling Assisted Capacity (resp. Sum-Capacity) Region $\mathcal{C}^{\mathrm{NS}_{\mathrm{SR}}}(W)$ (resp. $\mathcal{C}_{\text{sum}}^{\mathrm{NS}_{\mathrm{SR}}}(W)$)  of a MAC $W$]
  A rate pair $(R_1,R_2)$ is achievable (resp. sum-achievable) with independent non-signaling assistance if:
  \[ \underset{n \rightarrow +\infty}{\lim} \mathrm{S}^{\mathrm{NS}_{\mathrm{SR}}}(W^{\otimes n},\ceil{2^{R_1n}},\ceil{2^{R_2n}}) = 1 \ . \]

   \[ (\text{resp. }\underset{n \rightarrow +\infty}{\lim} \mathrm{S}_{\text{sum}}^{\mathrm{NS}_{\mathrm{SR}}}(W^{\otimes n},\ceil{2^{R_1n}},\ceil{2^{R_2n}}) = 1 \ .) \]
  We define the independent non-signaling assisted capacity (resp. sum-capacity) region $\mathcal{C}^{\mathrm{NS}_{\mathrm{SR}}}(W)$ (resp. $\mathcal{C}_{\text{sum}}^{\mathrm{NS}_{\mathrm{SR}}}(W)$) as the closure of the set of all achievable (resp. sum-achievable) rate pairs with independent non-signaling assistance.
\end{defi}

However, it turns out those two notions of success define the same capacity region:
\begin{prop}
  \label{prop:NSCapaSumJoint}
  $\mathcal{C}^{\mathrm{NS}_{\mathrm{SR}}}(W) = \mathcal{C}_{\text{sum}}^{\mathrm{NS}_{\mathrm{SR}}}(W)$
\end{prop}
\begin{proof}
  Given non-signaling boxes $P^1,P^2$, the maximum success probability of encoding and decoding correctly with those is given by:
\[ \mathrm{S}^{\mathrm{NS}_{\mathrm{SR}}}(W,k_1,k_2,P^1,P^2) := \frac{1}{k_1k_2} \sum_{i_1,i_2,x_1,x_2,y} W(y|x_1x_2)P^1(x_1i_1|i_1y)P^2(x_2i_2|i_2y) \ . \]

This should be compared to the sum success probability of encoding and decoding correctly with those, which we call $\mathrm{S}_{\text{sum}}^{\mathrm{NS}_{\mathrm{SR}}}(W,k_1,k_2,P^1,P^2)$ and is equal to:
\[   \frac{1}{k_1k_2} \sum_{i_1,i_2,x_1,x_2,y} W(y|x_1x_2)\frac{P^1(x_1i_1|i_1y)\sum_{j_2}P^2(x_2j_2|i_2y) + P^2(x_2i_2|i_2y)\sum_{j_1}P^1(x_1j_1|i_1y)}{2} \ . \]

Similarly to what was done in Proposition~\ref{prop:CapaSumJoint}, we focus on error probabilities rather than success probabilities. We have that:

\[ \frac{1}{k_1k_2}\sum_{i_1,i_2,x_1,x_2,y}W(y|x_1x_2)\sum_{j_1,j_2}P^1(x_1j_1|i_1y)P^2(x_2j_2|i_2y) = 1 \ ,\]

so we get that $\mathrm{E}^{\mathrm{NS}_{\mathrm{SR}}}(W,k_1,k_2,P^1,P^2)$ is equal to:

\[\frac{1}{k_1k_2} \sum_{i_1,i_2,x_1,x_2,y} W(y|x_1x_2)\left[\sum_{j_1,j_2}P^1(x_1j_1|i_1y)P^2(x_2j_2|i_2y) - P^1(x_1i_1|i_1y)P^2(x_2i_2|i_2y)\right] \ , \]

and thus:

\[\mathrm{E}^{\mathrm{NS}_{\mathrm{SR}}}(W,k_1,k_2,P^1,P^2) = \sum_{i_1,i_2,x_1,x_2,y} W(y|x_1x_2)\sum_{(j_1,j_2) \not= (i_1,i_2)}P^1(x_1j_1|i_1y)P^2(x_2j_2|i_2y)\ .\]

On the other hand, since:

\[ \sum_{j_1,j_2}P^1(x_1j_1|i_1y)P^2(x_2j_2|i_2y) - P^1(x_1i_1|i_1y)\sum_{j_2}P^2(x_2j_2|i_2y)= \sum_{j_1\not=i_1,j_2}P^1(x_1j_1|i_1y)P^2(x_2j_2|i_2y) \ ,\]

and:

\[ \sum_{j_1,j_2}P^1(x_1j_1|i_1y)P^2(x_2j_2|i_2y) - P^2(x_2i_2|i_2y)\sum_{j_1}P^1(x_1j_1|i_1y)= \sum_{j_1,j_2\not=i_2}P^1(x_1j_1|i_1y)P^2(x_2j_2|i_2y) \ ,\]

we get that $\mathrm{E}_{\text{sum}}^{\mathrm{NS}_{\mathrm{SR}}}(W,k_1,k_2,P^1,P^2)$ is equal to:
\begin{equation}
  \begin{aligned}
    &\frac{1}{k_1k_2}\sum_{i_1,i_2,x_1,x_2,y} W(y|x_1x_2)\left[\frac{\sum_{j_1\not=i_1,j_2}P^1(x_1j_1|i_1y)P^2(x_2j_2|i_2y) + \sum_{j_1,j_2\not=i_2}P^1(x_1j_1|i_1y)P^2(x_2j_2|i_2y)}{2}\right]\\
    &=\frac{1}{k_1k_2}\sum_{i_1,i_2,x_1,x_2,y} W(y|x_1x_2)\left[\sum_{j_1\not=i_1,j_2\not=i_2}P^1(x_1j_1|i_1y)P^2(x_2j_2|i_2y) + \frac{\sum_{(j_1,j_2) \in S}P^1(x_1j_1|i_1y)P^2(x_2j_2|i_2y)}{2}\right]\ ,
  \end{aligned}
\end{equation}

with $S := \{(j_1,i_2) : j_1 \in [k_1]-\{i_1\}\} \cup \{(i_1,j_2) : j_2 \in [k_2]-\{i_2\}\}$. However, we have that:

\begin{equation}
  \begin{aligned}
    &\sum_{j_1\not=i_1,j_2\not=i_2}P^1(x_1j_1|i_1y)P^2(x_2j_2|i_2y) + \frac{\sum_{(j_1,j_2) \in S}P^1(x_1j_1|i_1y)P^2(x_2j_2|i_2y)}{2}\\
    &\leq \sum_{j_1\not=i_1,j_2\not=i_2}P^1(x_1j_1|i_1y)P^2(x_2j_2|i_2y) + \sum_{(j_1,j_2) \in S}P^1(x_1j_1|i_1y)P^2(x_2j_2|i_2y)\\
    &= \sum_{(j_1,j_2) \not= (i_1,i_2)}P^1(x_1j_1|i_1y)P^2(x_2j_2|i_2y)\\
    &\leq 2\left(\sum_{j_1\not=i_1,j_2\not=i_2}P^1(x_1j_1|i_1y)P^2(x_2j_2|i_2y) + \frac{\sum_{(j_1,j_2) \in S}P^1(x_1j_1|i_1y)P^2(x_2j_2|i_2y)}{2}\right) \ .
  \end{aligned}
\end{equation}

This implies that:
\[\mathrm{E}_{\text{sum}}^{\mathrm{NS}_{\mathrm{SR}}}(W,k_1,k_2,P^1,P^2) \leq \mathrm{E}^{\mathrm{NS}_{\mathrm{SR}}}(W,k_1,k_2,P^1,P^2) \leq 2\mathrm{E}_{\text{sum}}^{\mathrm{NS}_{\mathrm{SR}}}(W,k_1,k_2,P^1,P^2) \ ,\]
and by maximizing over all $P^1$ and $P^2$:
\[\mathrm{E}_{\text{sum}}^{\mathrm{NS}_{\mathrm{SR}}}(W,k_1,k_2) \leq \mathrm{E}^{\mathrm{NS}_{\mathrm{SR}}}(W,k_1,k_2) \leq 2\mathrm{E}_{\text{sum}}^{\mathrm{NS}_{\mathrm{SR}}}(W,k_1,k_2) \ .\]
Thus, as before, the capacity regions are the same.
\end{proof}

\section{Properties of Non-Signaling Assisted Codes}
\label{section:complexity}
\subsection{Symmetrization}
One can prove an equivalent formulation of the linear program computing $\mathrm{S}^{\mathrm{NS}}(W,k_1,k_2)$ with a number of variables and constraints polynomial in only $|\mathcal{X}_1|,|\mathcal{X}_2|$ and $|\mathcal{Y}|$ and independent of $k_1$ and $k_2$:
\begin{prop}
  \label{prop:NSLP}
  For a MAC $W: \mathcal{X}_1 \times \mathcal{X}_2 \rightarrow \mathcal{Y}$ and $k_1,k_2 \in \mathbb{N}^*$, we have:
  \begin{equation}
  \begin{aligned}
    \mathrm{S}^{\mathrm{NS}}(W,k_1,k_2) = &&\underset{r,r^1,r^2,p}{\maxi} &&& \frac{1}{k_1k_2} \sum_{x_1,x_2,y} W(y|x_1x_2)r_{x_1,x_2,y}\\
    &&\st &&& \sum_{x_1,x_2} r_{x_1,x_2,y} = 1\\
    &&&&& \sum_{x_1} r^1_{x_1,x_2,y} = k_1 \sum_{x_1} r_{x_1,x_2,y}\\
    &&&&& \sum_{x_2} r^2_{x_1,x_2,y} = k_2 \sum_{x_2} r_{x_1,x_2,y}\\
    &&&&& \sum_{x_1} p_{x_1,x_2} = k_1 \sum_{x_1} r^2_{x_1,x_2,y}\\
    &&&&& \sum_{x_2} p_{x_1,x_2} = k_2 \sum_{x_2} r^1_{x_1,x_2,y}\\
    &&&&& 0 \leq r_{x_1,x_2,y} \leq r^1_{x_1,x_2,y},r^2_{x_1,x_2,y} \leq p_{x_1,x_2}\\
    &&&&& p_{x_1,x_2} -  r^1_{x_1,x_2,y} - r^2_{x_1,x_2,y} + r_{x_1,x_2,y} \geq 0
  \end{aligned}
  \end{equation}

\end{prop}

\begin{proof}
  One can check that given a solution of the original program, the following choice of variables is a valid solution of the second program achieving the same objective value:
\begin{equation}
  \begin{aligned}
    &r_{x_1,x_2,y} := \sum_{i_1,i_2} P(x_1x_2(i_1i_2)|i_1i_2y) \ , &r^1_{x_1,x_2,y} := \sum_{j_1,i_1,i_2} P(x_1x_2(j_1i_2)|i_1i_2y) \ ,\\
    &r^2_{x_1,x_2,y} := \sum_{j_2,i_1,i_2} P(x_1x_2(i_1j_2)|i_1i_2y) \ , &p_{x_1,x_2} := \sum_{j_1,j_2,i_1,i_2} P(x_1x_2(j_1j_2)|i_1i_2y) \ . \\
  \end{aligned}
\end{equation}
Note that $p_{x_1,x_2}$ is well-defined since $\sum_{j_1,j_2,i_1,i_2} P(x_1x_2(j_1j_2)|i_1i_2y)$ is independent from $y$ by since $P$ is a non-signaling box.

For the other direction, given those variables, a non-signaling probability distribution $P(x_1x_2(j_1j_2)|i_1i_2y)$ achieving the same objective value is given by, for $j_1 \not=i_1$ and $j_2 \not= i_2$:
\begin{equation}
  \begin{aligned}
    P(x_1x_2(i_1i_2)|i_1i_2y) &:= \frac{r_{x_1,x_2,y}}{k_1k_2} \ ,\\
    P(x_1x_2(j_1i_2)|i_1i_2y) &:= \frac{r^1_{x_1,x_2,y} - r_{x_1,x_2,y}}{k_1k_2(k_1-1)} \ ,\\
    P(x_1x_2(i_1j_2)|i_1i_2y) &:= \frac{r^2_{x_1,x_2,y} - r_{x_1,x_2,y}}{k_1k_2(k_2-1)} \ ,\\
    P(x_1x_2(j_1j_2)|i_1i_2y) &:= \frac{p_{x_1,x_2} -  r^1_{x_1,x_2,y} - r^2_{x_1,x_2,y} + r_{x_1,x_2,y}}{k_1k_2(k_1-1)(k_2-1)} \ .
  \end{aligned}
\end{equation}
\end{proof}

This symmetrization can also be done for the program computing $\mathrm{S}_{\text{sum}}^{\mathrm{NS}_{\mathrm{SR}}}(W,k_1,k_2)$:
\begin{prop}
  \label{prop:NSsrProgram}
  \begin{equation}
    \begin{aligned}
      \mathrm{S}_{\text{sum}}^{\mathrm{NS}_{\mathrm{SR}}}(W,k_1,k_2) = &&\underset{r^1,r^2,p^1,p^2}{\maxi} &&& \frac{1}{2k_1k_2} \sum_{x_1,x_2,y} W(y|x_1x_2)\left(p^2_{x_2}r^1_{x_1,y} + p^1_{x_1}r^2_{x_2,y}\right)\\
      && = &&& \frac{1}{2}\left[ \frac{1}{k_1}\sum_{x_1,y} W^1_{p^2,k_2}(y|x_1)r^1_{x_1,y} + \frac{1}{k_2}\sum_{x_2,y}W^2_{p^1,k_1}(y|x_2)r^2_{x_2,y} \right]\\
      && \text{with} &&& W^1_{p^2,k_2}(y|x_1) := \frac{1}{k_2}\sum_{x_2} W(y|x_1x_2)p^2_{x_2} ,W^2_{p^1,k_1}(y|x_2) := \frac{1}{k_1}\sum_{x_1} W(y|x_1x_2)p^1_{x_1}\\
      &&\st &&& \sum_{x_1} r^1_{x_1,y} = 1, \sum_{x_2} r^2_{x_2,y} = 1\\
      &&&&& \sum_{x_1} p^1_{x_1} = k_1, \sum_{x_2} p^2_{x_2} = k_2\\
      &&&&& 0 \leq r^1_{x_1,y} \leq p^1_{x_1}, 0 \leq r^2_{x_2,y} \leq p^2_{x_2}
    \end{aligned}
  \end{equation}
\end{prop}

\begin{proof}
  One can check that given a solution of the original program, the following choice of variables is a valid solution of the second program achieving the same objective value:
\begin{equation}
  \begin{aligned}
    &r^1_{x_1,y} := \sum_{i_1} P^1(x_1i_1|i_1y) \ , &p^1_{x_1} := \sum_{j_1,i_1} P^1(x_1j_1|i_1y) \ ,\\
    &r^2_{x_2,y} := \sum_{i_2} P^2(x_2i_2|i_2y) \ , &p^2_{x_2} := \sum_{j_2,i_2} P^2(x_2j_2|i_2y) \ .\\
  \end{aligned}
\end{equation}
Note that $p^1_{x_1}$ and $p^2_{x_2}$ are well-defined since  $\sum_{j_1,i_1} P^1(x_1j_1|i_1y)$ and $\sum_{j_2,i_2} P^2(x_2j_2|i_2y)$ are independent from $y$ since $P^1$ and $P^2$ are non-signaling boxes.

For the other direction, given those variables, non-signaling probability distributions $P^1(x_1j_1|i_1y)$ and $P^2(x_2j_2|i_2y)$ achieving the same objective value are given by, for $j_1 \not=i_1$ and $j_2 \not= i_2$:
\begin{equation}
  \begin{aligned}
    P^1(x_1i_1|i_1y) &:= \frac{r^1_{x_1,y}}{k_1} \ ,\\
    P^1(x_1j_1|i_1y) &:= \frac{p^1_{x_1,y} - r^1_{x_1,y}}{k_1(k_1-1)} \ ,\\
    P^2(x_2i_2|i_2y) &:= \frac{r^2_{x_2,y}}{k_2} \ ,\\
    P^2(x_2j_2|i_2y) &:= \frac{p^2_{x_2,y} - r^2_{x_2,y}}{k_2(k_2-1)} \ .\\
  \end{aligned}
\end{equation}
\end{proof}

\subsection{Properties of $\mathrm{S}^{\mathrm{NS}}(W,k_1,k_2)$, $\mathcal{C}^{\mathrm{NS}}(W)$ and $\mathcal{C}^{\mathrm{NS}}_0(W)$}
\begin{defi}
  We say that a conditional probability distribution $P(a^n|x^n)$ defined on $\bigtimes_{i=1}^n\mathcal{A}_i \times \bigtimes_{i=1}^n \mathcal{X}_i$ is \emph{non-signaling} if for all $a^n, x^n, \hat{x}^n$, we have
    \[ \forall i \in [n], \sum_{\hat{a}_i}P(a_1\ldots \hat{a}_i \ldots a_n|x_1\ldots x_i \ldots x_n) = \sum_{\hat{a}_i}P(a_1\ldots \hat{a}_i \ldots a_n|x_1\ldots \hat{x}_i \ldots x_n) \ .\]
\end{defi}

\begin{defi}
  Let $P(a^n|x^n)$ a conditional probability distribution defined on $\bigtimes_{i=1}^n\mathcal{A}_i \times \bigtimes_{i=1}^n \mathcal{X}_i$ and $P'(a'^n|x'^n)$ defined on $\bigtimes_{i=1}^n\mathcal{A}'_i \times \bigtimes_{i=1}^n \mathcal{X}'_i$. We define $P \otimes P'$ the tensor product conditional probability distribution defined on $\bigtimes_{i=1}^n(\mathcal{A}_i \times \mathcal{A}'_i) \times \bigtimes_{i=1}^n (\mathcal{X}_i \times \mathcal{X}'_i)$ by $\left(P \otimes P'\right)(a_1a'_1\ldots a_na'_n|x_1x'_1\ldots x_nx'_n) := P(a^n|x^n) \cdot P'(a'^n|x'^n)$.
\end{defi}

\begin{lem}
  \label{lem:NStensor}
  If both $P$ and $P'$ are non-signaling, then $P \otimes P'$ is non-signaling.
\end{lem}
\begin{proof}
  Let $a^n \in \bigtimes_{j=1}^n\mathcal{A}_j$, $a'^n \in \bigtimes_{j=1}^n\mathcal{A}'_j$, $x^n \in \bigtimes_{j=1}^n \mathcal{X}_j$, $x'^n \in \bigtimes_{j=1}^n \mathcal{X}'_j$ and $\hat{x}_i \in \mathcal{X}_i$, $\hat{x}'_i \in \mathcal{X}'_i$. Using the fact that $P,P'$ are non-signaling, we have:

  \begin{equation}
    \begin{aligned}
      &\sum_{\hat{a}_i\hat{a}_i'}P(a_1a'_1\ldots \hat{a}_i\hat{a}_i' \ldots a_na'_n|x_1x'_1\ldots x_ix'_i \ldots x_nx'_n)\\
      = &\sum_{\hat{a}_i\hat{a}_i'}  P(a_1\ldots \hat{a}_i \ldots a_n|x_1\ldots x_i \ldots x_n) \cdot P'(a'_1\ldots \hat{a}'_i \ldots a'_n|x'_1\ldots x'_i \ldots x'_n)\\
      = &\left(\sum_{\hat{a}_i}  P(a_1\ldots \hat{a}_i \ldots a_n|x_1\ldots x_i \ldots x_n)\right) \cdot \left(\sum_{\hat{a}'_i}  P'(a'_1\ldots \hat{a}'_i \ldots a'_n|x'_1\ldots x'_i \ldots x'_n)\right)\\
      = &\left(\sum_{\hat{a}_i}  P(a_1\ldots \hat{a}_i \ldots a_n|x_1\ldots \hat{x}_i \ldots x_n)\right) \cdot \left(\sum_{\hat{a}'_i}  P'(a'_1\ldots \hat{a}'_i \ldots a'_n|x'_1\ldots \hat{x}'_i \ldots x'_n)\right)\\
      = &\sum_{\hat{a}_i\hat{a}_i'}\left(P \otimes P'\right)(a_1a'_1\ldots \hat{a}_i\hat{a}'_i \ldots a_na'_n|x_1x'_1\ldots \hat{x}_i\hat{x}'_i \ldots x_nx'_n) \ ,
    \end{aligned}
  \end{equation}
  so $P \otimes P'$ is non-signaling.
\end{proof}

\begin{prop}
  \label{prop:oneShot}
  For a MAC $W: \mathcal{X}_1 \times \mathcal{X}_2 \rightarrow \mathcal{Y}$ and $k_1,k_2 \in \mathbb{N}^*$, we have:
  \begin{enumerate}
  \item $\frac{1}{k_1k_2} \leq \mathrm{S}^{\mathrm{NS}}(W,k_1,k_2) \leq 1$.
  \item $\mathrm{S}^{\mathrm{NS}}(W,k_1,k_2) \leq \min\left(\frac{|\mathcal{X}_1|}{k_1}, \frac{|\mathcal{X}_2|}{k_2}, \frac{|\mathcal{Y}|}{k_1k_2} \right)$.
  \item If $k_1' \leq k_1$ and $k_2' \leq k_2$, then $\mathrm{S}^{\mathrm{NS}}(W,k_1',k_2') \geq \mathrm{S}^{\mathrm{NS}}(W,k_1,k_2)$.
    
  \item For any MAC $W': \mathcal{X}'_1 \times \mathcal{X}'_2 \rightarrow \mathcal{Y}'$ and $k_1,k_2 \in \mathbb{N}^*$, we have $\mathrm{S}^{\mathrm{NS}}(W \otimes W',k_1k_1',k_2k_2') \geq \mathrm{S}^{\mathrm{NS}}(W,k_1,k_2) \cdot \mathrm{S}^{\mathrm{NS}}(W',k'_1,k'_2)$. In particular, for any positive integer $n$, $\mathrm{S}^{\mathrm{NS}}(W^{\otimes n},k_1^n,k_2^n) \geq \left[\mathrm{S}^{\mathrm{NS}}(W,k_1,k_2)\right]^n$ and $\mathrm{S}^{\mathrm{NS}}(W \otimes W',k_1,k_2) \geq \mathrm{S}^{\mathrm{NS}}(W,k_1,k_2)$.
  \end{enumerate}
\end{prop}
\begin{proof}
  \begin{enumerate}
  \item Let us first show that $\mathrm{S}^{\mathrm{NS}}(W,k_1,k_2) \geq \frac{1}{k_1k_2}$. Take $p_{x_1,x_2} := \frac{k_1k_2}{|\mathcal{X}_1||\mathcal{X}_2|}$, $r^1_{x_1,x_2,y} := \frac{p_{x_1,x_2}}{k_2}$, $r^2_{x_1,x_2,y} := \frac{p_{x_1,x_2}}{k_1}$ and $r_{x_1,x_2,y} := \frac{p_{x_1,x_2}}{k_1k_2} = \frac{1}{|\mathcal{X}_1||\mathcal{X}_2|}$. One can easily check that it is indeed a valid solution of the linear program computing $\mathrm{S}^{\mathrm{NS}}(W,k_1,k_2)$. Thus we have:
    \begin{equation}
      \begin{aligned}
        \mathrm{S}^{\mathrm{NS}}(W,k_1,k_2) &\geq \frac{1}{k_1k_2} \sum_{x_1,x_2,y} W(y|x_1x_2)r_{x_1,x_2,y} = \frac{1}{k_1k_2} \sum_{x_1,x_2} \frac{1}{|\mathcal{X}_1||\mathcal{X}_2|} \sum_y W(y|x_1x_2)\\
        &= \frac{1}{k_1k_2} \sum_{x_1,x_2} \frac{1}{|\mathcal{X}_1||\mathcal{X}_2|} = \frac{1}{k_1k_2} \ .
      \end{aligned}
    \end{equation}

    Furthermore, in order to show that it is at most $1$, let us consider an optimal solution of $\mathrm{S}^{\mathrm{NS}}(W,k_1,k_2)$. We have:
    \begin{equation}
      \begin{aligned}
        \mathrm{S}^{\mathrm{NS}}(W,k_1,k_2) &= \frac{1}{k_1k_2} \sum_{x_1,x_2,y} W(y|x_1x_2)r_{x_1,x_2,y} \leq \frac{1}{k_1k_2} \sum_{x_1,x_2,y} W(y|x_1x_2)p_{x_1,x_2}\\
        &= \frac{1}{k_1k_2} \sum_{x_1,x_2} p_{x_1,x_2} \sum_y W(y|x_1x_2) = \frac{1}{k_1k_2} \sum_{x_1,x_2} p_{x_1,x_2} = 1 \ ,
      \end{aligned}
    \end{equation}
    since $\sum_{x_1,x_2} p_{x_1,x_2} = k_1 \sum_{x_1,x_2} r^2_{x_1,x_2,y} = k_1k_2 \sum_{x_1,x_2} r_{x_1,x_2,y} =k_1k_2$.
    
  \item First let us show that $\mathrm{S}^{\mathrm{NS}}(W,k_1,k_2) \leq \frac{|\mathcal{X}_1|}{k_1}$ (the case $\mathrm{S}^{\mathrm{NS}}(W,k_1,k_2) \leq \frac{|\mathcal{X}_2|}{k_2}$ is symmetric):
    \begin{equation}
      \begin{aligned}
        \mathrm{S}^{\mathrm{NS}}(W,k_1,k_2) &= \frac{1}{k_1k_2} \sum_{x_1,x_2,y} W(y|x_1x_2)r_{x_1,x_2,y} \leq \frac{1}{k_1k_2} \sum_{x_1,x_2,y} W(y|x_1x_2)r^2_{x_1,x_2,y}\\
        &\leq \frac{1}{k_1k_2} \sum_{x_2,y} \left(\sum_{x_1'} W(y|x_1'x_2)\right)\cdot\left(\sum_{x_1} r^2_{x_1,x_2,y} \right) \quad \text{since nonnegative terms.}\\
        &=  \frac{1}{k_1k_2} \sum_{x_2,y} \left(\sum_{x_1'} W(y|x_1'x_2)\right)\cdot\left(\frac{1}{k_1} \sum_{x_1} p_{x_1,x_2} \right)\\
        &= \frac{1}{k_1^2k_2} \sum_{x_1,x_2} p_{x_1,x_2} \sum_{x_1'}\left(\sum_y W(y|x_1'x_2)\right) = \frac{|\mathcal{X}_1|}{k_1^2k_2} \sum_{x_1,x_2} p_{x_1,x_2} = \frac{|\mathcal{X}_1|}{k_1} \ .
      \end{aligned}
    \end{equation}

    Let us show now that $\mathrm{S}^{\mathrm{NS}}(W,k_1,k_2) \leq \frac{|\mathcal{Y}|}{k_1k_2}$:
    \begin{equation}
      \begin{aligned}
        \mathrm{S}^{\mathrm{NS}}(W,k_1,k_2) &= \frac{1}{k_1k_2} \sum_{x_1,x_2,y} W(y|x_1x_2)r_{x_1,x_2,y} \leq \frac{1}{k_1k_2} \sum_{y} \left(\max_{x_1,x_2} W(y|x_1x_2)\right) \sum_{x_1,x_2} r_{x_1,x_2,y}\\
        &\leq \frac{1}{k_1k_2} \sum_{y} \sum_{x_1,x_2} r_{x_1,x_2,y} = \frac{|\mathcal{Y}|}{k_1k_2} \ .
      \end{aligned}
    \end{equation}
  \item Let us assume that $k'_1 \leq k_1$ and that $k'_2 = k_2$, since this latter case will follow by symmetry. Consider an optimal solution of $\mathrm{S}^{\mathrm{NS}}(W,k_1,k_2) = \frac{1}{k_1} \sum_{i_1 \in [k_1]} f(i_1)$ with:
    \[ f(i_1) := \frac{1}{k_2} \sum_{x_1,x_2,y,i_2} W(y|x_1x_2)P(x_1x_2(i_1i_2)|i_1i_2y) \ ,\]
    and $P$ non-signaling. Let us consider $S \in \underset{S' \subseteq [k_1]: |S'|=k'_1}{\text{argmax}} \sum_{i_1 \in S'} f(i_1)$. Then, by construction, we have that $\frac{1}{k'_1}\sum_{i_1 \in S} f(i_1) \geq \frac{1}{k_1} \sum_{i_1 \in [k_1]} f(i_1) = \mathrm{S}^{\mathrm{NS}}(W,k_1,k_2)$, since we have taken the average of the $k'_1$ largest values of the sum.

    Let us define the strategy $P'$ on the smallest set $\mathcal{X}_1\times\mathcal{X}_2\times\left(S \times [k_2]\right)\times S \times [k_2] \times \mathcal{Y}$:
    \begin{equation}
      \begin{aligned}
        P'(x_1x_2(j_1j_2)|i_1i_2y) &:= P(x_1x_2(j_1j_2)|i_1i_2y) + C(x_1x_2j_2|i_1i_2y) \ ,\\
        \text{with } C(x_1x_2j_2|i_1i_2y) &:= \frac{1}{k'_1}\sum_{j'_1 \in [k_1]-S}P(x_1x_2(j'_1j_2)|i_1i_2y) \ .
      \end{aligned}
    \end{equation}

    $P'$ is a correct conditional probability distribution. Indeed, it is nonnegative by construction, and we have that:
    \begin{equation}
      \begin{aligned}
      &\sum_{x_1,x_2,j_1 \in S,j_2} P'(x_1x_2(j_1j_2)|i_1i_2y) = \sum_{x_1,x_2,j_1 \in S,j_2} P(x_1x_2(j_1j_2)|i_1i_2y) + \sum_{x_1,x_2,j_1 \in S,j_2} C(x_1x_2j_2|i_1i_2y)\\
      &= \sum_{x_1,x_2j_2} \sum_{j_1 \in S} P(x_1x_2(j_1j_2)|i_1i_2y) + \sum_{x_1,x_2,j_2}\sum_{j_1 \in S} \frac{1}{k'_1}\sum_{j'_1 \in [k_1]-S}P(x_1x_2(j'_1j_2)|i_1i_2y)\\
        &= \sum_{x_1,x_2,j_2}\sum_{j_1 \in S} P(x_1x_2(j_1j_2)|i_1i_2y) + \sum_{x_1,x_2,j_2} \sum_{j'_1 \in [k_1]-S}P(x_1x_2(j'_1j_2)|i_1i_2y)\\
      &= \sum_{x_1,x_2,j_1,j_2} P(x_1x_2(j_1j_2)|i_1i_2y) = 1 \ .
      \end{aligned}
    \end{equation}
    
    Let us show that $P'$ is non-signaling:
    \begin{enumerate}
    \item First with $x_1$:
    \begin{equation}
      \begin{aligned}
      \sum_{x_1} P'(x_1x_2(j_1j_2)|i_1i_2y) &= \sum_{x_1} P(x_1x_2(j_1j_2)|i_1i_2y) + \sum_{x_1} C(x_1x_2j_2|i_1i_2y)\\
      &= \sum_{x_1} P(x_1x_2(j_1j_2)|i_1i_2y) + \frac{1}{k'_1}\sum_{j'_1 \in [k_1]-S}\sum_{x_1}P(x_1x_2(j'_1j_2)|i_1i_2y)\\
      &= \sum_{x_1} P(x_1x_2(j_1j_2)|i'_1i_2y) + \frac{1}{k'_1}\sum_{j'_1 \in [k_1]-S}\sum_{x_1}P(x_1x_2(j'_1j_2)|i'_1i_2y)\\
      &\text{since $P$ is non-signaling.}\\
      &= \sum_{x_1} P'(x_1x_2(j_1j_2)|i'_1i_2y) \ .
      \end{aligned}
    \end{equation}
    
    \item Then with $x_2$:
    \begin{equation}
      \begin{aligned}
      \sum_{x_2} P'(x_1x_2(j_1j_2)|i_1i_2y) &= \sum_{x_2} P(x_1x_2(j_1j_2)|i_1i_2y) + \sum_{x_2} C(x_1x_2j_2|i_1i_2y)\\
      &= \sum_{x_2} P(x_1x_2(j_1j_2)|i_1i_2y) + \frac{1}{k'_1}\sum_{j'_1 \in [k_1]-S}\sum_{x_2}P(x_1x_2(j'_1j_2)|i_1i_2y)\\
      &= \sum_{x_2} P(x_1x_2(j_1j_2)|i_1i'_2y) + \frac{1}{k'_1}\sum_{j'_1 \in [k_1]-S}\sum_{x_2}P(x_1x_2(j'_1j_2)|i_1i'_2y)\\
      &\text{since $P$ is non-signaling.}\\
      &= \sum_{x_2} P'(x_1x_2(j_1j_2)|i_1i'_2y) \ .
      \end{aligned}
    \end{equation}
    \item Finally with $(j_1j_2)$:
    \begin{equation}
      \begin{aligned}
      &\sum_{j_1 \in S,j_2} P'(x_1x_2(j_1j_2)|i_1i_2y) = \sum_{j_2}\sum_{j_1 \in S} P(x_1x_2(j_1j_2)|i_1i_2y) + \sum_{j_2}\sum_{j_1 \in S} C(x_1x_2j_2|i_1i_2y)\\
      &= \sum_{j_2}\sum_{j_1 \in S} P(x_1x_2(j_1j_2)|i_1i_2y) + \sum_{j_2}\sum_{j_1 \in S}\frac{1}{k'_1}\sum_{j'_1 \in [k_1]-S}P(x_1x_2(j'_1j_2)|i_1i_2y)\\
      &= \sum_{j_2}\sum_{j_1 \in S} P(x_1x_2(j_1j_2)|i_1i_2y) + \sum_{j_2}\sum_{j'_1 \in [k_1]-S}P(x_1x_2(j'_1j_2)|i_1i_2y)\\
      &= \sum_{j_1,j_2} P(x_1x_2(j_1j_2)|i_1i_2y)\\
      &= \sum_{j_1,j_2} P(x_1x_2(j_1j_2)|i_1i_2y') \text{ since $P$ is non-signaling.}\\
      &= \sum_{j_1 \in S,j_2} P'(x_1x_2(j_1j_2)|i_1i_2y') \ .
      \end{aligned}
    \end{equation}

       Thus $P'$ is a correct solution of the program computing $\mathrm{S}^{\mathrm{NS}}(W,k'_1,k_2)$, and it leads to the value:
    \begin{equation}
      \begin{aligned}
         \mathrm{S}^{\mathrm{NS}}(W,k'_1,k_2) &\geq \frac{1}{k'_1k_2} \sum_{x_1,x_2,y,i_1 \in S,i_2} W(y|x_1x_2)P'(x_1x_2(i_1i_2)|i_1i_2y)\\
         &\geq \frac{1}{k'_1k_2} \sum_{x_1,x_2,y,i_1 \in S,i_2} W(y|x_1x_2)P(x_1x_2(i_1i_2)|i_1i_2y)\\
         &= \frac{1}{k'_1} \sum_{i_1 \in S} f(i_1) \geq \frac{1}{k_1} \sum_{i_1 \in [k_1]} f(i_1) = \mathrm{S}^{\mathrm{NS}}(W,k_1,k_2) \ .
      \end{aligned}
    \end{equation}
    \end{enumerate}

  \item Consider optimal non-signaling probability distributions $P$ and $P'$ reaching respectively the values $\mathrm{S}^{\mathrm{NS}}(W,k_1,k_2)$ and $\mathrm{S}^{\mathrm{NS}}(W',k'_1,k'_2)$. Then by Lemma~\ref{lem:NStensor}, $P \otimes P'$ is a non-signaling probability distribution on $\left(\mathcal{X}_1 \times \mathcal{X}'_1\right) \times \left(\mathcal{X}_2 \times \mathcal{X}'_2\right) \times \left(([k_1]\times[k'_1]) \times ([k_2]\times[k'_2])\right) \times ([k_1]\times[k'_1]) \times ([k_2]\times[k'_2]) \times \left(\mathcal{Y} \times \mathcal{Y}'\right)$, which is trivially in bijection with $\left(\mathcal{X}_1 \times \mathcal{X}'_1\right) \times \left(\mathcal{X}_2 \times \mathcal{X}'_2\right) \times \left([k_1k'_1] \times [k_2k'_2]\right) \times [k_1k'_1] \times [k_2k'_2] \times \left(\mathcal{Y} \times \mathcal{Y}'\right)$. This gives a valid solution of the program computing $\mathrm{S}^{\mathrm{NS}}(W \otimes W',k_1k'_1,k_2k'_2)$. Thus, we get that $\mathrm{S}^{\mathrm{NS}}(W \otimes W',k_1k'_1,k_2k'_2)$ is larger than or equal to:
    \begin{equation}
      \begin{aligned}
        &\sum_{x_1x_1',x_2x_2',yy',i_1i_1',i_2i_2'} \left(W \otimes W'\right)(yy'|x_1x_1'x_2x_2')\left(P \otimes P'\right)(x_1x_1'x_2x_2'(i_1i_1'i_2i_2')|i_1i_1',i_2i_2'yy')\\
        &= \sum_{x_1x_1',x_2x_2',yy',i_1i_1',i_2i_2'} \left(W(y|x_1x_2) \cdot W'(y'|x'_1x'_2)\right)\left( P(x_1x_2(i_1i_2)|i_1i_2y) \cdot P'(x'_1x'_2(i'_1i'_2)|i'_1i'_2y') \right)\\
        &= \left(\sum_{i_1,i_2,x_1,x_2,y} W(y|x_1x_2)P(x_1x_2(i_1i_2)|i_1i_2y)\right)\cdot\left(\sum_{x'_1,x'_2,y',i'_1,i'_2} W'(y'|x'_1x'_2)P'(x'_1x'_2(i'_1i'_2)|i'_1i'_2y')\right)\\
        &= \mathrm{S}^{\mathrm{NS}}(W,k_1,k_2) \cdot \mathrm{S}^{\mathrm{NS}}(W',k'_1,k'_2) \ .
      \end{aligned}
    \end{equation}

    In particular, applying this $n$ times on the same MAC $W$ gives the first corollary, and the second one comes from the fact that $\mathrm{S}^{\mathrm{NS}}(W \otimes W',k_1,k_2) \geq \mathrm{S}^{\mathrm{NS}}(W,k_1,k_2) \cdot \mathrm{S}^{\mathrm{NS}}(W',1,1) = \mathrm{S}^{\mathrm{NS}}(W,k_1,k_2)$, since $\mathrm{S}^{\mathrm{NS}}(W',1,1) = 1$ by the first property of Proposition~\ref{prop:oneShot}.
  \end{enumerate}
\end{proof}

\begin{cor}
  \label{cor:asymptotic}
  \begin{enumerate}
      \item $\mathcal{C}^{\mathrm{NS}}(W)$ is convex.
      \item If $(R_1,R_2)$ is achievable with non-signaling assistance, then $R_1 \leq \log_2|\mathcal{X}_1|$, $R_2 \leq \log_2|\mathcal{X}_2|$ and $R_1 + R_2 \leq \log_2|\mathcal{Y}|$.
      \item If $(R_1,R_2)$ is achievable with non-signaling assistance, then for all $R_i' \leq R_i$, $(R_1',R_2')$ is achievable with non-signaling assistance.
  \end{enumerate}
\end{cor}

\begin{proof}
  \begin{enumerate}
    \item Let $(R_1,R_2)$ and $(\tilde{R}_1,\tilde{R}_2)$, two pairs of rational rates achievable with non-signaling assistance for $W$, ie:
    \[ \mathrm{S}^{\mathrm{NS}}(W^{\otimes n},\ceil{2^{R_1n}},\ceil{2^{R_2n}}) \underset{n \rightarrow +\infty}{\rightarrow} 1 \text{ and }\mathrm{S}^{\mathrm{NS}}(W^{\otimes n},\ceil{2^{\tilde{R}_1n}},\ceil{2^{\tilde{R}_2n}}) \underset{n \rightarrow +\infty}{\rightarrow} 1 \ . \]
    Let $\lambda \in (0,1)$ rational and define $R_{\lambda,i} := \lambda \cdot R_i + (1-\lambda) \cdot \tilde{R}_i$, let us show that $(R_{\lambda,1},R_{\lambda,2})$ is achievable with non-signaling assistance. Let us call respectively $k_i:=2^{R_i}, \tilde{k}_i := 2^{\tilde{R}_i}, k_{\lambda,i} := 2^{R_{\lambda,i}} = k_i^{\lambda}\cdot k_i^{(1-\lambda)}$.

    We have $R_{\lambda,i}n = \lambda \cdot R_in + (1-\lambda) \cdot \tilde{R}_in = (\lambda n) \cdot R_i + (1-\lambda)n \cdot \tilde{R}_i$. This is the idea of \emph{time-sharing}: for $\lambda n$ copies of the MAC, we use the strategy with rate $(R_1,R_2)$ and for the $(1-\lambda)n$ other copies of the MAC, we use the strategy with rate $(\tilde{R}_1,\tilde{R}_2)$. There exists some $n$ such that $\lambda n,(1-\lambda)n,\lambda n R_i,(1-\lambda)n \tilde{R}_i$ are integers, since everything is rational. This implies that $k_i^{\lambda n},\tilde{k}_i^{(1-\lambda)n},k_{\lambda,i}^n$ are integers. Thus, thanks to the fourth property of Proposition~\ref{prop:oneShot}, we have:
    \begin{equation}
      \begin{aligned}
        \mathrm{S}^{\mathrm{NS}}(W^{\otimes n},k^n_{\lambda,1},k^n_{\lambda, 2}) &\geq \mathrm{S}^{\mathrm{NS}}(W^{\otimes (\lambda n)}, k_1^{\lambda n}, k_2^{\lambda n}) \cdot \mathrm{S}^{\mathrm{NS}}(W^{\otimes ((1-\lambda) n)}, \tilde{k}_1^{(1-\lambda) n}, \tilde{k}_2^{(1-\lambda) n})\\
        &\underset{n \rightarrow +\infty}{\rightarrow} 1 \cdot 1 = 1 \ .
      \end{aligned}
    \end{equation}

    Thus in particular, since we have $\mathrm{S}^{\mathrm{NS}}(W^{\otimes n},k^n_{\lambda,1},k^n_{\lambda,2}) \leq 1$, we get that $\mathrm{S}^{\mathrm{NS}}(W^{\otimes n},k^n_{\lambda,1},k^n_{\lambda,2}) \underset{n \rightarrow +\infty}{\rightarrow} 1$, so $(R_{\lambda,1},R_{\lambda,2})$ is achievable with non-signaling assistance. Finally, since $\mathcal{C}^{\mathrm{NS}}(W)$ is defined as the closure of achievable rates with non-signaling assistance, we get that $\mathcal{C}^{\mathrm{NS}}(W)$ is convex.
    
  \item By the second property of Proposition~\ref{prop:oneShot}, we have that $\mathrm{S}^{\mathrm{NS}}(W^{\otimes n},k_1^n,k_2^n) \leq \frac{|\mathcal{X}_1^n|}{k_1^n}$. In particular, if one takes $R_1 > \log_2 |\mathcal{X}_1|$, then $k_1 > |\mathcal{X}_1|$ and we get that $\mathrm{S}^{\mathrm{NS}}(W^{\otimes n},k_1^n,k_2^n) \leq \left(\frac{|\mathcal{X}_1|}{k_1}\right)^n \underset{n \rightarrow +\infty}{\rightarrow} 0$, so $R_1 > \log_2 |\mathcal{X}_1|$ is not achievable with non-signaling assistance. Symmetrically, $R_2 > \log_2 |\mathcal{X}_2|$ is not achievable with non-signaling assistance.

    Furthermore, if one takes $R_1+R_2 > \log_2|\mathcal{Y}|$, then in particular $k_1k_2 > |\mathcal{Y}|$, so by the second property of Proposition~\ref{prop:oneShot}, $\mathrm{S}^{\mathrm{NS}}(W^{\otimes n},k_1^n,k_2^n) \leq \frac{|\mathcal{Y}^n|}{k_1^nk_2^n} = \left(\frac{|\mathcal{Y}|}{k_1k_2}\right)^n \underset{n \rightarrow +\infty}{\rightarrow} 0$. Thus, $R_1+R_2 > \log_2|\mathcal{Y}|$ is not achievable with non-signaling assistance.
    
  \item Since $(R_1,R_2)$ is achievable with non-signaling assistance, we have $\mathrm{S}^{\mathrm{NS}}(W^{\otimes n},\ceil{2^{nR_1}},\ceil{2^{nR_2}}) \underset{n \rightarrow +\infty}{\rightarrow} 1$. But, for all positive integer $n$, we have that $\ceil{2^{nR'_1}} \leq \ceil{2^{nR_1}}$ and $\ceil{2^{nR'_2}}\leq \ceil{2^{nR_2}}$, so by the third property of Proposition~\ref{prop:oneShot}, we have that $\mathrm{S}^{\mathrm{NS}}(W^{\otimes n},\ceil{2^{nR'_1}},\ceil{2^{nR'_2}}) \geq \mathrm{S}^{\mathrm{NS}}(W^{\otimes n},\ceil{2^{nR_1}},\ceil{2^{nR_2}})$. Thus $\mathrm{S}^{\mathrm{NS}}(W^{\otimes n},\ceil{2^{nR'_1}},\ceil{2^{nR'_2}} \underset{n \rightarrow +\infty}{\rightarrow} 1$ since it is upper bounded by $1$, and so $(R'_1,R'_2)$ is achievable with non-signaling assistance.
  \end{enumerate}
\end{proof}

\begin{prop}
  \label{prop:ZENScapacity}
   $\mathcal{C}^{\mathrm{NS}}_0(W)$ is the closure of the set of rate pairs $(R_1,R_2)$ such that:
  \[ \exists n \in \mathbb{N}^*, \mathrm{S}^{\mathrm{NS}}(W^{\otimes n},\ceil{2^{R_1n}},\ceil{2^{R_2n}}) = 1 \ . \]
\end{prop}
\begin{proof}
  It is clear that if $(R_1,R_2)$ is such that $\exists n_0 \in \mathbb{N}^*, \forall n \geq n_0, \mathrm{S}^{\mathrm{NS}}(W^{\otimes n},\ceil{2^{R_1n}},\ceil{2^{R_2n}}) = 1$, then in particular $\exists n \in \mathbb{N}^*, \mathrm{S}^{\mathrm{NS}}(W^{\otimes n},\ceil{2^{R_1n}},\ceil{2^{R_2n}}) = 1$. So, $\mathcal{C}^{\mathrm{NS}}_0(W)$, which is the closure of the former rate pairs, is in particular included in the closure of the latter rate pairs.
  
  For the other inclusion, consider a rate pair $(R_1,R_2)$ and let us assume that there exists some positive integer $n$ such that $\mathrm{S}^{\mathrm{NS}}(W^{\otimes n},\ceil{2^{R_1n}},\ceil{2^{R_2n}}) = 1$. Let us show that for any $(R_1',R_2')$ such that $R'_1 < R_1$ and $R'_2 < R_2$:
  \[ \exists n_0 \in \mathbb{N}^*, \forall n \geq n_0, \mathrm{S}^{\mathrm{NS}}(W^{\otimes n},\ceil{2^{R_1'n}},\ceil{2^{R_2'n}}) = 1 \ , \]
  which is enough to conclude, since we consider only closure of such sets.
  
  First, for all positive integer $m$, we have that $\mathrm{S}^{\mathrm{NS}}(W^{\otimes nm},\ceil{2^{R_1nm}},\ceil{2^{R_2nm}}) = 1$. By the fourth property of Proposition~\ref{prop:oneShot}, we have that $\mathrm{S}^{\mathrm{NS}}(\left(W^{\otimes n}\right)^{\otimes m},\ceil{2^{R_1n}}^m,\ceil{2^{R_2n}}^m) \geq \left[\mathrm{S}^{\mathrm{NS}}(W^{\otimes n},\ceil{2^{R_1n}},\ceil{2^{R_2n}})\right]^m = 1$, so $\mathrm{S}^{\mathrm{NS}}(\left(W^{\otimes n}\right)^{\otimes m},\ceil{2^{R_1n}}^m,\ceil{2^{R_2n}}^m) = 1$ since $\mathrm{S}^{\mathrm{NS}}(W,k_1,k_2) \leq 1$ by the first property of Proposition~\ref{prop:oneShot}. But $\left(W^{\otimes n}\right)^{\otimes m} = W^{\otimes nm}$, and $\ceil{2^{R_1n}}^m \geq \ceil{2^{R_1nm}},\ceil{2^{R_2n}}^m \geq \ceil{2^{R_2nm}}$, so by the third property of Proposition~\ref{prop:oneShot}, we have $\mathrm{S}^{\mathrm{NS}}(W^{\otimes nm},\ceil{2^{R_1nm}},\ceil{2^{R_2nm}}) \geq 1$, so $\mathrm{S}^{\mathrm{NS}}(W^{\otimes nm},\ceil{2^{R_1nm}},\ceil{2^{R_2nm}})=1$.

  Then, consider some $r \in \{0,\ldots,n-1\}$. By the fourth property of Proposition~\ref{prop:oneShot}, we have that:
  \begin{equation}
    \begin{aligned}
      \mathrm{S}^{\mathrm{NS}}(W^{\otimes (nm+r)},\ceil{2^{R_1nm}},\ceil{2^{R_2nm}}) &= \mathrm{S}^{\mathrm{NS}}(W^{\otimes nm} \otimes W^{\otimes r},\ceil{2^{R_1nm}},\ceil{2^{R_2nm}})\\
      &\geq \mathrm{S}^{\mathrm{NS}}(W^{\otimes nm},\ceil{2^{R_1nm}},\ceil{2^{R_2nm}}) = 1 \ ,
    \end{aligned}
  \end{equation}
  so $\mathrm{S}^{\mathrm{NS}}(W^{\otimes (nm+r)},\ceil{2^{R_1nm}},\ceil{2^{R_2nm}}) = 1$. But $\ceil{2^{R_1nm}} = \ceil{2^{\frac{R_1nm}{nm+r}(nm+r)}} = \ceil{2^{\frac{R_1}{1+\delta}(nm+r)}}$ with $\delta = \frac{r}{nm} \leq \frac{1}{m}$, and symmetrically $\ceil{2^{R_1nm}} = \ceil{2^{\frac{R_1}{1+\delta}(nm+r)}}$. Thus in particular, for all $R'_1 \leq \frac{R_1}{1+\frac{1}{m}}$ and $R'_2 \leq \frac{R_2}{1+\frac{1}{m}}$, we have that for all $n'\geq nm, \mathrm{S}^{\mathrm{NS}}(W^{\otimes n'},\ceil{2^{R_1'n'}},\ceil{2^{R_2'n'}}) = 1$. So for any $(R_1',R_2')$ such that $R'_1 < R_1$ and $R'_2 < R_2$, there is large enough $m$ such that $R'_1 \leq \frac{R_1}{1+\frac{1}{m}}$ and $R'_2 \leq \frac{R_2}{1+\frac{1}{m}}$, and thus we get the expected property on $(R_1',R_2')$ for $n_0:=nm$.
\end{proof}

\subsection{Linear Program with Reduced Size for Structured Channels}

Although $\mathrm{S}^{\mathrm{NS}}(W,k_1,k_2)$ can be computed in polynomial time in $W$, $k_1$ and $k_2$, a channel of the form $W^{\otimes n}$ has exponential size in $n$. Thus, the linear program for $\mathrm{S}^{\mathrm{NS}}(W^{\otimes n},k_1,k_2)$ grows exponentially with $n$. However, using the invariance of $W^{\otimes n}$ under permutations, one can find a much smaller linear program computing $\mathrm{S}^{\mathrm{NS}}(W^{\otimes n},k_1,k_2)$.

\begin{defi} Let $G$ a group acting on $\mathcal{X}_1,\mathcal{X}_2,\mathcal{Y}$. We say that a MAC $W : \mathcal{X}_1 \times \mathcal{X}_2 \rightarrow \mathcal{Y}$ is \emph{invariant under the action of $G$} if:
  \[\forall g \in G, W(g \cdot y|g \cdot x_1 g \cdot x_2)=W(y|x_1x_2) \ .\]
\end{defi}

In particular, for  $W^{\otimes n}: \mathcal{X}_1^n \times \mathcal{X}_2^n \rightarrow \mathcal{Y}^n$, the symmetric group $G:=S_n$ acts in a natural way in any set $\mathcal{A}$ raised to power $n$. So for $\sigma \in S_n$, we have that:
\[ W^{\otimes n}(\sigma \cdot y^n|\sigma \cdot x_1^n \sigma \cdot x_2^n) =\prod_{i=1}^n W(y_{\sigma(i)}|x_{1,\sigma(i)}x_{2,\sigma(i)}) = \prod_{i=1}^n W(y_{i}|x_{1,i}x_{2,i}) = W^{\otimes n}(y^n|x_1^nx_2^n) \ ,\]
and so $W^{\otimes n}$ is invariant under the action of $S_n$.

Let $\mathcal{Z} := \{\mathcal{X}_1, \mathcal{X}_2, \mathcal{Y}, \mathcal{X}_1 \times \mathcal{Y}, \mathcal{X}_2 \times \mathcal{Y}, \mathcal{X}_1 \times \mathcal{X}_2,  \mathcal{X}_1 \times \mathcal{X}_2 \times \mathcal{Y} \}$. Let us call $\mathcal{O}_G(\mathcal{A})$ the set of orbits of $\mathcal{A}$ under the action of $G$. Then, one can find an equivalent smaller linear program for $\mathrm{S}^{\mathrm{NS}}(W,k_1,k_2)$:
\begin{theo}
  \label{theo:polyLP}
  Let $W : \mathcal{X}_1 \times \mathcal{X}_2 \rightarrow \mathcal{Y}$ a MAC invariant under the action of $G$. Let us name systematically $w \in \mathcal{O}_G(\mathcal{X}_1 \times \mathcal{X}_2 \times \mathcal{Y}), u \in \mathcal{O}_G(\mathcal{X}_1 \times \mathcal{X}_2), u^1 \in \mathcal{O}_G(\mathcal{X}_1), u^2 \in \mathcal{O}_G(\mathcal{X}_2), v^1 \in \mathcal{O}_G(\mathcal{X}_1 \times \mathcal{Y}), v^2 \in \mathcal{O}_G(\mathcal{X}_2 \times \mathcal{Y}), v \in \mathcal{O}_G(\mathcal{Y})$. We will also call $z_{\mathcal{A}}$ the projection of $z \in \mathcal{O}_G(\mathcal{B})$ on $\mathcal{A}$, for $\mathcal{A},\mathcal{B} \in \mathcal{Z}$ and $\mathcal{A}$ projection of $\mathcal{B}$; note that $z_{\mathcal{A}} \in \mathcal{O}_G(\mathcal{A})$, since by definition of the action, the projection of an orbit is an orbit. Let us finally call $W(w) := W(y|x_1x_2)$ for any $(x_1,x_2,y) \in w$, which is well-defined since $W$ is invariant under $G$.
  We have that $\mathrm{S}^{\mathrm{NS}}(W,k_1,k_2)$ is the solution of the following linear program:
  \begin{equation}
  \begin{aligned}
    \mathrm{S}^{\mathrm{NS}}(W,k_1,k_2) = &&\underset{r,r^1,r^2,p}{\maxi} &&& \frac{1}{k_1k_2}\sum_{w \in \mathcal{O}_G(\mathcal{X}_1 \times \mathcal{X}_2 \times \mathcal{Y})} W(w)r_w\\
    &&\st &&& \sum_{w:w_{\mathcal{Y}}=v} r_w = |v|, \forall v \in \mathcal{O}_G(\mathcal{Y})\\
    &&&&& \sum_{w:w_{\mathcal{X}_2\mathcal{Y}}=v^2} r^1_{w} = k_1 \sum_{w:w_{\mathcal{X}_2\mathcal{Y}}=v^2} r_w, \: \forall v^2 \in \mathcal{O}_G(\mathcal{X}_2 \times \mathcal{Y})\\
    &&&&& \sum_{w:w_{\mathcal{X}_1\mathcal{Y}}=v^1} r^2_{w} = k_2 \sum_{w:w_{\mathcal{X}_1\mathcal{Y}}=v^1} r_w, \: \forall v^1 \in \mathcal{O}_G(\mathcal{X}_1 \times \mathcal{Y})\\
    &&&&& \sum_{u:u_{\mathcal{X}_2}=v^2_{\mathcal{X}_2}} p_u = \frac{|v^2_{\mathcal{X}_2}|}{|v^2|} k_1 \sum_{w:w_{\mathcal{X}_2\mathcal{Y}}=v^2} r^2_w, \: \forall v^2 \in \mathcal{O}_G(\mathcal{X}_2 \times \mathcal{Y})\\
    &&&&& \sum_{u:u_{\mathcal{X}_1}=v^1_{\mathcal{X}_1}} p_u = \frac{|v^1_{\mathcal{X}_1}|}{|v^1|} k_2 \sum_{w:w_{\mathcal{X}_1\mathcal{Y}}=v^1} r^1_w, \: \forall v^1 \in \mathcal{O}_G(\mathcal{X}_1 \times \mathcal{Y})\\
    &&&&& 0 \leq r_w \leq r^1_w,r^2_w \leq \frac{|w|}{|w_{\mathcal{X}_1\mathcal{X}_2}|}p_{w_{\mathcal{X}_1\mathcal{X}_2}}, \: \forall w \in \mathcal{O}_G(\mathcal{X}_1 \times \mathcal{X}_2 \times \mathcal{Y})\\
    &&&&& \frac{|w|}{|w_{\mathcal{X}_1\mathcal{X}_2}|}p_{w_{\mathcal{X}_1\mathcal{X}_2}} -  r^1_w - r^2_w + r_w \geq 0, \: \forall w \in \mathcal{O}_G(\mathcal{X}_1 \times \mathcal{X}_2 \times \mathcal{Y}) \ .\\
  \end{aligned}
  \end{equation}
\end{theo}

\begin{cor}
  For a channel $W : \mathcal{X}_1 \times \mathcal{X}_2 \to \mathcal{Y}$, $\mathrm{S}^{\mathrm{NS}}(W^{\otimes n},k_1,k_2)$ is the solution of a linear program of size bounded by $O\left(n^{|\mathcal{X}_1|\cdot|\mathcal{X}_2 |\cdot|\mathcal{Y}|-1}\right)$, thus it can be computed in polynomial time in $n$.
  \label{cor:poly}
\end{cor}
\begin{proof}
  We use the linear program obtained in Theorem~\ref{theo:polyLP} with $G:=S_n$ acting on $W^{\otimes n}$ as described before. 
  The number of variables and constraints is linear in the number of orbits of the action of $S^n$ on the different sets $\mathcal{A} \in \mathcal{Z}$, where here $\mathcal{Z} = \{\mathcal{X}_1^n, \mathcal{X}_2^n, \mathcal{Y}^n, \mathcal{X}^n_1 \times \mathcal{Y}^n, \mathcal{X}^n_2 \times \mathcal{Y}^n, \mathcal{X}^n_1 \times \mathcal{X}^n_2,  \mathcal{X}^n_1 \times \mathcal{X}^n_2 \times \mathcal{Y}^n \}$. For example, for $\mathcal{A} \in \mathcal{X}^n_1 \times \mathcal{X}^n_2 \times \mathcal{Y}^n$, we have that:

\[ |\mathcal{O}_{S_n}(\mathcal{X}^n_1 \times \mathcal{X}^n_2 \times \mathcal{Y}^n)| = \binom{n + |\mathcal{X}_1||\mathcal{X}_2| |\mathcal{Y}|   - 1}{|\mathcal{X}_1| |\mathcal{X}_2| |\mathcal{Y}| - 1} \leq (n+|\mathcal{X}_1||\mathcal{X}_2| |\mathcal{Y}|-1)^{|\mathcal{X}_1||\mathcal{X}_2| |\mathcal{Y}| - 1} \ .\]

So the number of variables and constraints is $O(n^{|\mathcal{X}_1|\cdot|\mathcal{X}_2 |\cdot|\mathcal{Y}|-1})$. Note also that all the numbers occurring this linear program are integers or fractions of integers, with those integers ranging in $\left[\left(|\mathcal{X}_1||\mathcal{X}_2||\mathcal{Y}|\right)^n\right]$, thus of size $O(n\log(|\mathcal{X}_1||\mathcal{X}_2||\mathcal{Y}|))$. So the size of this linear program is bounded by $O(n^{|\mathcal{X}_1|\cdot|\mathcal{X}_2 |\cdot|\mathcal{Y}|-1})$, and thus $\mathrm{S}^{\mathrm{NS}}(W^{\otimes n},k_1,k_2)$ can be computed in polynomial time in $n$; see for instance Section 7.1 of~\cite{LinearProgramming}.
\end{proof}

In order to prove Theorem~\ref{theo:polyLP}, we will need several lemmas. For all of them, $\mathcal{A}$ and $\mathcal{B}$ will denote finite sets on which a group $G$ is acting, and $x^G$ will denote the orbit of $x$ under $G$:

\begin{lem}
  \label{lem:orbitProjCard}
  Let $\tau \in \mathcal{O}_G(\mathcal{A} \times \mathcal{B})$, and call $\nu := \tau_{\mathcal{A}}$ and $\mu := \tau_{\mathcal{B}}$. For $x \in \nu$, let us call $B_{\tau}^x := \left\{y:(x,y)\in\tau\right\}$. Then, $|B_{\tau}^x|=|B_{\tau}^{x'}|=:c_{\tau}^{\nu}$ for any $x,x' \in \nu$, and furthermore, we have that $c_{\tau}^{\nu}=\frac{|\tau|}{|\nu|}$. Symmetrically, the same occurs for $A_{\tau}^y := \{x:(x,y)\in\tau\}$ with $y \in \mu$, where one gets that $|A_{\tau}^y|=|A_{\tau}^{y'}|=:c_{\tau}^{\mu}=\frac{|\tau|}{|\mu|}$ for $y,y' \in \mu$.
\end{lem}

\begin{proof}
  Let $x,x' \in \nu$. Thus there exists $g \in G$ such that $x' = g \cdot x$. Let:
  \[ \begin{array}{ccccc}
    f & : & B_{\tau}^x  & \to & B_{\tau}^{x'} \\
    & & y & \mapsto & g \cdot y \ .\\
  \end{array}\]
  First, $f$ is well defined. Indeed, if $y \in B_{\tau}^x = \left\{y:(x,y)\in\tau\right\}$, then $g \cdot y \in \left\{y:(g \cdot x,y)\in\tau\right\} = B_{\tau}^{x'}$, since $\tau \in \mathcal{O}_G(\mathcal{A} \times \mathcal{B})$. Let us show that $f$ is injective. If $g \cdot y = g \cdot y'$, then $g^{-1} \cdot (g \cdot y) = (g^{-1}g) \cdot y = y$, $g^{-1} \cdot (g \cdot y') = y'$, so $y=y'$. Thus we get that $|B_{\tau}^x|\leq|B_{\tau}^{x'}|$. By a symmetric argument with $x'$ replacing $x$ and $g^{-1}$ replacing $g$, we get that $|B_{\tau}^{x'}|\leq|B_{\tau}^x|$, and so $|B_{\tau}^x|=|B_{\tau}^{x'}|=:c_{\tau}^{\nu}$.

  Furthermore, $\{ B_{\tau}^x \}_{x \in \nu}$ is a partition of $\tau$, so $\sum_{x \in \nu} |B_{\tau}^x| = |\nu| c_{\tau}^{\nu} = |\tau|$, and thus $c_{\tau}^{\nu}=\frac{|\tau|}{|\nu|}$.
\end{proof}

\begin{lem}
  \label{lem:fromOrbitToEle} For any $(x,y) \in \mathcal{A} \times \mathcal{B}$ and $v_{(x,y)^G}$ variable indexed by orbits of $\mathcal{A} \times \mathcal{B}$, let us define the variable $v_{x,y} := \frac{v_{(x,y)^G}}{|(x,y)^G|}$. We have:

  \[ \sum_{x \in \mathcal{A}} v_{x,y} = \frac{1}{|y^G|}\sum_{\tau \in \mathcal{O}_G(\mathcal{A} \times \mathcal{B}): \tau_{\mathcal{B}}=y^G} v_{\tau}, \forall y \in \mathcal{B} \ .\]
\end{lem}

\begin{proof}
  \begin{equation}
    \begin{aligned}
      \sum_{x \in \mathcal{A}} v_{x,y} &= \sum_{\tau \in \mathcal{O}_G(\mathcal{A} \times \mathcal{B}):\tau_{\mathcal{B}}=y^G}\sum_{x \in \mathcal{A}:(x,y) \in \tau} v_{x,y}\\
      &= \sum_{\tau \in \mathcal{O}_G(\mathcal{A} \times \mathcal{B}):\tau_{\mathcal{B}}=y^G}\sum_{x \in \mathcal{A}:(x,y) \in \tau} \frac{v_{\tau}}{|\tau|} \quad \text{since } (x,y)^G=\tau\\
      &= \sum_{\tau \in \mathcal{O}_G(\mathcal{A} \times \mathcal{B}):\tau_{\mathcal{B}}=y^G} c^{y^G}_{\tau} \frac{v_{\tau}}{|\tau|} \quad \text{by Lemma~\ref{lem:orbitProjCard}, since } y \in \tau_{\mathcal{B}}\\
      &= \sum_{\tau \in \mathcal{O}_G(\mathcal{A} \times \mathcal{B}):\tau_{\mathcal{B}}=y^G} \frac{|\tau|}{|y^G|} \frac{v_{\tau}}{|\tau|} = \frac{1}{|y^G|}\sum_{\tau \in \mathcal{O}_G(\mathcal{A} \times \mathcal{B}): \tau_{\mathcal{B}}=y^G} v_{\tau} \ .
    \end{aligned}
  \end{equation}
\end{proof}

\begin{lem}
  \label{lem:fromEleToOrbit}
  For any $\tau \in \mathcal{O}_G(\mathcal{A} \times \mathcal{B})$, $\mu \in \mathcal{O}_G(\mathcal{B})$ and $v_{x,y}$ variable indexed by elements of $\mathcal{A} \times \mathcal{B}$, let us define $v_{\tau} := \sum_{(x,y) \in \tau} v_{x,y}$. We have:

  \[ \sum_{\tau  \in \mathcal{O}_G(\mathcal{A} \times \mathcal{B}): \tau_{\mathcal{B}}=\mu} v_{\tau} = \sum_{y \in \mu} \sum_{x \in \mathcal{A}} v_{x,y} \ .\]
\end{lem}

\begin{proof}
  \[ \sum_{\tau \in \mathcal{O}_G(\mathcal{A} \times \mathcal{B}):\tau_{\mathcal{B}}=\mu} v_{\tau} =  \sum_{\tau \in \mathcal{O}_G(\mathcal{A} \times \mathcal{B}):\tau_{\mathcal{B}}=\mu} \sum_{(x,y) \in \tau} v_{x,y} = \sum_{y \in \mu} \sum_{x \in \mathcal{A}} v_{x,y} \ . \]
\end{proof}

\begin{proof}[Proof of Theorem~\ref{theo:polyLP}]
  Let $r_{x_1,x_2,y},r^1_{x_1,x_2,y},r^2_{x_1,x_2,y},p_{x_1,x_2}$ a feasible solution of the program defined in Proposition~\ref{prop:NSLP}, and call $S:=\frac{1}{k_1k_2} \sum_{x_1,x_2,y} W(y|x_1x_2)r_{x_1,x_2,y}$ its value. Define:
  \begin{equation}
    \begin{aligned}
      &r_w := \sum_{(x_1,x_2,y) \in w} r_{x_1,x_2,y} \ , &r^1_w := \sum_{(x_1,x_2,y) \in w} r^1_{x_1,x_2,y} \ ,\\
      &r^2_w := \sum_{(x_1,x_2,y) \in w} r^2_{x_1,x_2,y} \ , &p_u := \sum_{(x_1,x_2) \in u} p_{x_1,x_2} \ .
    \end{aligned}
  \end{equation}
  Let us show that $r_w,r^1_w,r^2_w,p_u$ is a feasible solution of the program defined in Theorem~\ref{theo:polyLP}, and that its value $S^*:=\frac{1}{k_1k_2}\sum_{w} W(w)r_w = S$.

  First, we have  $S^* = S$. Indeed:
   \begin{equation}
     \begin{aligned}
       S^*&=\frac{1}{k_1k_2}\sum_{w} W(w)r_w =\frac{1}{k_1k_2}\sum_{w} W(w) \sum_{(x_1,x_2,y) \in w} r_{x_1,x_2,y}\\
       &=\frac{1}{k_1k_2}\sum_{w} \sum_{(x_1,x_2,y) \in w} W(y|x_1x_2) r_{x_1,x_2,y} \quad \text{since $W(w)= W(y|x_1x_2)$ for all $(x_1,x_2,y) \in w$}\\
       &= \frac{1}{k_1k_2}\sum_{x_1,x_2,y} W(y|x_1x_2) r_{x_1,x_2,y} = S \ .
     \end{aligned}
  \end{equation}

   Then, all the constraints are satisfied. Indeed, thanks to Lemma~\ref{lem:fromEleToOrbit}, we have for the first constraint:
   \[ \sum_{w:w_{\mathcal{Y}}=v} r_w = \sum_{y \in v} \sum_{x_1,x_2} r_{x_1,x_2,y} = \sum_{y \in v} 1 = |v| \ .\]

   For the second constraint (and symmetrically for the third constraint), we have:
   \[ \sum_{w:w_{\mathcal{X}_2\mathcal{Y}}=v^2} r^1_{w} = \sum_{(x_2,y) \in v^2} \sum_{x_1} r^1_{x_1,x_2,y} = \sum_{(x_2,y) \in v^2} k_1 \sum_{x_1} r_{x_1,x_2,y} = k_1\sum_{w:w_{\mathcal{X}_2\mathcal{Y}}=v^2} r_{w} \ .\]
   
   For the fourth (and symmetrically for the fifth), we have:
   \begin{equation}
     \begin{aligned}
       \sum_{w:w_{\mathcal{X}_2\mathcal{Y}}=v^2} r^2_w &= \sum_{(x_2,y) \in v^2} \sum_{x_1} r^2_{x_1,x_2,y} = \sum_{(x_2,y) \in v^2} \frac{1}{k_1} \sum_{x_1} p_{x_1,x_2} = \frac{1}{k_1} \sum_{x_2 \in v^2_{\mathcal{X}_2}} \sum_{y: (x_2,y) \in v^2} \sum_{x_1} p_{x_1,x_2}\\
       &= \frac{1}{k_1} \sum_{x_2 \in v^2_{\mathcal{X}_2}} \frac{|v^2|}{|v^2_{\mathcal{X}_2}|} \sum_{x_1} p_{x_1,x_2} \quad \text{thanks to Lemma~\ref{lem:orbitProjCard}}\\
       &= \frac{1}{k_1}\frac{|v^2|}{|v^2_{\mathcal{X}_2}|} \sum_{u:u_{\mathcal{X}_2}=v^2_{\mathcal{X}_2}} p_u \ .
     \end{aligned}
   \end{equation}

   Finally for the last constraints, we only need to compute:

   \[ \sum_{(x_1,x_2,y) \in w} p_{x_1,x_2} = \sum_{(x_1,x_2) \in w_{\mathcal{X}_1\mathcal{X}_2}} \sum_{y : (x_1,x_2,y) \in w} p_{x_1,x_2} = \sum_{(x_1,x_2) \in w_{\mathcal{X}_1\mathcal{X}_2}} \frac{|w|}{|w_{\mathcal{X}_1\mathcal{X}_2}|} p_{x_1,x_2} = \frac{|w|}{|w_{\mathcal{X}_1\mathcal{X}_2}|} p_{w_{\mathcal{X}_1\mathcal{X}_2}} \ ,\]
   which implies that the linear inequalities on $p_{x_1,x_2}, r_{x_1,x_2,y}, r^1_{x_1,x_2,y}, r^2_{x_1,x_2,y}$ get transposed respectively to the values $\frac{|w|}{|w_{\mathcal{X}_1\mathcal{X}_2}|} p_{w_{\mathcal{X}_1\mathcal{X}_2}}, r_w, r^1_w, r^2_w$. Indeed, for instance, one has for any $x_1,x_2,y$ that $p_{x_1,x_2}- r^1_{x_1,x_2,y}- r^2_{x_1,x_2,y} + r_{x_1,x_2,y} \geq 0$. Thus for some orbit $w$:
   \[\sum_{(x_1,x_2,y) \in w}\left(p_{x_1,x_2}- r^1_{x_1,x_2,y}- r^2_{x_1,x_2,y} + r_{x_1,x_2,y}\right) \geq 0 \ ,\]
   and then $ \frac{|w|}{|w_{\mathcal{X}_1\mathcal{X}_2}|} p_{w_{\mathcal{X}_1\mathcal{X}_2}}- r^1_w- r^2_w + r_w \geq 0$, which was what we wanted to show.

   Now let us consider a feasible solution  $r_w,r^1_w,r^2_w,p_u$ of the program defined in Theorem~\ref{theo:polyLP}, with a value $S^*:=\frac{1}{k_1k_2}\sum_{w} W(w)r_w$. Define:
  \begin{equation}
    \begin{aligned}
      &r_{x_1,x_2,y} := \frac{r_{(x_1,x_2,y)^G}}{|(x_1,x_2,y)^G|} \ , &r^1_{x_1,x_2,y} := \frac{r^1_{(x_1,x_2,y)^G}}{|(x_1,x_2,y)^G|} \ ,\\
      &r^2_{x_1,x_2,y} :=\frac{r^2_{(x_1,x_2,y)^G}}{|(x_1,x_2,y)^G|} \ , &p_{x_1,x_2} := \frac{p_{(x_1,x_2)^G}}{|(x_1,x_2)^G|} \ .
    \end{aligned}
  \end{equation}

  Let us show that $r_{x_1,x_2,y},r^1_{x_1,x_2,y},r^2_{x_1,x_2,y},p_{x_1,x_2}$ is a feasible solution of the program defined in Proposition~\ref{prop:NSLP}, and that its value $S:=\frac{1}{k_1k_2} \sum_{x_1,x_2,y} W(y|x_1x_2)r_{x_1,x_2,y} = S^*$.

  First we have $S = S^*$. Indeed:
   \begin{equation}
     \begin{aligned}
       S &= \frac{1}{k_1k_2}\sum_{x_1,x_2,y} W(y|x_1x_2) r_{x_1,x_2,y} = \frac{1}{k_1k_2}\sum_{x_1,x_2,y} W(y|x_1x_2) \frac{r_{(x_1,x_2,y)^G}}{|r_{(x_1,x_2,y)^G}|}\\
       &= \frac{1}{k_1k_2}\sum_{w} \sum_{(x_1,x_2,y) \in w} W(y|x_1x_2) \frac{r_w}{|w|} = \frac{1}{k_1k_2}\sum_{w} \sum_{(x_1,x_2,y) \in w} W(w) \frac{r_w}{|w|}\\
       &= \frac{1}{k_1k_2}\sum_{w} |w| W(w) \frac{r_w}{|w|} = \frac{1}{k_1k_2}\sum_{w} W(w)r_w = S^* \ .
     \end{aligned}
  \end{equation}

   Then, all the constraints are satisfied. Indeed, thanks to Lemma~\ref{lem:fromOrbitToEle}, we have for the first constraint:
   
   \[\sum_{x_1,x_2} r_{x_1,x_2,y} = \frac{1}{|y^G|} \sum_{w:w_{\mathcal{Y}}=y^G} r_w = \frac{|y^G|}{|y^G|} = 1 \ .\]
   For the second constraint (and symmetrically for the third constraint), we have:

   \[\sum_{x_1} r^1_{x_1,x_2,y} = \frac{1}{|(x_2,y)^G|} \sum_{w:w_{\mathcal{X}_2\mathcal{Y}}=(x_2,y)^G} r^1_w = \frac{k_1}{|(x_2,y)^G|}\sum_{w:w_{\mathcal{X}_2\mathcal{Y}}=(x_2,y)^G} r_w = k_1 \sum_{x_1} r_{x_1,x_2,y} \ .\]

   For the fourth (and symmetrically for the fifth), we have:
   \begin{equation}
     \begin{aligned}
       \sum_{x_1} r^2_{x_1,x_2,y} &= \frac{1}{|(x_2,y)^G|} \sum_{w:w_{\mathcal{X}_2\mathcal{Y}}=(x_2,y)^G} r^2_w = \frac{1}{|(x_2,y)^G|} \frac{1}{k_1}\frac{|(x_2,y)^G|}{|(x_2,y)^G_{\mathcal{X}_2}|} \sum_{u:u_{\mathcal{X}_2}=(x_2,y)^G_{\mathcal{X}_2}} p_u \\
       &= \frac{1}{k_1}\frac{1}{|(x_2,y)^G_{\mathcal{X}_2}|} \sum_{u:u_{\mathcal{X}_2}=(x_2,y)^G_{\mathcal{X}_2}} p_u = \frac{1}{k_1} \frac{1}{|x_2^G|}\sum_{u:u_{\mathcal{X}_2}=x_2^G} p_u \text{ since $(x_2,y)^G_{\mathcal{X}_2} = x_2^{G}$}\\
       &= \frac{1}{k_1} \sum_{x_1} p_{x_1,x_2} \ .
     \end{aligned}
   \end{equation}
   
   Finally, to conclude with the last constraints, one has only to see that for any $x_1,x_2,y$:
   \[\frac{|(x_1,x_2,y)^G|}{|(x_1,x_2,y)^G_{\mathcal{X}_1\mathcal{X}_2}|}p_{(x_1,x_2,y)^G_{\mathcal{X}_1\mathcal{X}_2}} = \frac{|(x_1,x_2,y)^G|}{|(x_1,x_2)^G|}p_{(x_1,x_2)^G} = |(x_1,x_2,y)^G|p_{x_1,x_2}  \ ,\]
   which implies that the linear inequalities on $\frac{|w|}{|w_{\mathcal{X}_1\mathcal{X}_2}|} p_{w_{\mathcal{X}_1\mathcal{X}_2}}, r_w, r^1_w, r^2_w$ get transposed respectively to the values $p_{x_1,x_2}, r_{x_1,x_2,y}, r^1_{x_1,x_2,y}, r^2_{x_1,x_2,y}$. Indeed, for instance, one has for any $w$ that $\frac{|w|}{|w_{\mathcal{X}_1\mathcal{X}_2}|} p_{w_{\mathcal{X}_1\mathcal{X}_2}}- r^1_w - r^2_w + r_w \geq 0$. But for any $(x_1,x_2,y) \in w$, one has that $r_{x_1,x_2,y} = \frac{r_w}{|w|},r^1_{x_1,x_2,y} = \frac{r^1_w}{|w|},r^2_{x_1,x_2,y} = \frac{r^2_w}{|w|}$. Thanks to the previous inequality, we have that $p_{x_1,x_2}  = \frac{p_{w_{\mathcal{X}_1\mathcal{X}_2}}}{|w_{\mathcal{X}_1\mathcal{X}_2}|}$, and thus:
   \[p_{x_1,x_2}- r^1_{x_1,x_2,y}- r^2_{x_1,x_2,y} + r_{x_1,x_2,y} = \frac{p_{w_{\mathcal{X}_1\mathcal{X}_2}}}{|w_{\mathcal{X}_1\mathcal{X}_2}|} - \frac{r^1_w}{|w|} - \frac{r^2_w}{|w|} + \frac{r_w}{|w|} \geq 0 \ ,\]
   which was what we wanted to show.
\end{proof}
  
\section{Non-Signaling Achievability Bounds}
\label{section:IB}
\subsection{Zero-Error Non-Signaling Assisted Achievable Rate Pairs}
We will now present a numerical method to find efficiently inner bounds on $\mathcal{C}_0^{\mathrm{NS}}(W)$. Thanks to Corollary~\ref{cor:poly}, we know how to decide in polynomial time in $n,k_1,k_2$ whether $\mathrm{S}^{\mathrm{NS}}(W^{\otimes n},k_1,k_2)=1$. However, by Proposition~\ref{prop:ZENScapacity}, if $\mathrm{S}^{\mathrm{NS}}(W^{\otimes n},k_1,k_2)=1$, then we have that $\left(\frac{\log(k_1)}{n},\frac{\log(k_2)}{n}\right) \in \mathcal{C}_0^{\mathrm{NS}}(W)$, which describes a way of computing achievable points for that capacity region. More precisely, this leads to the following result:

\begin{prop}[Inner Bounds on $\mathcal{C}_0^{\mathrm{NS}}(W)$]
  \label{prop:IB}
  Let us define the zero-error non-signaling assisted $n$-shots capacity region $\mathcal{C}_{0,\leq n}^{\mathrm{NS}}(W)$ in the following way:
  \[ \mathcal{C}_{0,\leq n}^{\mathrm{NS}}(W)  := \left\{ \left(\frac{\log(k_1)}{n},\frac{\log(k_2)}{n}\right) : \mathrm{S}^{\mathrm{NS}}(W^{\otimes n},k_1,k_2)=1 \right\}\ . \]
  Then, we have that $\forall n \in \mathbb{N}, \mathcal{C}_{0,\leq n}^{\mathrm{NS}}(W) \subseteq \mathcal{C}_0^{\mathrm{NS}}(W)$, and that one can decide in polynomial time in $n,k_1,k_2$ if $\left(\frac{\log(k_1)}{n},\frac{\log(k_2)}{n}\right) \in \mathcal{C}_{0,\leq n}^{\mathrm{NS}}(W)$.
\end{prop}

This implies that we can find efficiently achievable rate pairs for MACs.

\paragraph{Application to the binary adder channel}
The binary adder channel $W_{\text{BAC}}$ is the following MAC:
\[ \forall x_1,x_2 \in \{0,1\}, \forall y \in \{0,1,2\}, W_{\text{BAC}}(y|x_1x_2) := \delta_{y,x_1+x_2} \ .\]

Its classical capacity region $\mathcal{C}(W_{\text{BAC}})$ is well known and consists of all $(R_1,R_2)$ such that $R_1 \leq 1,R_2 \leq 1, R_1+R_2 \leq \frac{3}{2}$, as a consequence of Theorem~\ref{theo:capacity}. Its zero-error classical capacity $\mathcal{C}_0(W_{\text{BAC}})$ is not yet characterized. A lot of work has been done in finding outer and inner bounds on this region~\cite{Lindstrom69,Tilborg78,KL78,Weldon78,KLWY83,BT85,BB98,UL98,AB99,MO05,OS15}. To date, the best lower bound on the sum-rate capacity is $\log_2(240/6) \simeq 1.3178$~\cite{MO05}.

Thanks to Proposition~\ref{prop:IB}, we were able to compute the regions $\mathcal{C}_{0,\leq n}^{\mathrm{NS}}(W)$ for $n$ going up to $7$, which led to Figure~\ref{fig:BAC}. The code can be found on \href{https://github.com/pferme/MAC_NS_LP}{GitHub}. It uses Mosek linear programming solver~\cite{mosek}.

Note that the linear program from Theorem~\ref{theo:polyLP} has still a large number of variables and constraints although polynomial in $n$. Specifically, for $n=2$, it has $244$ variables and $480$ constraints; for $n=3$, it has $1112$ variables and $2054$ constraints; for $n=7$, it has $95592$ variables and $162324$ constraints; finally, for $n=8$, it has $226911$ variables and $383103$ constraints.

  \begin{figure}[!h]
    \begin{center}
      \begin{tikzpicture}
        \begin{axis}[
          xmin = 0, xmax = 1.05,
          ymin = 0, ymax = 1.05,
          xtick distance = 0.25,
          ytick distance = 0.25,
          grid = both,
          width = 0.55\textwidth,
          height = 0.55\textwidth,
          legend cell align = {left},
          legend pos = south west,
          ylabel=$R_2$,
          xlabel=$R_1$,
        ]
        \addplot[
          thick,
          dashed,
          black,
        ] coordinates {(0,1) (0.5,1) (1,0.5) (1,0)};
        \addplot[
          densely dashed,
          darkgray,
          mark = star,
        ] coordinates {(0.33, 0.962409352487) (0.43425,0.86873) (0.59749375012,0.720321349148) (0.720321349148,0.59749375012) (0.86873,0.43425) (0.962409352487,0.33)};
        \addplot[
          domain = 0:1,
          brown,
          mark = triangle,
        ] table[x=R1,y=R2,col sep=comma] {Data_BAC/ZEBorder2.csv};
        \addplot[
          domain = 0:1,
          red,
          mark = o,
        ] table[x=R1,y=R2,col sep=comma] {Data_BAC/ZEBorder3.csv};
        \addplot[
          domain = 0:1,
          blue,
          mark = x,
        ] table[x=R1,y=R2,col sep=comma] {Data_BAC/ZEBorder7.csv};
        \addplot[
          domain = 0.25:0.33,
          darkgray,
          densely dashed,
        ] {(1 + -(2*\x)*ln(2*\x)/ln(2))-(1-2*\x)*ln(1-2*\x)/ln(2))/2};
        \draw[domain=0.25:0.33,variable=\y,darkgray ,densely dashed] plot({(1 + -(2*\y)*ln(2*\y)/ln(2))-(1-2*\y)*ln(1-2*\y)/ln(2))/2},\y);
        \addplot[
          darkgray,
          densely dashed,
        ]  coordinates {(0,1) (0.25,1)};
        \addplot[
          darkgray,
          densely dashed,
        ]  coordinates {(1,0) (1,0.25)};  
        \legend{$\mathcal{C}(W_{\text{BAC}})$, Best Inner Bounds on $\mathcal{C}_0(W_{\text{BAC}})$, $\mathcal{C}_{0,\leq 2}^{\mathrm{NS}}(W_{\text{BAC}})$, $\mathcal{C}_{0,\leq 3}^{\mathrm{NS}}(W_{\text{BAC}})$, $\mathcal{C}_{0,\leq 7}^{\mathrm{NS}}(W_{\text{BAC}})$}
        \end{axis}
      \end{tikzpicture}
      \caption{Capacity regions of the binary adder channel $W_{\mathrm{BAC}}$. The black dashed curve depicts the classical capacity region $\mathcal{C}(W_{\text{BAC}})$, whereas the grey dashed curve shows the best known inner bound border on the zero-error classical capacity region $\mathcal{C}_0(W_{\text{BAC}})$, made from results by~\cite{MO05,BT85,KLWY83}; see~\cite{MO05} for a description of this border. On the other hand, the continuous curves depict the best zero-error non-signaling assisted achievable rate pairs for respectively $2,3$ and $7$ copies of the binary adder channel.}
    \label{fig:BAC}
    \end{center}
\end{figure}
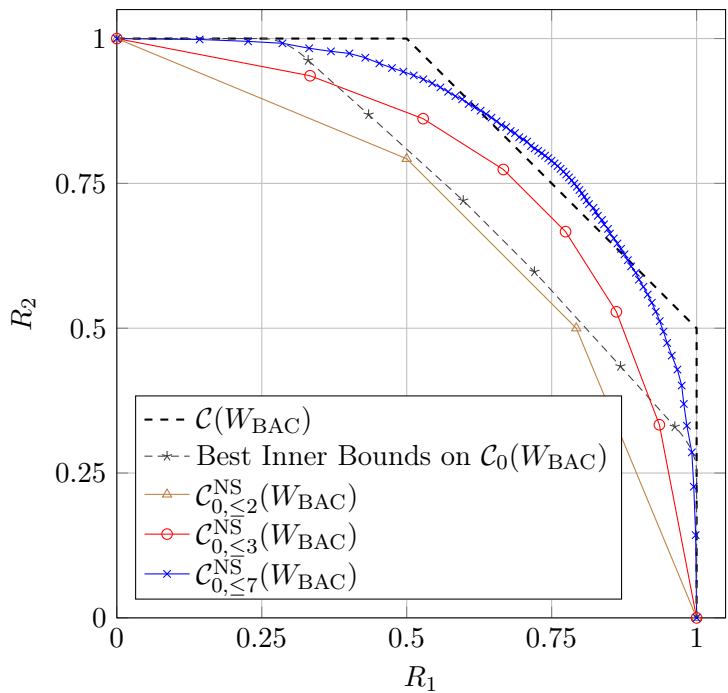

  The first noticeable result coming from these curves is that the zero-error non-signaling assisted sum-rate capacity beats with only $7$ copies the classical sum-rate capacity of $\frac{3}{2}$, even without a zero-error constraint, with a value of $\frac{2\log_2(42)}{7} \simeq 1.5406$, coming from the fact that $\mathrm{S}^{\mathrm{NS}}(W^{\otimes 7}_{\mathrm{BAC}},42,42)=1$ and Proposition~\ref{prop:ZENScapacity}. This implies that $\mathcal{C}^{\mathrm{NS}}_0(W_{\text{BAC}})$ has larger sum-rate pairs than $\mathcal{C}(W_{\text{BAC}})$, and that $\mathcal{C}^{\mathrm{NS}}(W_{\text{BAC}})$ is strictly larger than $\mathcal{C}(W_{\text{BAC}})$. This sum-rate can even be increased up to $\frac{\log_2(72)}{4} \simeq 1.5425$, since we have computed $\mathrm{S}^{\mathrm{NS}}(W^{\otimes 8}_{\mathrm{BAC}},72,72)=1$, which is the largest number of copies we have been able to manage with our efficient version of the linear program from Theorem~\ref{theo:polyLP}. This should be compared with the upper bound on the non-signaling assisted sum-rate capacity coming from Proposition~\ref{prop:BACcapacityNSrelaxed}, which is $\log_2(3) \simeq 1.5850$ for $R_1=R_2$.

  Another surprising property is the speed at which one obtains efficient zero-error non-signaling assisted codes compared to classical zero-error codes. Indeed, with only three copies of the binary adder channel, one gets that $\mathrm{S}^{\mathrm{NS}}(W^{\otimes 3}_{\mathrm{BAC}},4,5)=1$, which corresponds to a sum-rate of $\frac{2+\log_2(5)}{3} \simeq 1.4406$, which already largely beats the best known zero-error achieved sum-rate of $\log_2(240/6) \simeq 1.3178$~\cite{MO05}. These results are summarized in the following theorem:

  \begin{theo}
    \label{theo:BAC}
    We have that $\left(\frac{\log_2(72)}{8},\frac{\log_2(72)}{8}\right) \in \mathcal{C}^{\mathrm{NS}}_0(W_{\mathrm{BAC}})$ but $\left(\frac{\log_2(72)}{8},\frac{\log_2(72)}{8}\right) \not\in \mathcal{C}(W_{\mathrm{BAC}})$, and as a consequence, we have that $\mathcal{C}(W_{\mathrm{BAC}}) \subsetneq \mathcal{C}^{\mathrm{NS}}(W_{\mathrm{BAC}})$.
  \end{theo}
  \begin{proof}
    Since $2^{8\frac{\log_2(72)}{8}}=72$ and numerically $\mathrm{S}^{\mathrm{NS}}(W^{\otimes 8}_{\mathrm{BAC}},72,72)=1$ thanks to Corollary~\ref{cor:poly}, we get that  $\left(\frac{\log_2(72)}{8},\frac{\log_2(72)}{8}\right) \in \mathcal{C}^{\mathrm{NS}}_0(W_{\mathrm{BAC}})$ by Proposition~\ref{prop:ZENScapacity}. However, $\frac{\log_2(72)}{8}+\frac{\log_2(72)}{8} > \frac{3}{2}$ so $\left(\frac{\log_2(72)}{8},\frac{\log_2(72)}{8}\right) \not\in \mathcal{C}(W_{\mathrm{BAC}})$ by Theorem~\ref{theo:capacity} applied to $W_{\mathrm{BAC}}$. Since $\mathcal{C}(W_{\mathrm{BAC}}) \subseteq \mathcal{C}^{\mathrm{NS}}(W_{\mathrm{BAC}})$ and $\mathcal{C}^{\mathrm{NS}}_0(W_{\mathrm{BAC}}) \subseteq \mathcal{C}^{\mathrm{NS}}(W_{\mathrm{BAC}})$, we thus get that $\mathcal{C}(W_{\mathrm{BAC}}) \subsetneq \mathcal{C}^{\mathrm{NS}}(W_{\mathrm{BAC}})$.
  \end{proof}

  \subsection{Non-Signaling Assisted Achievable Rate Pairs with Non-Zero Error}
  We have analyzed the non-signaling assisted capacity region through zero-error strategies and applied it to the BAC. However, if some noise is added to that channel, its zero-error non-signaling assisted capacity region becomes trivial (see Proposition~\ref{prop:NScapaNoisy}). Thus, the previous method fails to find significant inner bounds on the non-signaling assisted capacity region of noisy MACs.

  In this section, we use concatenated codes to obtain achievable rate pairs, and apply it to a noisy version of the BAC:

  \begin{defi}[Concatenated Codes]
  Given a MAC $W$ and a non-signaling assisted code $P$, define $W[P] : [k_1] \times [k_2] \rightarrow [\ell]$ with $W[P](j|i_1i_2) := \sum_{x_1,x_2,y}W(y|x_1x_2)P(x_1x_2j|i_1i_2y)$:
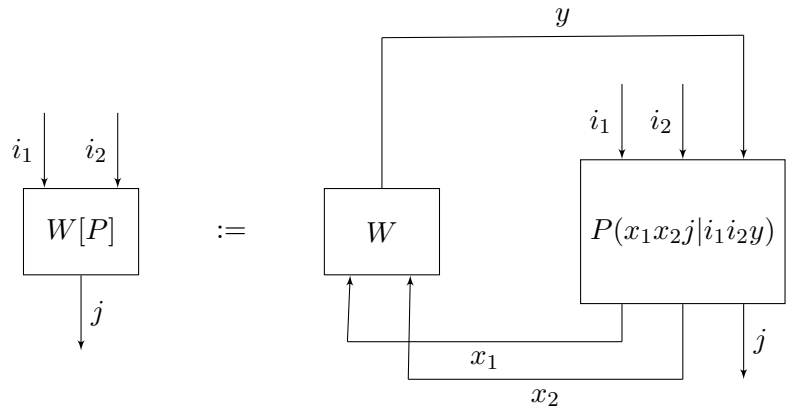
\begin{center}
  \begin{tikzpicture}[auto, node distance=2cm,>=latex']
    \node [input, name=i1] {};
    \node [input, name=i2] {};
    \node [Bigblock, below of=i2] (P) {$P(x_1x_2j|i_1i_2y)$};

    \draw [<-] (P.130) -- node {$i_1$} +(0pt,1cm);
    \draw [<-] (P.90) -- node {$i_2$} +(0pt,1cm);
    \coordinate (ybis) at ($ (P.50) + (0pt,1.65cm) $);
    \draw [<-] (P.50) -- (ybis);
    \coordinate (x1) at ($ (P.230)+(0pt,-0.5cm) $);
    \draw [-] (P.230) -- (x1);
    \coordinate (x2) at ($ (P.270)+(0pt,-1cm) $);
    \draw [-] (P.270) -- (x2);
    \draw [->] (P.310) -- node {$j$} +(0pt,-1cm);

    \node [left of=P] (A) {};
    \node [bigblock, left of=A] (W) {$W$};
    \coordinate (x1bis) at ($ (x1)+(-3.65cm,0pt) $);
    \coordinate (x2bis) at ($ (x2)+(-3.65cm,0pt) $);
    \draw [-] (x1) -- node {$x_1$} (x1bis);
    \draw [-] (x2) -- node {$x_2$} (x2bis);
    \draw [->] (x1bis) -- (W.234);
    \draw [->] (x2bis) -- (W.303);
    \coordinate (y) at ($ (W.north)+(0pt,2cm) $) ;
    \draw (W.north) -- (y);
    \draw (y) -- node {$y$} (ybis);

    \node [left of=W] (Equal) {$:=$};
    \node [bigblock, left of=Equal] (WP) {$W[P]$};

    \draw [<-] (WP.130) -- node {$i_1$} +(0pt,1cm);
    \draw [<-] (WP.50) -- node {$i_2$} +(0pt,1cm);
    \draw [->] (WP.270) -- node {$j$} +(0pt,-1cm);
    
  \end{tikzpicture}
\end{center}
\end{defi}

 Note that $W[P]$ is a MAC since $W[P](j|i_1i_2) \geq 0$ and:
  \begin{equation}
    \begin{aligned}
      \sum_j W[P](j|i_1i_2) &= \sum_{x_1,x_2,y}W(y|x_1x_2) \sum_j P(x_1x_2j|i_1i_2y)\\
      &= \sum_{x_1,x_2} \left(\sum_y W(y|x_1x_2)\right)P(x_1x_2|i_1i_2) \text{ since $P$ is non-signaling}\\
      &=  \sum_{x_1,x_2} P(x_1x_2|i_1i_2) = 1 \ .
    \end{aligned}
  \end{equation}
  
  The following proposition states that combining a classical code to a non-signaling strategy leads to inner bounds on the non-signaling assisted capacity region of a MAC:
  
  \begin{prop}
    \label{prop:concatCodes}
    If $P$ is a non-signaling assisted code for the MAC $W$, we have that $\mathcal{C}(W[P]) \subseteq \mathcal{C}^{\mathrm{NS}}(W)$.
  \end{prop}

  \begin{proof}
    Let $(R_1,R_2) \in \mathcal{C}(W[P])$. Then, by definition, we have that:
    \[ \underset{n \rightarrow +\infty}{\lim} \mathrm{S}(W[P]^{\otimes n},\ceil{2^{R_1n}},\ceil{2^{R_2n}}) = 1 \ . \]

    Let us fix $\varepsilon > 0$. For a large enough $N$, we have $\mathrm{S}(W[P]^{\otimes N},\ceil{2^{R_1N}},\ceil{2^{R_2N}}) \geq 1 - \varepsilon$. Let us call $\ell_1 := \ceil{2^{R_1N}}$ and $\ell_2 := \ceil{2^{R_2N}}$. Thus, there exists encoders $e_1:[\ell_1] \rightarrow [k_1],e_2:[\ell_2] \rightarrow [k_2]$ and a decoder $d:[\ell] \rightarrow [\ell_1] \times [\ell_2]$ such that:
    \[ \frac{1}{\ell_1\ell_2} \sum_{i_1,i_2,j} W[P](j|i_1i_2)\sum_{a_1 \in [\ell_1],a_2 \in [\ell_2]}e_1(i_1|a_1)e_2(i_2|a_2)d(a_1a_2|j) \geq 1 - \varepsilon \ . \]

    In particular, we have:
    \[ \frac{1}{\ell_1\ell_2} \sum_{x_1,x_2,y} W(y|x_1x_2) \left[\sum_{i_1,i_2,j,a_1,a_2}P(x_1x_2j|i_1i_2y)e_1(i_1|a_1)e_2(i_2|a_2)d(a_1a_2|j)\right] \geq 1 - \varepsilon \ . \]

    Let us define $\hat{P}(x_1x_2(b_1b_2)|a_1a_2y) := \sum_{i_1,i_2,j}P(x_1x_2j|i_1i_2y)e_1(i_1|a_1)e_2(i_2|a_2)d(b_1b_2|j)$. Then, one can easily check that $\hat{P}$ is non-signaling, and thus:

    \[ \mathrm{S}^{\mathrm{NS}}(W^{\otimes N},\ell_1,\ell_2) \geq \frac{1}{\ell_1\ell_2} \sum_{x_1,x_2,y} W(y|x_1x_2)\sum_{a_1,a_2}\hat{P}(x_1x_2(a_1,a_2)|a_1a_2y) \geq 1 - \varepsilon \ .  \]

    This implies that $\underset{n \rightarrow +\infty}{\lim} \mathrm{S}^{\mathrm{NS}}(W^{\otimes n},\ceil{2^{R_1n}},\ceil{2^{R_2n}}) = 1$, i.e. $(R_1,R_2) \in \mathcal{C}^{\mathrm{NS}}(W)$.
  \end{proof}

  Thanks to Proposition~\ref{prop:concatCodes}, we have for any non-signaling assisted code $P$, $\mathcal{C}(W^{\otimes n}[P]) \subseteq \mathcal{C}^{\mathrm{NS}}(W^{\otimes n})$. But if $(R_1,R_2) \in \mathcal{C}^{\mathrm{NS}}(W^{\otimes n})$, we have that $(\frac{R_1}{n},\frac{R_2}{n}) \in \mathcal{C}^{\mathrm{NS}}(W)$. Thus, applying Theorem~\ref{theo:capacity} on $W^{\otimes n}[P]$ leads to inner bounds on $\mathcal{C}^{\mathrm{NS}}(W)$:
  \begin{prop}[Inner Bounds on $\mathcal{C}^{\mathrm{NS}}(W)$]
    \label{prop:NumericalMethod} For any number of copies $n$, number of inputs $k_1 \in [|\mathcal{X}_1|^n]$ and $k_2 \in [|\mathcal{X}_2|^n]$, non-signaling assisted codes $P$ on inputs in $[k_1],[k_2]$ for $W^{\otimes n}$, and distributions $q_1,q_2$ on $[k_1],[k_2]$, we have that the following $(R_1,R_2)$ are in $\mathcal{C}^{\mathrm{NS}}(W)$: 
  \[ R_1 \leq \frac{I(I_1:J|I_2)}{n}\ ,\ R_2 \leq \frac{I(I_2:J|I_1)}{n}\ ,\ R_1+R_2 \leq \frac{I((I_1,I_2):J)}{n} \ ,\]
  for $(I_1,I_2) \in [k_1] \times [k_2]$ following the product law $q_1 \times q_2$, and $J \in [\ell]$ the outcome of $W^{\otimes n }[P]$ on inputs $I_1,I_2$. In particular, the corner points of this capacity region are given by:
  \[ \left(\frac{I(I_1:J|I_2)}{n},\frac{I(I_2 : J)}{n} \right) \text{ and } \left(\frac{I(I_1 : J)}{n} ,\frac{I(I_2:J|I_1)}{n}\right) \ .\]
  \end{prop}

  \begin{proof}
    The achievable region comes from the previous discussion. We just need to prove that the corner points are of the given form. If $R_1=\frac{I(I_1:J|I_2)}{n}$, constraints on $R_2$ and $R_1+R_2$ leads to a maximum $R_2 = \min\left(\frac{I(I_2:J|I_1)}{n}, \frac{I((I_1,I_2):J)}{n}-\frac{I(I_1:J|I_2)}{n}\right)$. However, $I((I_1,I_2):J)-I(I_1:J|I_2)=I(I_2 : J)$ by the chain rule. We only need to show that $I(I_2:J) \leq I(I_2:J|I_1)$ and the proof will be complete, since the other corner point is symmetric. We have:
    \[ I(I_2:J) = H(I_2)-H(I_2|J) = H(I_2|I_1) - H(I_2|J) \leq H(I_2|I_1) - H(I_2|JI_1) = I(I_2:J|I_1) \ ,\]
    the second equality coming from the fact that $I_1$ and $I_2$ are independent, and the inequality coming from the fact that $H(A|BC) \leq H(A|B)$ for any $A,B,C$.
  \end{proof}
  
  \paragraph{Application to the Noisy Binary Adder Channel} We will now apply this strategy to a noisy version of the BAC. We will consider flip errors $\varepsilon_1, \varepsilon_2$ of inputs $x_1,x_2$ on $W_{\text{BAC}}$, which leads to the following definition of $W_{\text{BAC},\varepsilon_1, \varepsilon_2}$:
  \begin{equation}
    \begin{aligned}
      \forall y,x_1,x_2, W_{\text{BAC},\varepsilon_1, \varepsilon_2}(y|x_1x_2) &:=  (1-\varepsilon_1)(1-\varepsilon_2) W_{\text{BAC}}(y|x_1x_2)\\
      &+ \varepsilon_1(1-\varepsilon_2) W_{\text{BAC}}(y|\overline{x_1}x_2)\\
      &+ (1-\varepsilon_1)\varepsilon_2 W_{\text{BAC}}(y|x_1\overline{x_2})\\
      &+ \varepsilon_1\varepsilon_2 W_{\text{BAC}}(y|\overline{x_1}\overline{x_2}) \ .
      \end{aligned}
  \end{equation}

  First, let us note that the zero-error non-signaling assisted capacity region of $W_{\text{BAC},\varepsilon_1, \varepsilon_2}$ is trivial for $\varepsilon \in (0,1)$:

  \begin{prop}
    \label{prop:NScapaNoisy}
    If $\varepsilon_1,\varepsilon_2 \in (0,1)$, then $\mathcal{C}_0^{\mathrm{NS}}(W_{\text{BAC},\varepsilon_1, \varepsilon_2}) = \left\{ (0,0) \right\}$.
  \end{prop}
  \begin{proof}
    If $\mathrm{S}^{\mathrm{NS}}(W^{\otimes n},k_1,k_2) = 1$, then $\forall y^n,x_1^n,x_2^n : W^{\otimes n}(y^n|x_1^nx_2^n) > 0 \implies r_{x_1^n,x_2^n,y^n} = p_{x_1^n,x_2^n}$. Indeed, we have for an optimal $p,r$ that:
    \[ \mathrm{S}^{\mathrm{NS}}(W^{\otimes n},k_1,k_2) = \frac{1}{k_1k_2} \sum_{x_1^n,x_2^n,y^n} W^{\otimes n}(y^n|x_1^nx_2^n)r_{x_1^n,x_2^n,y^n} \leq \frac{1}{k_1k_2} \sum_{x_1^n,x_2^n,y^n} W^{\otimes n}(y^n|x_1^nx_2^n)p_{x_1^n,x_2^n} = 1 \ , \]
    which implies the previous statement. But, for $W_{\text{BAC},\varepsilon_1, \varepsilon_2}^{\otimes n}$, one can easily check that for all $y^n,x_1^n,x_2^n$, $W^{\otimes n}(y^n|x_1^nx_2^n) > 0$ since $\varepsilon_1,\varepsilon_2 \in (0,1)$. Indeed, you just have to flip the inputs to a valid preimage of the output. Thus if $\mathrm{S}^{NS}(W_{\text{BAC},\varepsilon_1, \varepsilon_2}^{\otimes n},k_1,k_2) = 1$, we have that $\forall y^n,x_1^n,x_2^n, r_{x_1^n,x_2^n,y^n} = p_{x_1^n,x_2^n}$. In particular, this implies that $\sum_{x_1^n,x_2^n} r_{x_1^n,x_2^n,y^n} = \sum_{x_1^n,x_2^n} p_{x_1^n,x_2^n}$, therefore $1 = k_1k_2$, so $k_1 = 1$ and $k_2 = 1$. Thus $\mathrm{S}^{NS}(W^{\otimes n},2^{nR_1},2^{nR_2}) = 1$ implies that $(R_1,R_2)=(0,0)$. 
  \end{proof}

  We have then applied the numerical method described in Proposition~\ref{prop:NumericalMethod} to $W_{\text{BAC},\varepsilon_1, \varepsilon_2}$ for the symmetric case $\varepsilon_1 = \varepsilon_2 = \varepsilon := 10^{-3}$. Since it is hard to go through all non-signaling assisted codes $P$ and product distributions $q_1,q_2$, we have applied the heuristic of using non-signaling assisted codes obtained while optimizing $\mathrm{S}^{\mathrm{NS}}(W^{\otimes n},k_1,k_2)$ in the symmetrized linear program. We have combined them with uniform $q_1,q_2$, as the form of those non-signaling assisted codes coming from our optimization program is symmetric. We have evaluated the achievable corner points for all $k_1, k_2 \leq 2^n$ for $n \leq 5$ copies which led to Figure~\ref{fig:noisyBAC}:
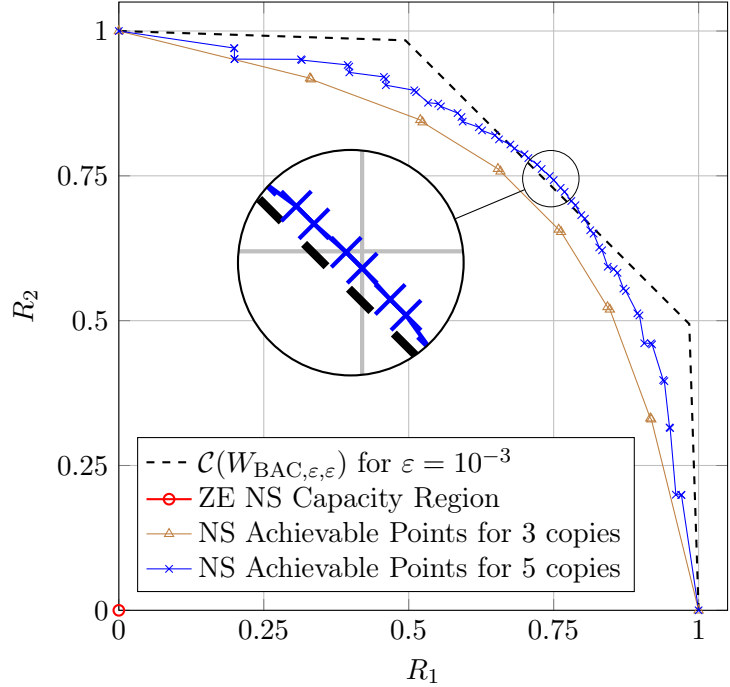
\begin{figure}[!h]
    \begin{center}
      \begin{tikzpicture}
        [spy using outlines={circle, magnification=4, size=2cm, connect spies}]
        \begin{axis}[
          xmin = 0, xmax = 1.05,
          ymin = 0, ymax = 1.05,
          xtick distance = 0.25,
          ytick distance = 0.25,
          grid = both,
          width = 0.55\textwidth,
          height = 0.55\textwidth,
          legend cell align = {left},
          legend pos = south west,
          ylabel=$R_2$,
          xlabel=$R_1$,
        ]
        \addplot[
          thick,
          dashed,
          black,
        ] coordinates {(0,1) (0.49429612113126953,0.983898049082558) (0.983898049082558, 0.49429612113126953) (1,0)};
        \addplot[
          red,
          mark = o,
          thick,
        ] coordinates {(0,0)};
        \addplot[
          domain = 0:1,
          brown,
          mark = triangle,
        ] table[x=R1,y=R2,col sep=comma] {Data_Capacity_BACwithSmallBitFlipNoise/Capacity3.csv};
        \addplot[
          domain = 0:1,
          blue,
          mark = x,
        ] table[x=R1,y=R2,col sep=comma] {Data_Capacity_BACwithSmallBitFlipNoise/Capacity5.csv};
        \legend{$\mathcal{C}(W_{\text{BAC},\varepsilon,\varepsilon})$ for $\varepsilon = 10^{-3}$, ZE NS Capacity Region, NS Achievable Points for $3$ copies, NS Achievable Points for $5$ copies}
        
        \coordinate (spypoint) at (0.745,0.745);
        \coordinate (magnifyglass) at (0.4,0.6);
        \end{axis}
        \spy [size=3cm] on (spypoint) in node [fill=white] at (magnifyglass);
      \end{tikzpicture}
      \caption{Capacity regions of the noisy binary adder channel $W_{\text{BAC},\varepsilon,\varepsilon}$ for $\varepsilon = 10^{-3}$. The black dashed curve depicts the classical capacity region $\mathcal{C}(W_{\text{BAC},\varepsilon,\varepsilon})$ which was found numerically using Theorem~\ref{theo:capacity}. The red point depicts the zero-error non-signaling assisted capacity region (Proposition~\ref{prop:NScapaNoisy}). The blue curve depicts achievable non-signaling assisted rates pairs obtained from $\mathcal{C}(W_{\text{BAC},\varepsilon,\varepsilon}^{\otimes 5}[P])$ through the numerical method described in Proposition~\ref{prop:NumericalMethod}.}
    \label{fig:noisyBAC}
    \end{center}
\end{figure}

Compared to the noiseless binary adder channel, we can first notice that the classical capacity region is slightly smaller, with a classical  sum-rate capacity of $1.478$ at most. On the other hand, although the zero-error non-signaling assisted capacity of $W_{\text{BAC},\varepsilon,\varepsilon}$ is completely trivial, we have with our concatenated codes strategy found significant rate pairs achievable with non-signaling assistance. In particular, we have reached a non-signaling assisted sum-rate capacity of $1.493$ which beats the best classical sum-rate capacity. Thus, it shows that non-signaling assistance can improve the capacity of the noisy binary adder channel as well.

\section{Relaxed Non-Signaling Assisted Capacity Region and Outer Bounds}
\label{section:OB}
A natural question that arises when studying the strength of non-signaling assistance is whether a result similar to Theorem~\ref{theo:capacity} can be found to describe by a single-letter formula the non-signaling assisted capacity region of MACs. In particular, dropping the constraint that $(X_1,X_2)$ is in product form in Theorem~\ref{theo:capacity} seems to be a particularly good candidate to characterize the non-signaling assisted capacity region of MACs, as this looks quite similar to allowing correlations between parties.

We have not been able to show the equivalence between this region and the non-signaling assisted capacity region; however, it turns out to be equivalent to the capacity region defined by a slight relaxation of non-signaling assistance, which we call $\mathrm{S}^{\overline{\mathrm{NS}}}(W,k_1,k_2)$. In particular, this will give us the best known outer bound on the non-signaling capacity.

  \begin{defi}
    \begin{equation}
      \begin{aligned}
        \mathrm{S}^{\overline{\mathrm{NS}}}(W,k_1,k_2) := &&\underset{r,p}{\maxi} &&& \frac{1}{k_1k_2} \sum_{x_1,x_2,y} W(y|x_1x_2)r_{x_1,x_2,y}\\
        &&\st &&& \sum_{x_1,x_2} r_{x_1,x_2,y} \leq 1\\
        &&&&& \sum_{x_1,x_2} p_{x_1,x_2} = k_1k_2\\
        &&&&& \sum_{x_1} p_{x_1,x_2} \geq k_1 \sum_{x_1} r_{x_1,x_2,y}\\
        &&&&& \sum_{x_2} p_{x_1,x_2} \geq k_2 \sum_{x_2} r_{x_1,x_2,y}\\
        &&&&& 0 \leq r_{x_1,x_2,y} \leq p_{x_1,x_2}\\
      \end{aligned}
    \end{equation}
  \end{defi}

The following proposition shows that this is indeed a relaxation of the non-signaling constraint.

  \begin{prop}
    \label{prop:NSisrelaxed}
   $\mathrm{S}^{\mathrm{NS}}(W,k_1,k_2) \leq \mathrm{S}^{\overline{\mathrm{NS}}}(W,k_1,k_2)$.
  \end{prop}
  \begin{proof}
    Let us take a solution $(p_{x_1,x_2}, r_{x_1,x_2,y}, r^1_{x_1,x_2,y}, r^2_{x_1,x_2,y})_{x_1 \in \mathcal{X}_1,x_2 \in \mathcal{X}_2,y \in \mathcal{Y}}$ of the linear program computing $\mathrm{S}^{\mathrm{NS}}(W,k_1,k_2)$. Let us show that $(p_{x_1,x_2}, r_{x_1,x_2,y})_{x_1 \in \mathcal{X}_1,x_2 \in \mathcal{X}_2,y \in \mathcal{Y}}$ is a solution of the linear program computing $\mathrm{S}^{\overline{\mathrm{NS}}}(W,k_1,k_2)$ with a same objective value, from which the proposition follows.

    They have indeed the same value, since the definition which is the same for both programs depends only on $r_{x_1,x_2,y}$. Let us show that all constraints are satisfied for $(p_{x_1,x_2}, r_{x_1,x_2,y})_{x_1 \in \mathcal{X}_1,x_2 \in \mathcal{X}_2,y \in \mathcal{Y}}$.

    We have $\sum_{x_1,x_2} r_{x_1,x_2,y} = 1 \leq 1$ so the first constraint is satisfied. We have then that:
    \[\sum_{x_1,x_2} p_{x_1,x_2} = k_1\sum_{x_1,x_2} r^2_{x_1,x_2,y} = k_1k_2\sum_{x_1,x_2} r_{x_1,x_2,y} = k_1k_2 \ ,\]
    so the second constraint is satisfied.

    For the third constraint (and symmetrically the fourth constraint), we have:

    \[\sum_{x_1} p_{x_1,x_2} = k_1\sum_{x_1} r^2_{x_1,x_2,y} \geq  k_1\sum_{x_1} r_{x_1,x_2,y} \ .\]

    Finally, we have directly $0 \leq r_{x_1,x_2,y} \leq p_{x_1,x_2}$, so the last constraint is satisfied.
  \end{proof}

  We can now introduce the relaxed non-signaling assisted capacity region $\mathcal{C}^{\overline{\mathrm{NS}}}(W)$:
  \begin{defi}[$\mathcal{C}^{\overline{\mathrm{NS}}}(W)$]
  A rate pair $(R_1,R_2)$ is achievable with relaxed non-signaling assistance if:
  \[ \underset{n \rightarrow +\infty}{\lim} \mathrm{S}^{\overline{\mathrm{NS}}}(W^{\otimes n},\ceil{2^{R_1n}},\ceil{2^{R_2n}}) = 1 \ . \]
  We define $\mathcal{C}^{\overline{\mathrm{NS}}}(W)$ as the closure of the convex hull of the set of all achievable rate pairs with relaxed non-signaling assistance.
  \end{defi}

  \begin{rk}
    One could show as in the non-relaxed case that $\mathcal{C}^{\overline{\mathrm{NS}}}(W)$ is convex without taking the convex hull in its definition.
  \end{rk}

  A direct property that follows from this definition and Proposition~\ref{prop:NSisrelaxed} is the fact that the non-signaling assisted capacity region is included in the relaxed non-signaling assisted capacity region.
  \begin{cor}
    \label{cor:NSisrelaxed}
      $\mathcal{C}^{\mathrm{NS}}(W) \subseteq \mathcal{C}^{\overline{\mathrm{NS}}}(W)$.
  \end{cor}

  We present now the main result of this section, the characterization of $\mathcal{C}^{\overline{\mathrm{NS}}}(W)$ by a single-letter formula.
  
  \begin{theo}[Characterization of $\mathcal{C}^{\overline{\mathrm{NS}}}(W)$]
        \label{theo:CharaNSrelaxed}
        $\mathcal{C}^{\overline{\mathrm{NS}}}(W)$ is the closure of the convex hull of all rate pairs $(R_1,R_2)$ satisfying:
        \[ R_1 < I(X_1:Y|X_2)\ ,\ R_2 < I(X_2:Y|X_1)\ ,\ R_1+R_2 < I((X_1,X_2):Y) \ ,\]
        for $(X_1,X_2)$ following some law $P_{X_1X_2}$ on $\mathcal{X}_1 \times \mathcal{X}_2$, and $Y \in \mathcal{Y}$ the outcome of $W$ on inputs $X_1,X_2$.
  \end{theo}

  \begin{rk}
    Note that the only difference with the classical capacity region of MACs in Theorem~\ref{theo:capacity} is that the joint distribution of $X_1$ and $X_2$ does not have any product form constraints here.
  \end{rk}

  The proof of Theorem~\ref{theo:CharaNSrelaxed} will be divided in Proposition~\ref{prop:OBNSrelaxed} (outer bound part) and Proposition~\ref{prop:AchievabilityNSrelaxed} (achievability part). But first, let us apply these results to the binary adder channel.
  
  \paragraph{Application to the Binary Adder Channel} 
  Let us determine the relaxed non-signaling assisted capacity of the binary adder channel which will be plotted in Figure~\ref{fig:BACNSrelaxed}.
  
  \begin{prop}
    \label{prop:BACcapacityNSrelaxed}
     $\mathcal{C}^{\overline{\mathrm{NS}}}(W_{\text{BAC}})$ has the following description:
   \[ \mathcal{C}^{\overline{\mathrm{NS}}}(W_{\text{BAC}}) =  \bigcup_{q \in \left[\frac{1}{2},\frac{2}{3}\right]} \{ (R_1,R_2) : R_1 \leq h\left(q\right), R_2 \leq h\left(q\right), R_1+R_2 \leq q+h\left(q\right)\} \ . \]
  \end{prop}

  \begin{rk}
    Note that for $q=\frac{1}{2}$, the bound becomes $R_1 \leq 1, R_2 \leq 1, R_1+R_2 \leq \frac{3}{2}$ and when $q=\frac{2}{3}$ the bound becomes $R_1 \leq \log_2(3)-\frac{2}{3}, R_2 \leq \log_2(3)-\frac{2}{3}, R_1+R_2 \leq \log_2(3)$.
  \end{rk}
  
  \begin{proof}
    We use the characterization of $\mathcal{C}^{\overline{\mathrm{NS}}}$ provided by Theorem~\ref{theo:CharaNSrelaxed}.

    Let us consider an arbitrary $P_{X_1X_2}=(p_{00},p_{01},p_{10},p_{11})$. First, we have that $I((X_1,X_2):Y) = H(Y) - H(Y|X_1X_2) = H(Y)$ since $Y$ is a deterministic function of $(X_1,X_2)$. Then, we have that $I(X_1:Y|X_2) = H(Y|X_2) - H(Y|X_1X_2) = H(Y|X_2)$ for the same  reason. Furthermore, given $X_2$, $Y$  is a deterministic function of $X_1$, so we have $I(X_1:Y|X_2) = H(Y|X_2) - H(Y|X_1X_2) = H(X_1|X_2)$. Symmetrically we have as well $I(X_2:Y|X_1) = H(X_2|X_1)$. In all:
    \[ \mathcal{C}^{\overline{\mathrm{NS}}}(W_{\text{BAC}}) =  \bigcup_{P_{X_1X_2}}\left\{ (R_1,R_2) : R_1 \leq H(X_1|X_2), R_2 \leq H(X_2|X_1), R_1+R_2 \leq H(X_1+X_2)\right\} \]
    
    Let us call $B_1(P_{X_1X_2}):=H(X_1|X_2),B_2(P_{X_1X_2}):=H(X_2|X_1),B_{12}(P_{X_1X_2}):=H(X_1+X_2)$ the three bounds. Let us call $P_{\overline{X}_1\overline{X}_2}=(p_{11},p_{10},p_{01},p_{00})$. One can notice that:
    \begin{equation}
      \begin{aligned}
        B_1(P_{\overline{X}_1\overline{X}_2}) &= H(\overline{X}_1|\overline{X}_2) = H(1-X_1|1-X_2) = H(X_1|X_2) = B_1(P_{X_1X_2}) \ ,\\
        B_2(P_{\overline{X}_1\overline{X}_2}) &= H(1-X_2|1-X_1) = H(X_2|X_1) = B_2(P_{X_1X_2}) \ ,\\
        B_{12}(P_{\overline{X}_1\overline{X}_2}) &= H(\overline{X}_1+\overline{X}_2) = H(1-X_1+1-X_2)\\
        &= H(2-(X_1+X_2))=H(X_1+X_2) = B_{12}(P_{X_1X_2}) \ .
      \end{aligned}
    \end{equation}   
    
    Since $B_{12}(P_{X_1X_2}) = H(X_1+X_2) = H(p_{00},p_{11},p_{01}+p_{10})$, it is concave in $P_{X_1X_2}$ as $H$ is concave and $(p_{00},p_{11},p_{01}+p_{10})$ is linear in $P_{X_1X_2}$. Also, $B_1(P_{X_1X_2}) = H(X_1|X_2) = -D(P_{X_1X_2} ||I \otimes P_{X_2})$ is concave in $P_{X_1X_2}$  as the divergence $D$ is jointly convex and $I \otimes P_{X_2}$ is linear in $P_{X_1X_2}$. By symmetry, $B_2(P_{X_1X_2})$ is as well concave in $P_{X_1X_2}$. Let us consider any of those three bounds, which we call $B$. We have by concavity of $B$ and the fact that $B(P_{X_1X_2})=B(P_{\overline{X}_1\overline{X}_2})$:
    \[B(P_{X_1X_2}) = \frac{B(P_{X_1X_2})+B(P_{\overline{X}_1\overline{X}_2})}{2} \leq B\left(\frac{P_{X_1X_2}+P_{\hat{X}_1\hat{X}_2}}{2}\right) = B\left(\frac{q}{2},\frac{1-q}{2},\frac{1-q}{2},\frac{q}{2}\right) \ , \]
    with $q=p_{00}+p_{11}$. This holds for the three bounds at the same time, so we can restrict ourselves to the distributions of the form $\left(\frac{q}{2},\frac{1-q}{2},\frac{1-q}{2},\frac{q}{2}\right)$ for some $q \in [0,1]$, i.e., $P_{X_1X_2}(0,0) = P_{X_1X_2}(1,1) = \frac{q}{2}$ and $P_{X_1X_2}(0,1) = P_{X_1X_2}(1,0) = \frac{1-q}{2}$.

     We have $P_{Y}(0) = P_{Y}(2) = \frac{q}{2}$ and $P_{Y}(1) = 1-q$, so:

    \begin{equation}
      \begin{aligned}
        B_{12}(P_{X_1X_2}) &= H(Y) = -q\log(\frac{q}{2})-\left(1-q\right)\log(1-q)\\
        &= -q\left(\log(q)-1\right)-\left(1-q\right)\log(1-q)\\
        &= q+h\left(q\right) \ .
      \end{aligned}
    \end{equation}
    
    We have $P_{X_2}(0) = P_{X_1X_2}(0,0)+P_{X_1X_2}(1,0) = \frac{q}{2}+\frac{1-q}{2} = \frac{1}{2}$ so $P_{X_2}(1) = \frac{1}{2}$. Thus:
    \[ B_1(P_{X_1X_2}) = H(X_1|X_2) = \frac{1}{2}H(X_1|X_2=0)+\frac{1}{2}H(X_1|X_2=1) \ .\]
    
    We have $P_{X_1|X_2=0}(0)=\frac{P_{X_1X_2}(0,0)}{P_{X_2}(0)} = q$ so $H(X_1|X_2=0)=h\left(q\right)$. On the other hand, we have $P_{X_1|X_2=1}(1)=\frac{P_{X_1X_2}(1,1)}{P_{X_2}(1)} = q$ so we get as well $H(Y|X_2=1)=h\left(q\right)$, and $B_1(P_{X_1X_2}) = H(X_1|X_2) = h\left(q\right)$. Symmetrically, we also get $B_2(P_{X_1X_2}) = h\left(q\right)$. Therefore, we get that $\mathcal{C}^{\overline{\mathrm{NS}}}(W_{\text{BAC}})$ is the closure of the convex hull of:
    \[\bigcup_{q \in [0,1]}\{ (R_1,R_2) : R_1 < h\left(q\right), R_2 < h\left(q\right), R_1+R_2 < q+h\left(q\right)\} \ .\]

    However this set is already convex, so we have:

    \[\mathcal{C}^{\overline{\mathrm{NS}}}(W_{\text{BAC}}) = \bigcup_{q \in [0,1]}\{ (R_1,R_2) : R_1 \leq h\left(q\right), R_2 \leq h\left(q\right), R_1+R_2 \leq q+h\left(q\right)\} \ .\]
    
    Finally, we can restrict ourselves to $q \in \left[\frac{1}{2},\frac{2}{3}\right]$, since $h$ is increasing from $0$ to $\frac{1}{2}$ (thus $q \mapsto q+h\left(q\right)$ as well), and the fact that $q \mapsto q+h\left(q\right)$ achieves its maximum for $q=\frac{2}{3}$ with $\frac{2}{3}+h\left(\frac{2}{3}\right)=\log_2(3)$ and then decreases (whereas $h$ is decreasing from $\frac{1}{2}$ to $1$), which completes the proof.
  \end{proof}

  As before, one can also define a symmetrized version of the relaxed linear program computing the value $\mathrm{S}^{\overline{\mathrm{NS}}}(W^{\otimes n},k_1,k_2)$ in polynomial time in $n$ and compute the zero-error $n$-shots capacity region by looking at the rates where $\mathrm{S}^{\overline{\mathrm{NS}}}(W^{\otimes n},k_1,k_2)=1$. We have computed this up to $7$ copies of the binary adder channel, which led to Figure~\ref{fig:BACNSrelaxed}:
  
  \begin{figure}[!h]
    \begin{center}
      \begin{tikzpicture}
        [spy using outlines={circle, magnification=4, size=2cm, connect spies}]
        \begin{axis}[
          xmin = 0, xmax = 1.05,
          ymin = 0, ymax = 1.05,
          xtick distance = 0.25,
          ytick distance = 0.25,
          grid = both,
          width = 0.55\textwidth,
          height = 0.55\textwidth,
          legend cell align = {left},
          legend pos = south west,
          ylabel=$R_2$,
          xlabel=$R_1$,
        ]
        \addplot[
          thick,
          dashed,
          black,
        ] coordinates {(0,1) (0.5,1) (1,0.5) (1,0)};
        \addplot[
          thick,
          dotted,
          darkgray,
        ] coordinates {(0,1) (0.5,1)};{(0.6666667,0.9182958) (0.9182958,0.6666667)} {(1,0.5) (1,0)};
        \addplot[
          domain = 0:1,
          blue,
        ] table[x=R1,y=R2,col sep=comma] {Data_BAC/ZEBorder7.csv};
        \addplot[
          domain = 0:1,
          red,
        ] table[x=R1,y=R2,col sep=comma] {Data_BAC_relaxed/ZEBorderINEQ7.csv};
        \addplot[
          thick,
          dotted,
          darkgray,
        ] coordinates {(0.6666667,0.9182958) (0.9182958,0.6666667)};
         \addplot[
          thick,
          dotted,
          darkgray,
         ] coordinates {(1,0.5) (1,0)};
         \addplot[
           thick,
           dotted,
           darkgray,
           domain = 0.5:0.6666667,
         ] {-(\x)*ln(\x)/ln(2)-(1-\x)*ln(1-\x)/ln(2)};
         \draw[
           thick,
           dotted,
           darkgray,
           domain = 0.48:0.6666667,
           variable=\y] plot({-(\y)*ln(\y)/ln(2)-(1-\y)*ln(1-\y)/ln(2)},\y);
        \coordinate (spypoint) at (0.77,0.77);
        \coordinate (magnifyglass) at (0.4,0.6);
        
        \legend{$\mathcal{C}(W_{\text{BAC}})$, $\mathcal{C}^{\overline{\mathrm{NS}}}(W_{\text{BAC}})$, $\mathcal{C}_{0,\leq 7}^{\mathrm{NS}}(W_{\text{BAC}})$, $\mathcal{C}_{0,\leq 7}^{\overline{\mathrm{NS}}}(W_{\text{BAC}})$}
        \end{axis}

        \spy [size=3cm] on (spypoint) in node [fill=white] at (magnifyglass);
      \end{tikzpicture}
      \caption{Comparison of relaxed and regular non-signaling assisted capacity regions of the binary adder channel. The black dashed curve depicts the classical capacity region $\mathcal{C}(W_{\text{BAC}})$, whereas the grey dotted curve depicts the relaxed non-signaling assisted capacity region $\mathcal{C}^{\overline{\mathrm{NS}}}(W_{\text{BAC}})$ as described in Proposition~\ref{prop:BACcapacityNSrelaxed}. In particular, the curved corners are obtained by taking $R_1=h(R_2)$ for $R_2 \in \left[\frac{1}{2},\frac{2}{3}\right]$ and symmetrically by switching the roles played by $R_1$ and $R_2$. The continuous blue (respectively red) curve depicts the zero-error (respectively relaxed) non-signaling assisted achievable rate pairs for $7$ copies of the binary adder channel.}
      \label{fig:BACNSrelaxed}
    \end{center}   
\end{figure}
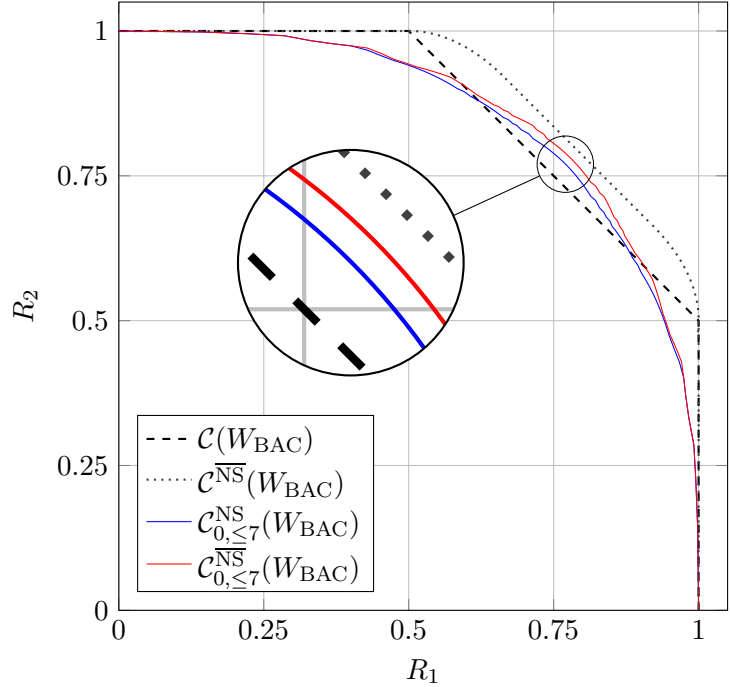

The first noticeable result coming from these curves is that the values $\mathrm{S}^{\overline{\mathrm{NS}}}$ and $\mathrm{S}^{\mathrm{NS}}$ differ. While the highest sum-rate of $\frac{2\log_2(42)}{7} \simeq 1.5406$ is achieved on $7$ copies of the binary adder channel with zero-error and non-signaling assistance, coming from the fact that $\mathrm{S}^{\mathrm{NS}}(W^{\otimes 7}_{\mathrm{BAC}},42,42)=1$, we have that $\mathrm{S}^{\overline{\mathrm{NS}}}(W^{\otimes 7}_{\mathrm{BAC}},44,44)=1>\mathrm{S}^{\mathrm{NS}}(W^{\otimes 7}_{\mathrm{BAC}},44,44) \simeq 0.9581$ which implies that a sum-rate of $\frac{2\log_2(44)}{7} \simeq 1.5598$ is achieved on $7$ copies of the binary adder channel with zero-error and relaxed non-signaling assistance. It also largely beats the best found sum-rate of $\frac{\log_2(72)}{4} \simeq 1.5425$ achieved on $8$ copies with the regular version. However the fact that the non-signaling assisted capacity region is strictly contained in the relaxed one is still open, as the same rates could potentially be achieved by the cost of using more copies of the channel.
  
  \subsection{Outer Bound Part of Theorem~\ref{theo:CharaNSrelaxed}}
In order to prove Proposition~\ref{prop:OBNSrelaxed}, we use a connection between hypothesis testing and relaxed non-signaling assisted codes as established in~\cite{Matthews12} for point-to-point channels.

\begin{defi}[Hypothesis Testing]
  \label{defi:beta}
  Given distributions $P^{(0)}$ and  $P^{(1)}$ on the same space $C$, we define $\beta_{1-\varepsilon}(P^{(0)},P^{(1)})$ to be the minimum type II error $\sum_{r \in C} T_rP^{(1)}(r)$ that can be achieved by statistical tests $T$ which give a type I error no greater than $\varepsilon$, i.e. $\sum_{r \in C} T_rP^{(0)}(r) \geq 1-\varepsilon$.

  In other words, we have that:
  \begin{equation}
    \begin{aligned}
      \beta_{1-\varepsilon}(P^{(0)},P^{(1)}) = &&&\underset{T_r}{\mini} &&\sum_{r \in C} T_rP^{(1)}(r)\\
      &&&\st &&\sum_{r \in C} T_rP^{(0)}(r) \geq 1-\varepsilon\\
      &&&&& 0 \leq T_r \leq 1 \ .
    \end{aligned}
  \end{equation}
\end{defi}

\begin{lem}
  \label{lem:beta}
    For any relaxed non-signaling assisted code $(p_{x_1,x_2}, r_{x_1,x_2,y})_{x_1 \in \mathcal{X}_1,x_2 \in \mathcal{X}_2,y \in \mathcal{Y}}$ with $(k_1,k_2)$ messages and a probability of success $1-\varepsilon$, if $P_{X_1X_2}(x_1,x_2)=\frac{p_{x_1,x_2}}{k_1k_2}$ and $Y \in \mathcal{Y}$ is the outcome of $W$ on inputs $X_1,X_2$, we have:
    \begin{equation}
      \begin{aligned}
        &\beta_{1-\varepsilon}\left(P_{X_1X_2Y},P_{X_1X_2} \times P_{Y|X_2}\right) \leq \frac{1}{k_1}\\
        &\beta_{1-\varepsilon}\left(P_{X_1X_2Y},P_{X_1X_2} \times P_{Y|X_1}\right) \leq \frac{1}{k_2}\\
        &\beta_{1-\varepsilon}\left(P_{X_1X_2Y},P_{X_1X_2} \times P_{Y}\right) \leq \frac{1}{k_1k_2}       \ .
      \end{aligned}
    \end{equation}
  \end{lem}

\begin{rk}
  These three bounds are actually achieved with the same statistical test.
\end{rk}

  \begin{proof}
    This result is a direct generalization of Theorem 9 in~\cite{Matthews12} for point-to-point channels, itself a generalization of Theorem 27 in~\cite{PPV10} without non-signaling assistance.

    Let us name $W_0:=W$ and $W_1$ a MAC yet to be defined. The coding strategy described by $r_{x_1,x_2,y}$ and $p_{x_1, x_2}$ leads to a probability of success on channel $i \in \{0,1\}$ is given by:
  \begin{equation}
    \begin{aligned}
      1 - \varepsilon_i &= \frac{1}{k_1k_2}\sum_{x_1,x_2,y} r_{x_1,x_2,y}W_i(y|x_1x_2)\\
      &= \sum_{x_1,x_2,y:p_{x_1,x_2} > 0} \frac{r_{x_1,x_2,y}}{p_{x_1,x_2}}W_i(y|x_1x_2)\frac{p_{x_1,x_2}}{k_1k_2} \quad \text{ since } 0 \leq r_{x_1,x_2,y} \leq p_{x_1,x_2}\\
      &= \sum_{x_1,x_2,y} T_{x_1,x_2,y}W_i(y|x_1x_2)\frac{p_{x_1,x_2}}{k_1k_2} \ ,
    \end{aligned}
  \end{equation}
  with $T_{x_1,x_2,y} := \frac{r_{x_1,x_2,y}}{p_{x_1,x_2}}$ if $p_{x_1,x_2} > 0$, and $T_{x_1,x_2,y} :=0$ otherwise.

  If now $Y$ is the output of the channel $W_i$, the joint distribution of $X_1,X_2,Y$ is given by $P^{(i)}_{X_1X_2Y}(x_1,x_2,y) = W_i(y|x_1x_2)P_{X_1X_2}(x_1,x_2)= W_i(y|x_1x_2)\frac{p_{x_1,x_2}}{k_1k_2}$.

  On the other hand, we have that for all $x_1,x_2,y$, $0 \leq T_{x_1,x_2,y} \leq 1$ since $0 \leq r_{x_1,x_2,y} \leq p_{x_1,x_2}$. So we get that:

  \[ 1 - \varepsilon_i = \sum_{x_1,x_2,y} T_{x_1,x_2,y}P^{(i)}_{X_1X_2Y}(x_1,x_2,y) \ . \]

  Since $\sum_{x_1,x_2,y} T_{x_1,x_2,y}P^{(0)}_{X_1X_2Y}(x_1,x_2,y) \geq 1 - \varepsilon_0$ and $0 \leq T_{x_1,x_2,y} \leq 1$, we have:
  \[ \beta_{1-\varepsilon_0}(P^{(0)},P^{(1)}) \leq \sum_{x_1,x_2,y} T_{x_1,x_2,y}P^{(1)}_{X_1X_2Y}(x_1,x_2,y) = 1 - \varepsilon_1 \ . \]

  Let us now consider three general cases, depending on the fact that $W_1$ does not depend on $x_1$, $x_2$ or both: $W_1(y|x_1x_2) := Q^{(1)}(y|x_2)$; $W_1(y|x_1x_2) := Q^{(2)}(y|x_1)$;  $W_1(y|x_1x_2) := Q^{(0)}(y)$. These will give respectively the three bounds we want.
 
  First, let us consider the case where $W_1(y|x_1x_2) := Q^{(1)}(y|x_2)$ (the second case where $W_1(y|x_1x_2) := Q^{(2)}(y|x_1)$ being symmetric), we have that:

  \begin{equation}
    \begin{aligned}
      1 - \varepsilon_1 &= \sum_{x_1,x_2,y} T_{x_1,x_2,y}Q^{(1)}(y|x_2)\frac{p_{x_1,x_2}}{k_1k_2} = \frac{1}{k_1k_2} \sum_{x_2,y} Q^{(1)}(y|x_2) \sum_{x_1}T_{x_1,x_2,y}p_{x_1,x_2}\\
      &= \frac{1}{k_1k_2}\sum_{x_2,y} Q^{(1)}(y|x_2) \sum_{x_1}r_{x_1,x_2,y} \leq \frac{1}{k_1k_2} \sum_{x_2,y} Q^{(1)}(y|x_2)\frac{1}{k_1}\sum_{x_1}p_{x_1,x_2}\\
      &= \frac{1}{k_1} \sum_{x_1,x_2}\frac{p_{x_1,x_2}}{k_1k_2}\sum_yQ^{(1)}(y|x_2) = \frac{1}{k_1}\sum_{x_1,x_2}\frac{p_{x_1,x_2}}{k_1k_2} = \frac{1}{k_1} \ .
    \end{aligned}
  \end{equation}
  
  For the third case, when $W_1(y|x_1x_2) := Q^{(0)}(y)$, we have:
  \begin{equation}
    \begin{aligned}
      1 - \varepsilon_1 &= \sum_{x_1,x_2,y} T_{x_1,x_2,y}Q^{(0)}(y)\frac{p_{x_1,x_2}}{k_1k_2} = \frac{1}{k_1k_2} \sum_{y} Q^{(0)}(y) \sum_{x_1,x_2}T_{x_1,x_2,y}p_{x_1,x_2}\\
      &= \frac{1}{k_1k_2}\sum_{y} Q^{(0)}(y) \sum_{x_1,x_2}r_{x_1,x_2,y} \leq \frac{1}{k_1k_2}\sum_{y} Q^{(0)}(y) = \frac{1}{k_1k_2} \ .
    \end{aligned}
  \end{equation}
  
  In those three cases, we have respectively $P^{(1)}_{X_1X_2Y} = P_{X_1X_2} \times Q^{(1)}_{Y|X_2}; P_{X_1X_2} \times Q^{(2)}_{Y|X_1}; P_{X_1X_2} \times Q^{(0)}_{Y}$. Specializing those cases with $Q^{(1)}_{Y|X_2}:=P_{Y|X_2}; Q^{(2)}_{Y|X_1}:=P_{Y|X_1}; Q^{(0)}_{Y}:=P_Y$ and using the fact that ${\beta_{1-\varepsilon_0}\left(P^{(0)},P^{(1)}\right) \leq 1 - \varepsilon_1}$ concludes the proof.
  \end{proof}

  \begin{lem}
    \label{lem:oneshotOB}
    For any relaxed non-signaling assisted code $(p_{x_1,x_2}, r_{x_1,x_2,y})_{x_1 \in \mathcal{X}_1,x_2 \in \mathcal{X}_2,y \in \mathcal{Y}}$ with $(k_1,k_2)$ messages and a probability of success $1-\varepsilon$, if $P_{X_1X_2}(x_1,x_2)=\frac{p_{x_1,x_2}}{k_1k_2}$ and $Y \in \mathcal{Y}$ is the outcome of $W$ on inputs $X_1,X_2$, we have:
    \begin{equation}
      \begin{aligned}
        \log(k_1) &\leq \frac{I(X_1:Y|X_2)+h(\varepsilon)}{1-\varepsilon} \ ,\\
        \log(k_2) &\leq \frac{I(X_2:Y|X_1)+h(\varepsilon)}{1-\varepsilon} \ ,\\
        \log(k_1)+\log(k_2) &\leq \frac{I((X_1,X_2):Y)+h(\varepsilon)}{1-\varepsilon} \ .
      \end{aligned}
    \end{equation}
  \end{lem}
  
  \begin{proof}
    Thanks to Lemma~\ref{lem:beta}, with the fact that $P_{X_1X_2} = P_{X_1|X_2} \times P_{X_2} = P_{X_2|X_1} \times P_{X_1}$, we have already:
    \begin{equation}
      \begin{aligned}
        &\beta_{1-\varepsilon}\left(P_{X_1X_2Y},\left(P_{X_1|X_2} \times P_{Y|X_2} \right) \times P_{X_2}\right) \leq \frac{1}{k_1}\\
        &\beta_{1-\varepsilon}\left(P_{X_1X_2Y},\left(P_{X_2|X_1} \times P_{Y|X_1} \right) \times P_{X_1}\right) \leq \frac{1}{k_2}\\
        &\beta_{1-\varepsilon}\left(P_{X_1X_2Y},P_{X_1X_2} \times P_{Y}\right) \leq \frac{1}{k_1k_2}       
      \end{aligned}
    \end{equation}

    Following the steps of section G in~\cite{PPV10}, since any hypothesis test is a binary-output transformation, by data-processing inequality for divergence, we have that:

    \begin{equation}
      \begin{aligned}
        &d\left(1-\varepsilon||\beta_{1-\varepsilon}\left(P_{X_1X_2Y},\left(P_{X_1|X_2} \times P_{Y|X_2} \right) \times P_{X_2}\right)\right)\\
        &= d\left(\beta_{1-\varepsilon}\left(P_{X_1X_2Y},P_{X_1X_2Y}\right)||\beta_{1-\varepsilon}\left(P_{X_1X_2Y},\left(P_{X_1|X_2} \times P_{Y|X_2} \right) \times P_{X_2}\right)\right)\\
        &\leq D\left(P_{X_1X_2Y}||\left(P_{X_1|X_2} \times P_{Y|X_2} \right) \times P_{X_2}\right) = I(X_1:Y|X_2)
      \end{aligned}
    \end{equation}

    where the binary divergence $d(a||b):= a\log\left(\frac{a}{b}\right) + (1-a)\log\left(\frac{1-a}{1-b}\right)$ and satisfies, $d(a||b) \geq -h(a) - a\log(b)$ and thus:
    \[ \log(\frac{1}{b}) \leq \frac{d(a||b) + h(a)}{a} = \frac{d(a||b) + h(1-a)}{a} \ , \]

    This leads to:
    \[\log(k_1) \leq \frac{1}{\log\left(\left(\beta_{1-\varepsilon}\left(P_{X_1X_2Y},\left(P_{X_1|X_2} \times P_{Y|X_2} \right) \times P_{X_2}\right)\right)\right)} \leq \frac{I(X_1:Y|X_2) + h(\varepsilon)}{1-\varepsilon} \ .\]

    Similarly for the two other inequalities, since $D\left(P_{X_1X_2Y}||\left(P_{X_2|X_1} \times P_{Y|X_1} \right) \times P_{X_1}\right) = I(X_2:Y|X_1)$ and $D\left(P_{X_1X_2Y}||P_{X_1X_2} \times P_Y \right) =  I((X_1,X_2):Y)$, we get:
    \begin{equation}
      \begin{aligned}
        \log(k_1) &\leq \frac{I(X_1:Y|X_2)+h(\varepsilon)}{1-\varepsilon} \ ,\\
        \log(k_2) &\leq \frac{I(X_2:Y|X_1)+h(\varepsilon)}{1-\varepsilon} \ ,\\
        \log(k_1)+\log(k_2) &\leq \frac{I((X_1,X_2):Y)+h(\varepsilon)}{1-\varepsilon} \ .
      \end{aligned}
    \end{equation}
  \end{proof}

In order to show additivity of the outer bound, we use the following lemma.

  \begin{lem}
    \label{lem:multiletterOB}
     For any distribution $P_{X_1^nX_2^n}$ of $(X_1^n,X_2^n)$, if $Y^n \in \mathcal{Y}^n$ is the outcome of $W^n$ on inputs $X_1^n,X_2^n$, we have:
      \begin{equation}
        \begin{aligned}
          I(X_1^n:Y^n|X_2^n) &\leq \sum_{i=1}^n I(X_{1,i}:Y_i|X_{2,i})\\
          I(X_2^n:Y^n|X_1^n) &\leq \sum_{i=1}^n I(X_{2,i}:Y_i|X_{1,i})\\
          I((X_1^n,X_2^n):Y^n) &\leq \sum_{i=1}^n I((X_{1,i},X_{2,i}):Y_i) \ .
        \end{aligned}
      \end{equation}
  \end{lem}
  
  \begin{proof}
    Consider $n$ copies of the MAC $W$. Let us write $X_{1,-i} := X_{1,1} \ldots X_{1,i-1} X_{1,i+1} \ldots X_{1,n}$ and $Z^n := Z_1 \ldots Z_n$. We have:

    \begin{equation}
      \begin{aligned}
        I(X_1^n:Y^n|X_2^n) &= I(X_1^n:Y^n|X_2^n)\\
        &= \sum_{i=1}^n I(X_1^n:Y_i|X_2^nY^{i-1}) \text{ by the chain rule}\\
        &= \sum_{i=1}^n I(X_{1,i}:Y_i|X_2^nY^{i-1}) + \sum_{i=1}^n I(X_{1,-i}:Y_i|X_2^nY^{i-1}X_{1,i})\\
        &= \sum_{i=1}^n I(X_{1,i}:Y_i|X_2^nY^{i-1}) \ ,
        \end{aligned}
     \end{equation}
     where the last equality comes from Lemma~\ref{lem:claim1}. As a result,
     \begin{equation}
     \begin{aligned}
        I(X_1^n:Y^n|X_2^n) &= \sum_{i=1}^n H(Y_i|X_2^nY^{i-1}) - H(Y_i|X_2^nY^{i-1}X_{1,i})\\
        &= \sum_{i=1}^n H(Y_i|X_2^nY^{i-1}) - H(Y_i|X_{2,i}X_{1,i}) \text{ since $X_{2,-i} Y^{i-1} \rightarrow (X_{1,i},X_{2,i}) \rightarrow Y_i$ Markov chain.}\\
        &\leq \sum_{i=1}^n H(Y_i|X_{2,i}) - H(Y_i|X_{2,i}X_{1,i}) = \sum_{i=1}^n I(X_{1,i}:Y_i|X_{2,i}) \ .
      \end{aligned}
    \end{equation}
    
      Symmetrically by switching the roles of $X_1$ and $X_2$, we get the second part of Lemma~\ref{lem:multiletterOB}.
      
      For the sum-rate case:
      \begin{equation*}
        \begin{aligned}
          I((X_1^n,X_2^n):Y^n) &= \sum_{i=1}^n I((X_1^n,X_2^n):Y_i|Y^{i-1}) \text{ by the chain rule}\\
          &= \sum_{i=1}^n I((X_{1,i},X_{2,i}):Y_i|Y^{i-1}) + \sum_{i=1}^n I((X_{1,-i},X_{2,-i}):Y_i|Y^{i-1}X_{1,i}X_{2,i})\\
          &= \sum_{i=1}^n I((X_{1,i},X_{2,i}):Y_i|Y^{i-1}) \text{ since $(X_{1,-i},X_{2,-i}) \rightarrow Y^{i-1}X_{1,i}X_{2,i} \rightarrow Y_i$ Markov chain.}\\
          &= \sum_{i=1}^n H(Y_i|Y^{i-1}) - H(Y_i|Y^{i-1}X_{1,i}X_{2,i})\\
          &= \sum_{i=1}^n H(Y_i|Y^{i-1}) - H(Y_i|X_{1,i}X_{2,i}) \text{ since $Y^{i-1} \rightarrow (X_{1,i},X_{2,i}) \rightarrow Y_i$ Markov chain.}\\
          &\leq \sum_{i=1}^n H(Y_i) - H(Y_i|X_{2,i}X_{1,i}) = \sum_{i=1}^n I((X_{1,i},X_{2,i}):Y_i) \ .
        \end{aligned}
      \end{equation*}
      
  \end{proof}
  
  We next prove a technical lemma that was used in the previous proof.
  
      \begin{lem}
        \label{lem:claim1}
        For any distribution $P_{X_1^nX_2^n}$ of $(X_1^n,X_2^n)$, if $Y^n \in \mathcal{Y}^n$ is the outcome of $W^n$ on inputs $X_1^n,X_2^n$, we have:
        \[ I(X_{1,-i}:Y_i|X_{1,i}X_2^nY^{i-1}) = 0 \ . \]
      \end{lem}
      \begin{proof}
        Let us show that, conditioned on any particular instance of $X_{1,i}=x_{i,1}, X_2^n=x_2^n, Y_1^{i-1}=y^{i-1}$, $X_{1,-i}$ and $Y_i$ are independent.

        We have:
        \[ \mathbb{P}\left(Y_i = y_i |  X_{1,i}=x_{i,1}, X_2^n=x_2^n, Y_1^{i-1}=y^{i-1}\right) = \mathbb{P}\left(Y_i = y_i |  X_{1,i}=x_{i,1}, X_{2,i}=x_{2,i}\right) = W(y_i|x_{1,i}x_{2,i}) \ , \]
        by definition of the law of $Y_i$. On the other hand, we have that:
        \[ \mathbb{P}\left(X_1^n=x_1^n,X_2^n=x_2^n,Y^n=y_n\right) = \mathbb{P}\left(X_1^n=x_1^n,X_2^n=x_2^n\right)\prod_{j=1}^n W(y_j|x_{1,j}x_{2,j}) \ . \]
        Thus, we have:
        \begin{equation}
          \begin{aligned}
            &\mathbb{P}\left(X_{1,i}=x_{i,1}, X_2^n=x_2^n, Y_1^{i-1}=y^{i-1}\right) = \sum_{x_{1,-i},x_2^n,y_i,\ldots, y_n}\mathbb{P}\left(X_1^n=x_1^n,X_2^n=x_2^n\right)\prod_{j=1}^n W(y_j|x_{1,j}x_{2,j})\\
            &= \sum_{x_{1,-i},x_2^n} \mathbb{P}\left(X_1^n=x_1^n,X_2^n=x_2^n\right)\prod_{j=1}^{i-1}W(y_j|x_{1,j}x_{2,j})\prod_{j=i}^n\left(\sum_{y_j}W(y_j|x_{1,j}x_{2,j})\right)\\
            &= \sum_{x_{1,-i},x_2^n} \mathbb{P}\left(X_1^n=x_1^n,X_2^n=x_2^n\right)\prod_{j=1}^{i-1}W(y_j|x_{1,j}x_{2,j}) \ .
          \end{aligned}
        \end{equation}

        And then:
        \[ \mathbb{P}\left(X_{1,-i}=x_{1,-i}|X_{1,i}=x_{i,1}, X_2^n=x_2^n, Y_1^{i-1}=y^{i-1}\right) = \frac{\sum_{x_2^n}\mathbb{P}\left(X_1^n=x_1^n,X_2^n=x_2^n\right)\prod_{j=1}^{i-1} W(y_j|x_{1,j}x_{2,j})}{\sum_{x_{1,-i},x_2^n}\mathbb{P}\left(X_1^n=x_1^n,X_2^n=x_2^n\right)\prod_{j=1}^{i-1} W(y_j|x_{1,j}x_{2,j})} \ . \]

        But:

        \begin{equation}
          \begin{aligned}
            &\mathbb{P}\left(X_{1,-i}=x_{1,-i},Y_i=y_i|X_{1,i}=x_{i,1}, X_2^n=x_2^n, Y_1^{i-1}=y^{i-1}\right)\\
            &= \frac{\sum_{x_2^n}\mathbb{P}\left(X_1^n=x_1^n,X_2^n=x_2^n\right)\prod_{j=1}^{i} W(y_j|x_{1,j}x_{2,j})}{\sum_{x_{1,-i},x_2^n}\mathbb{P}\left(X_1^n=x_1^n,X_2^n=x_2^n\right)\prod_{j=1}^{i-1} W(y_j|x_{1,j}x_{2,j})} \\
            &= \mathbb{P}\left(X_{1,-i}=x_{1,-i}|X_{1,i}=x_{i,1}, X_2^n=x_2^n, Y_1^{i-1}=y^{i-1}\right) W(y_i|x_{1,i}x_{2,i})\\
            &= \mathbb{P}\left(X_{1,-i}=x_{1,-i}|X_{1,i}=x_{i,1}, X_2^n=x_2^n, Y_1^{i-1}=y^{i-1}\right) \mathbb{P}\left(Y_i=y_i|X_{1,i}=x_{i,1}, X_2^n=x_2^n, Y_1^{i-1}=y^{i-1}\right) \ .
          \end{aligned}
        \end{equation}

        Thus, conditioned on any particular instance of $X_{1,i}=x_{i,1}, X_2^n=x_2^n, Y_1^{i-1}=y^{i-1}$, $X_{1,-i}$ and $Y_i$ are independent, and so $I(X_{1,-i}:Y_i|X_{1,i}X_2^nY^{i-1}) = 0$.
      \end{proof}

	Combining the previous results gives the desired outer bound.

      \begin{prop}[Outer bound part of Theorem~\ref{theo:CharaNSrelaxed}]
        \label{prop:OBNSrelaxed}
        If a rate pair is achievable with relaxed non-signaling assistance then it is in the closure of the convex hull of all $(R_1,R_2)$ satisfying:
        \[ R_1 < I(X_1:Y|X_2)\ ,\ R_2 < I(X_2:Y|X_1)\ ,\ R_1+R_2 < I((X_1,X_2):Y) \ ,\]
        for $(X_1,X_2)$ following some law $P_{X_1X_2}$ on $\mathcal{X}_1 \times \mathcal{X}_2$, and $Y \in \mathcal{Y}$ the outcome of $W$ on inputs $X_1,X_2$.
      \end{prop}
      \begin{proof}
        Consider $(R_1,R_2)$ achievable with relaxed non-signaling assistance: we have a sequence of relaxed non-signaling assisted codes for $n$ copies of the MAC $W$ with $k_1 = 2^{nR_1}, k_2 = 2^{nR_2}$ messages and an error probability $\varepsilon_n \underset{n \rightarrow +\infty}{\rightarrow} 0$, along with associated distributions of $X_1^nX_2^nY^n$.

    Thus combining Lemma~\ref{lem:oneshotOB} and Lemma~\ref{lem:multiletterOB}, we have that:
    \begin{equation}
      \begin{aligned}
        R_1 &\leq \frac{1}{n}\frac{I(X_1^n:Y^n|X_2^n)+h(\varepsilon_n)}{1-\varepsilon_n} \leq \frac{1}{n} \frac{\sum_{i=1}^nI(X_{1,i}:Y_i|X_{2,i})+h(\varepsilon_n)}{1-\varepsilon_n}\ ,\\
        R_2 &\leq \frac{1}{n}\frac{I(X_2^n:Y^n|X_1^n)+h(\varepsilon_n)}{1-\varepsilon_n} \leq \frac{1}{n}\frac{\sum_{i=1}^nI(X_{2,i}:Y_i|X_{1,i})+h(\varepsilon_n)}{1-\varepsilon_n} \ ,\\
        R_1+R_2 &\leq \frac{1}{n}\frac{I((X_1^n,X_2^n):Y^n)+h(\varepsilon_n)}{1-\varepsilon_n} \leq \frac{1}{n} \frac{\sum_{i=1}^nI((X_{1,i},X_{2,i}):Y_i)+h(\varepsilon_n)}{1-\varepsilon_n} \ .
      \end{aligned}
    \end{equation}

    Then let us consider some random variable $Q$ uniform on $[n]$ and independent from $(X_1^n,X_2^n,Y^n)$. Then we can write:
    \[ \sum_{i=1}^nI(X_{1,i}:Y_i|X_{2,i}) = \sum_{i=1}^nI(X_{1,i}:Y_i|X_{2,i},Q=i) = nI(X_{1,Q}:Y_Q|X_{2,Q},Q ) \ . \]

    Since $Y_Q$ conditioned on $X_{1,Q}$ and $X_{2,Q}$ still follows the law of the MAC $W(y|x_1x_2)$, we can take $X_1=X_{1,Q},X_2=X_{2,Q}$, and then the output of the channel $Y$ satisfies $Y=Y_Q$, and thus we obtain:
    \[ R_1 \leq \frac{I(X_1:Y|X_2,Q)+\frac{h(\varepsilon_n)}{n}}{1-\varepsilon_n} \ .\]
    
    Doing this similarly on the other conditional mutual informations, we get:

    \begin{equation}
      \begin{aligned}
        R_1 &\leq \frac{I(X_1:Y|X_2,Q)+\frac{h(\varepsilon_n)}{n}}{1-\varepsilon_n}\ ,\\
        R_2 &\leq \frac{I(X_2:Y|X_1,Q)+\frac{h(\varepsilon_n)}{n}}{1-\varepsilon_n} \ ,\\
        R_1+R_2 &\leq  \frac{I((X_1,X_2):Y|Q)+\frac{h(\varepsilon_n)}{n}}{1-\varepsilon_n} \ .
      \end{aligned}
    \end{equation}

    By taking the limit as $n$ goes to infinity, since the limit of $\varepsilon_n$ is $0$, then the limit of $\frac{h(\varepsilon_n)}{n}$ is $0$ as well and we get that $(R_1,R_2)$ must be in the set of rate pairs such that:
    \begin{equation}
      \begin{aligned}
        R_1 &\leq I(X_1:Y|X_2,Q) \ ,\\
        R_2 &\leq I(X_2:Y|X_1,Q) \ ,\\
        R_1+R_2 &\leq I((X_1,X_2):Y|Q) \ ,
      \end{aligned}
    \end{equation}
    for some uniform $Q$ in a finite set, $(X_1,X_2)$ any joint law depending on $Q$, and $Y$ the output of $W$ on inputs $(X_1,X_2)$.

    Finally, in order to show that this is the right region, one has only to see that the corner points of this region, such as for instance $(I(X_1:Y|Q),I(X_2:Y|X_1,Q))$, are finite convex combination of the points $(I(X_1:Y|Q=q),I(X_2:Y|X_1,Q=q))$ which are all in the capacity region of the theorem by taking $(X_1X_2) \sim P_{X_1X_2|Q=q}$. This implies that $(R_1,R_2)$ is in the convex hull of that region, so we can drop the random variable $Q$ and the proof is completed.
      \end{proof}

      The main consequence of that outer bound on the relaxed non-signaling assisted capacity region is that it holds also for the non-signaling assisted capacity region thanks to Corollary~\ref{cor:NSisrelaxed}:

      \begin{cor}[Outer Bound on the Non-Signaling Assisted Capacity Region]
        \label{cor:OB}
        If a rate pair is achievable with non-signaling assistance, then it is in the closure of the convex hull of all $(R_1,R_2)$ satisfying:
        \[ R_1 < I(X_1:Y|X_2)\ ,\ R_2 < I(X_2:Y|X_1)\ ,\ R_1+R_2 < I((X_1,X_2):Y) \ ,\]
        for $(X_1,X_2)$ following any law $P_{X_1X_2}$ on $\mathcal{X}_1 \times \mathcal{X}_2$, and $Y \in \mathcal{Y}$ the outcome of $W$ on inputs $X_1,X_2$.
      \end{cor}

      \subsection{Achievability Part of Theorem~\ref{theo:CharaNSrelaxed}}
      In order to construct the relaxed non-signaling assisted code for achievability, we will need the notions of jointly and conditional typical sets. We will consider the following typical sets defined in Chapter 2.5 of~\cite{GK11}: $\mathcal{T}^n_{\varepsilon}(X_1,X_2,Y)$, $\mathcal{T}^n_{\varepsilon}(X_1,X_2)$, $\mathcal{T}^n_{\varepsilon}(Y)$, $\mathcal{T}^n_{\varepsilon}(X_1|x_2^n)$, $\mathcal{T}^n_{\varepsilon}(X_2|x_1^n)$, $\mathcal{T}^n_{\varepsilon}(X_1|x_2^n,y^n)$, $\mathcal{T}^n_{\varepsilon}(X_2|x_1^n,y^n)$, $\mathcal{T}^n_{\varepsilon}(X_1,X_2|y^n)$. Recall that:
      \begin{defi}[Typical set and conditional typical set]
        We have the following definitions:
        \begin{enumerate}
        \item $\mathcal{T}^n_{\varepsilon}(X_1,X_2) := \left\{(x_1^n,x_2^n) : |\pi(x_1,x_2|x_1^n,x_2^n) - P_{X_1X_2}(x_1,x_2)|\leq \varepsilon P_{X_1X_2}(x_1,x_2) \text{ for all } (x_1,x_2) \in \mathcal{X}_1 \times \mathcal{X}_2\right\}$
      where $\pi(x_1,x_2|x_1^n,x_2^n) := \frac{|\{i : (x_{1,i},x_{2,i}) = (x_1,x_2)\}|}{n}$. This definition generalizes for any $t$-uple of variables.
    \item $\forall y^n \in \mathcal{T}^n_{\varepsilon}(Y), \mathcal{T}^n_{\varepsilon}(X_1,X_2|y^n) := \{ (x_1^n,x_2^n) : (x_1^n,x_2^n,y^n) \in \mathcal{T}^n_{\varepsilon}(X_1,X_2,Y) \}$
        \end{enumerate}
      \end{defi}

      A crucial property of such typical sets is the typical average lemma:

      \begin{lem}[Typical Average Lemma~\cite{GK11}]
        Let $(x_1^n,x_2^n) \in \mathcal{T}^n_{\varepsilon}(X_1,X_2)$. Then for any nonnegative function $g$ on $\mathcal{X}_1\times\mathcal{X}_2$:
        \[ (1-\varepsilon)\mathbb{E}[g(X_1,X_2)] \leq \frac{1}{n}\sum_{i=1}^ng(x_{1,i},x_{2,i}) \leq (1+\varepsilon)\mathbb{E}[g(X_1,X_2)] \ . \]
      \end{lem}
      
      In particular, with this tool, we can derive the following properties:
      
      \begin{lem}[Properties of typical sets~\cite{GK11}]
        We have, among others, the following statements about typical sets:
        \begin{enumerate}
        \item $\forall (x_1^n,x_2^n) \in \mathcal{T}^n_{\varepsilon}(X_1,X_2), 2^{-n(1+\varepsilon)H(X_1,X_2)} \leq P_{X_1^nX_2^n}(x_1^n,x_2^n) \leq 2^{-n(1-\varepsilon)H(X_1,X_2)}$.
        \item $\underset{n \rightarrow +\infty}{\lim} \mathbb{P}\left((X_1^n,X_2^n) \in \mathcal{T}^n_{\varepsilon}(X_1,X_2) \right) = 1$.
        \item $|\mathcal{T}^n_{\varepsilon}(X_1,X_2)| \leq 2^{n(1+\varepsilon)H(X_1,X_2)}$.
        \item For $n$ sufficiently large, $|\mathcal{T}^n_{\varepsilon}(X_1,X_2)| \geq (1-\varepsilon)2^{n(1-\varepsilon)H(X_1,X_2)}$.
        \item If $(x_1^n,x_2^n) \in \mathcal{T}^n_{\varepsilon}(X_1,X_2)$ then $x_1^n \in \mathcal{T}^n_{\varepsilon}(X_1)$ and $x_2^n \in \mathcal{T}^n_{\varepsilon}(X_2)$.
        \item $\forall y^n \in \mathcal{T}^n_{\varepsilon}(Y), \mathcal{T}^n_{\varepsilon}(X_1,X_2|y^n) \subseteq \mathcal{T}^n_{\varepsilon}(X_1,X_2)$.
        \item $\forall (x_1^n,x_2^n,y^n) \in \mathcal{T}^n_{\varepsilon}(X_1,X_2,Y), 2^{-n(1+\varepsilon)H(X_1,X_2|Y)} \leq P_{X_1^nX_2^nY^n}(x_1^n,x_2^n|y^n) \leq 2^{-n(1-\varepsilon)H(X_1,X_2|Y)}$.
        \item $\forall y^n \in \mathcal{T}^n_{\varepsilon}(Y), |\mathcal{T}^n_{\varepsilon}(X_1,X_2|y^n)| \leq 2^{n(1+\varepsilon)H(X_1,X_2|Y)}$.
        \item For $\varepsilon' < \varepsilon$ and $n$ sufficiently large, we get $\forall y^n \in \mathcal{T}^n_{\varepsilon'}(Y), |\mathcal{T}^n_{\varepsilon}(X_1,X_2|y^n)| \geq (1-\varepsilon)2^{n(1-\varepsilon)H(X_1,X_2|Y)}$.
        \end{enumerate}
      \end{lem}
      \begin{proof}
        We reproduce the proof of the last statement here to emphasize on the fact that there is an $n_0$ such that for all $n \geq n_0$ and for all $y^n \in \mathcal{T}^n_{\varepsilon'}(Y)$, the property holds.

        For any $\varepsilon > \varepsilon' > 0$, let us show that there exists $n$ such that we have:
        \[ \forall y^n \in \mathcal{T}^n_{\varepsilon'}(Y), \mathbb{P}\left((X_1^n,X_2^n,y^n) \in \mathcal{T}^n_{\varepsilon}(X_1,X_2,Y) \right) \geq 1-\varepsilon \ , \]
        
        where $X_1^n, X_2^n$ are drawn from the distribution $P_{X_1^nX_2^n|Y^n=y^n}$. This will imply the statement. Indeed, we have that:

        \begin{equation}
          \begin{aligned}
            \mathbb{P}\left((X_1^n,X_2^n,y^n) \in \mathcal{T}^n_{\varepsilon}(X_1,X_2,Y) \right) &= \sum_{(x_1^n,x_2^n) \in \mathcal{T}^n_{\varepsilon}(X_1,X_2|y^n)} P_{X_1^nX_2^n|Y^n}(x_1^n,x_2^n|y^n) \\
            &\leq |\mathcal{T}^n_{\varepsilon}(X_1,X_2|y^n)|2^{-n(1-\varepsilon)H(X_1,X_2|Y)} \ ,
          \end{aligned}
        \end{equation}
      
        since $P_{X_1^nX_2^n|Y^n}(x_1^n,x_2^n|y^n) \leq 2^{-n(1-\varepsilon)H(X_1,X_2|Y)}$ as $(x_1^n,x_2^n,y^n) \in \mathcal{T}^n_{\varepsilon}(X_1,X_2,Y)$. Thus, we have that $|\mathcal{T}^n_{\varepsilon}(X_1,X_2|y^n)| \geq (1-\varepsilon)2^{n(1-\varepsilon)H(X_1,X_2|Y)}$. In order to prove our result, we take the proof in Appendix 2A of~\cite{GK11}. We take $y^n \in \mathcal{T}^n_{\varepsilon'}(\mathcal{Y})$ and $(X_1^n,X_2^n) \sim P_{X_1^nX_2^n|Y^n}(x_1^n,x_2^n|y^n) = \prod_{i=1}^nP_{X_1X_2|Y}(x_{1,i},x_{2,i}|y_i)$. Applied to our choice of variables, we have the following result :

        \begin{equation}
          \begin{aligned}
            &\mathbb{P}\left((X_1^n,X_2^n,y^n) \notin \mathcal{T}^n_{\varepsilon}(X_1,X_2,Y) \right)\\
            &= \mathbb{P}\left(\exists (x_1,x_2,y) : |\pi(x_1,x_2,y|X_1^n,X_2^n,y^n) - P_{X_1X_2Y}(x_1,x_2,y)| > \varepsilon P_{X_1X_2Y}(x_1,x_2,y) \right) \\
            &\leq \sum_{x_1,x_2,y} \mathbb{P}\left(|\pi(x_1,x_2,y|X_1^n,X_2^n,y^n) - P_{X_1X_2Y}(x_1,x_2,y)| > \varepsilon P_{X_1X_2Y}(x_1,x_2,y) \right) \text{ by union bound,} \\
            &= \sum_{x_1,x_2,y} \mathbb{P}\left(\left|\frac{\pi(x_1,x_2,y|X_1^n,X_2^n,y^n)}{P_{X_1X_2Y}(x_1,x_2,y)} - 1\right| > \varepsilon \right) \\
            &= \sum_{x_1,x_2,y} \mathbb{P}\left(\left|\frac{\pi(x_1,x_2,y|X_1^n,X_2^n,y^n)}{P_{X_1X_2|Y}(x_1,x_2|y)\pi(y|y^n)}\frac{\pi(y|y^n)}{P_Y(y)} - 1\right| > \varepsilon \right) \\
            &\leq \sum_{x_1,x_2,y} \mathbb{P}\left(\frac{\pi(x_1,x_2,y|X_1^n,X_2^n,y^n)}{\pi(y|y^n)} > \frac{1+\varepsilon}{1+\varepsilon'}P_{X_1X_2|Y}(x_1,x_2|y) \right)\\
            &+ \sum_{x_1,x_2,y} \mathbb{P}\left(\frac{\pi(x_1,x_2,y|X_1^n,X_2^n,y^n)}{\pi(y|y^n)} < \frac{1-\varepsilon}{1-\varepsilon'}P_{X_1X_2|Y}(x_1,x_2|y) \right) \ ,
          \end{aligned}
        \end{equation}
        since $y^n \in  \mathcal{T}^n_{\varepsilon'}(\mathcal{Y})$ and thus $1-\varepsilon' \leq \frac{\pi(y|y^n)}{P_Y(y)} \leq 1+\varepsilon'$. However, since $\varepsilon' < \varepsilon$, we have $\frac{1+\varepsilon}{1+\varepsilon'} > 1$ and $\frac{1-\varepsilon}{1-\varepsilon'} < 1$. We will show that for all $x_1,x_2,y$ with $P_Y(y) > 0$, we have $\frac{\pi(x_1,x_2,y|X_1^n,X_2^n,y^n)}{\pi(y|y^n)} \underset{n \rightarrow +\infty}{\rightarrow} P_{X_1X_2|Y}(x_1,x_2|y)$ in probability, with a convergence rate independent from $y^n \in \mathcal{T}^n_{\varepsilon'}(Y)$, which will be enough to conclude the proof.

        Let us fix some $x_1,x_2,y$ with $P_Y(y) > 0$. Since $y^n \in \mathcal{T}^n_{\varepsilon'}(Y)$, we have in particular $(1-\varepsilon')P_Y(y) \leq \pi(y|y^n) \leq (1+\varepsilon')P_Y(y)$. Thus $N := |\{i:y_i=y\}| = n\pi(y|y^n) \geq (1-\varepsilon')P_Y(y)n$. Then we have:

        \[ \frac{\pi(x_1,x_2,y|X_1^n,X_2^n,y^n)}{\pi(y|y^n)} = \frac{1}{N}\sum_{i \in S}Z_i \text{ with }Z_i := \mathbbm{1}_{(X_{1,i},X_{2,i})=(x_1,x_2)} \text{ and } S := \{i:y_i=y\}\ . \]

        Thus, all $Z_i$ with $i \in S$ are independent and follow the same law:
        
          \[ Z_i := \begin{cases}
            1 & \text{ with probability } P_{X_1X_2|Y}(x_1,x_2|y) \\
            0 & \text{ otherwise}
          \end{cases}
          \]

          Furthermore, we have $\mathbb{E}[Z_i] = P_{X_1X_2|Y}(x_1,x_2|y)$, and all $Z_i$ have the same variance $\sigma_{x_1,x_2|y}^2 < +\infty$ (depending only on $X_1,X_2,Y,x_1,x_2,y$). Thus we can apply Chebyshev inequality:

          \[ \mathbb{P}\left( \left|\frac{1}{N}\sum_{i \in S} Z_i - P_{X_1X_2|Y}(x_1,x_2|y) \right| \geq \eta \right) \leq \frac{\sigma_{x_1,x_2|y}^2}{N\eta^2} \ .\]

          However, since $N \geq (1-\varepsilon')P_Y(y)n$, we get:

          \[ \mathbb{P}\left( \left|\frac{\pi(x_1,x_2,y|X_1^n,X_2^n,y^n)}{\pi(y|y^n)} - P_{X_1X_2|Y}(x_1,x_2|y) \right| \geq \eta \right) \leq \frac{\sigma_{x_1,x_2|y}^2}{\eta^2(1-\varepsilon')P_Y(y)n} \underset{n \rightarrow +\infty}{\rightarrow} 0 \ .\]

          Thus, we have $\frac{\pi(x_1,x_2,y|X_1^n,X_2^n,y^n)}{\pi(y|y^n)} \underset{n \rightarrow +\infty}{\rightarrow} P_{X_1X_2|Y}(x_1,x_2|y)$ in probability with a convergence rate independent from $y^n \in \mathcal{T}^n_{\varepsilon'}(Y)$.
      \end{proof}

  \begin{prop}[Achievability part of Theorem~\ref{theo:CharaNSrelaxed}]      
      \label{prop:AchievabilityNSrelaxed}
      If a rate pair is in the closure of the convex hull of $(R_1,R_2)$ satisfying:
      \[ R_1 < I(X_1:Y|X_2)\ ,\ R_2 < I(X_2:Y|X_1)\ ,\ R_1+R_2 < I((X_1,X_2):Y) \ ,\]
      for $(X_1,X_2)$ following some law $P_{X_1X_2}$ on $\mathcal{X}_1 \times \mathcal{X}_2$, and $Y \in \mathcal{Y}$ the outcome of $W$ on inputs $X_1,X_2$, then it is in $\mathcal{C}^{\overline{\mathrm{NS}}}(W)$.
  \end{prop}
  \begin{proof}
    Let us fix $\varepsilon,\varepsilon' \in (0,1)$ such that $\varepsilon' < \varepsilon \leq \frac{1}{2}$. Let $n \in \mathbb{N}$ which will be chosen large enough during the proof.

    We consider $n$ independent random variables $(X_{1,i}X_{2,i}Y_i) \sim P_{X_1X_2Y}$, with $P_{X_1X_2Y}(x_{1,i},x_{2,i},y_i) = W(y_i|x_{1,i}x_{2,i})P_{X_1X_2}(x_{1,i},x_{2,i})$. We call $P_{X_1^nX_2^nY^n} $ the law of their product. We have then $P_{X_1^nX_2^n}(x_1^n,x_2^n) := \prod_{i=1}^n P_{X_1X_2}(x_{1,i},x_{2,i})$. If $\hat{Y}$ is the output of $W^{\otimes n}$ on $X_1^nX_2^n$, we have that:
      \begin{equation}
        \begin{aligned}
          P_{X_1^nX_2^n\hat{Y}}(x_1^n,x_2^n,y^n) &= W^{\otimes n}(y^n|x_1^nx_2^n)P_{X_1^nX_2^n}(x_1^n,x_2^n) = W^{\otimes n}(y^n|x_1^nx_2^n)\prod_{i=1}^n P_{X_1X_2}(x_{1,i},x_{2,i})\\
          &= \prod_{i=1}^n W(y_i|x_{1,i}x_{2,i})P_{X_1X_2}(x_{1,i},x_{2,i}) = \prod_{i=1}^n P_{X_1X_2Y}(x_{1,i},x_{2,i},y_i) \ .
        \end{aligned}
      \end{equation}

      Thus, $\hat{Y}$ follows the product law of $Y_i$, i.e.  $\hat{Y} = Y^n$.
      
  Let us consider $C_1,C_2,C_3$ some positive numbers independent from $n$ and $\varepsilon$ which we will define later,  $k_1 = 2^{nR_1}$, $k_2 = 2^{nR_2}$ integers with $(R_1,R_2)$ positive rates such that:
  \begin{equation}
      \begin{aligned}
        R_1 &\leq I(X_1:Y|X_2) - \frac{1}{n} - C_1\varepsilon \ ,\\
        R_2 &\leq I(X_2:Y|X_1) - \frac{1}{n} - C_2\varepsilon \ ,\\
        R_1+R_2 &\leq I((X_1,X_2):Y) - \frac{1}{n} - C_3\varepsilon \ .
      \end{aligned}
    \end{equation}
  We define a solution of $\mathrm{S}^{\overline{\mathrm{NS}}}(W^{\otimes n},2^{nR_1},2^{nR_2})$ in the following way:
      
  \[ p_{x_1^n,x_2^n} := \begin{cases}
    \frac{2^{n(R_1+R_2)}P_{X_1^nX_2^n}(x_1^n,x_2^n)}{\sum_{(x_1^n,x_2^n) \in \mathcal{T}^n_{\varepsilon}(X_1,X_2)}P_{X_1^nX_2^n}(x_1^n,x_2^n)} & \text{ if } (x_1^n,x_2^n) \in \mathcal{T}^n_{\varepsilon}(X_1,X_2) \\
    0 & \text{ otherwise}
  \end{cases}
  \]

  and

  \[ r_{x_1^n,x_2^n,y^n} := \begin{cases}
    p_{x_1^n,x_2^n} & \text{ if } (x_1^n,x_2^n,y^n) \in \mathcal{T}^n_{\varepsilon'}(X_1,X_2,Y) \\
    0 & \text{ otherwise}
  \end{cases}
    \]

    By construction, the constraint $0 \leq r_{x_1^n,x_2^n,y^n} \leq p_{x_1^n,x_2^n}$ is satisfied. We have also that:
    \[\sum_{x_1^n,x_2^n} p_{x_1^n,x_2^n} = \sum_{(x_1^n,x_2^n) \in \mathcal{T}^n_{\varepsilon}(X_1,X_2)} \frac{2^{n(R_1+R_2)}P_{X_1^nX_2^n}(x_1^n,x_2^n)}{\sum_{(x_1^n,x_2^n) \in \mathcal{T}^n_{\varepsilon}(X_1,X_2)}P_{X_1^nX_2^n}(x_1^n,x_2^n)} = 2^{n(R_1+R_2)} = k_1k_2 \ . \]

    If $(x_1^n,x_2^n,y^n) \in \mathcal{T}^n_{\varepsilon'}(X_1,X_2,Y)$, we have that $(x_1^n,x_2^n) \in \mathcal{T}^n_{\varepsilon'}(X_1,X_2) \subseteq \mathcal{T}^n_{\varepsilon}(X_1,X_2)$, so in that case:
    \[ r_{x_1^n,x_2^n,y^n} = \frac{2^{n(R_1+R_2)}P_{X_1^nX_2^n}(x_1^n,x_2^n)}{\sum_{(x_1^n,x_2^n) \in \mathcal{T}^n_{\varepsilon}(X_1,X_2)}P_{X_1^nX_2^n}(x_1^n,x_2^n)} \ . \]

    If $y^n \not\in \mathcal{T}^n_{\varepsilon'}(Y)$, then for all $(x_1^n,x_2^n)$, $(x_1^n,x_2^n,y^n) \notin \mathcal{T}^n_{\varepsilon'}(X_1,X_2,Y)$, so $\sum_{x_1^n,x_2^n} r_{x_1^n,x_2^n,y^n}  = 0 \leq 1$ in that case.

    Otherwise, if $y^n \in \mathcal{T}^n_{\varepsilon'}(Y)$, then:
    \begin{equation}
      \begin{aligned}
        \sum_{x_1^n,x_2^n} r_{x_1^n,x_2^n,y^n} &= 2^{n(R_1+R_2)} \frac{\sum_{(x_1^n,x_2^n) \in \mathcal{T}^n_{\varepsilon'}(X_1,X_2|y^n)} P_{X_1^nX_2^n}(x_1^n,x_2^n)}{\sum_{(x_1^n,x_2^n) \in \mathcal{T}^n_{\varepsilon}(X_1,X_2)}P_{X_1^nX_2^n}(x_1^n,x_2^n)}\\
        &\leq 2^{n(R_1+R_2)} \frac{\sum_{(x_1^n,x_2^n) \in \mathcal{T}^n_{\varepsilon}(X_1,X_2|y^n)} P_{X_1^nX_2^n}(x_1^n,x_2^n)}{\sum_{(x_1^n,x_2^n) \in \mathcal{T}^n_{\varepsilon}(X_1,X_2)}P_{X_1^nX_2^n}(x_1^n,x_2^n)}\\
        &\leq 2^{n(R_1+R_2)}\frac{2^{-n(1-\varepsilon)H(X_1,X_2)}}{2^{-n(1+\varepsilon)H(X_1,X_2)}} \frac{|\mathcal{T}^n_{\varepsilon}(X_1,X_2|y^n)|}{|\mathcal{T}^n_{\varepsilon}(X_1,X_2)|} \text{ since $(x_1^n,x_2^n) \in \mathcal{T}^n_{\varepsilon}(X_1,X_2)$} \\
        &= 2^{n(R_1+R_2 + 2\varepsilon H(X_1,X_2))}\frac{|\mathcal{T}^n_{\varepsilon}(X_1,X_2|y^n)|}{|\mathcal{T}^n_{\varepsilon}(X_1,X_2)|} \ .
      \end{aligned}
    \end{equation}

    But $|\mathcal{T}^n_{\varepsilon}(X_1,X_2|y^n)| \leq 2^{n(1+\varepsilon)H(X_1,X_2|Y)}$ and for a large enough $n$ we have that $|\mathcal{T}^n_{\varepsilon}(X_1,X_2)| \geq (1-\varepsilon)2^{n(1-\varepsilon)H(X_1,X_2)} \geq 2^{n\left((1-\varepsilon)H(X_1,X_2) - \frac{1}{n}\right)}$, so in that case:

    \[ \sum_{x_1^n,x_2^n} r_{x_1^n,x_2^n,y^n} \leq 2^{n(R_1+R_2 + 2\varepsilon H(X_1,X_2))}\frac{2^{n(1+\varepsilon)H(X_1,X_2|Y)}}{2^{n\left((1-\varepsilon)H(X_1,X_2) - \frac{1}{n}\right)}} = 2^{n\left(R_1+R_2 - I(X_1,X_2:Y) + \frac{1}{n} + C_3\varepsilon\right)} \leq 1 \ ,\]

    since $I(X_1,X_2:Y) = H(X_1,X_2) - H(X_1,X_2|Y)$ and $R_1+R_2 \leq I(X_1,X_2:Y) - \frac{1}{n} - C_3\varepsilon$, with $C_3 := H(X_1,X_2|Y) + 3H(X_1,X_2)$.

    Let us focus on the constraint $\sum_{x_1^n} p_{x_1^n,x_2^n} \geq k_1 \sum_{x_1^n} r_{x_1^n,x_2^n,y^n}$ (the symmetric constraint $\sum_{x_2^n} p_{x_1^n,x_2^n} \geq k_2 \sum_{x_2^n} r_{x_1^n,x_2^n,y^n}$ will be achieved for symmetric reasons).

    Let us fix $(x_2^n,y^n)$. If $(x_2^n,y^n) \not\in \mathcal{T}^n_{\varepsilon'}(X_2,Y)$, then for all $x_1^n$, $(x_1^n,x_2^n,y^n) \not\in \mathcal{T}^n_{\varepsilon'}(X_1,X_2,Y)$, thus $r_{x_1^n,x_2^n,y^n} = 0$ and the constraint is fulfilled. Let us assume that $(x_2^n,y^n) \in \mathcal{T}^n_{\varepsilon'}(X_2,Y)$. Since $r_{x_1^n,x_2^n,y^n} > 0$ implies that $(x_1^n,x_2^n,y^n) \in \mathcal{T}^n_{\varepsilon'}(X_1,X_2,Y)$, we have that:

    \[ \sum_{x_1^n} r_{x_1^n,x_2^n,y^n} = \sum_{x_1^n \in \mathcal{T}^n_{\varepsilon'}(X_1|x_2^n,y^n)} r_{x_1^n,x_2^n,y^n} = \sum_{x_1^n \in \mathcal{T}^n_{\varepsilon'}(X_1|x_2^n,y^n)} p_{x_1^n,x_2^n} \ .\]
    
    Thus:
     \begin{equation}
      \begin{aligned}
        \frac{\sum_{x_1^n} p_{x_1^n,x_2^n}}{k_1\sum_{x_1^n} r_{x_1^n,x_2^n,y^n}} &\geq \frac{1}{k_1} \frac{ \sum_{x_1 \in \mathcal{T}^n_{\varepsilon}(X_1|x_2^n)}P_{X_1^nX_2^n}(x_1^n,x_2^n)}{\sum_{x_1 \in \mathcal{T}^n_{\varepsilon'}(X_1|x_2^n,y^n)} P_{X_1^nX_2^n}(x_1^n,x_2^n)} \geq \frac{1}{k_1} \frac{ \sum_{x_1 \in \mathcal{T}^n_{\varepsilon}(X_1|x_2^n)}P_{X_1^nX_2^n}(x_1^n,x_2^n)}{\sum_{x_1 \in \mathcal{T}^n_{\varepsilon}(X_1|x_2^n,y^n)} P_{X_1^nX_2^n}(x_1^n,x_2^n)}\\
        &\geq \frac{1}{k_1}\frac{2^{-n(1+\varepsilon)H(X_1,X_2)}}{2^{-n(1-\varepsilon)H(X_1,X_2)}}\frac{|\mathcal{T}^n_{\varepsilon}(X_1|x_2^n)|}{|\mathcal{T}^n_{\varepsilon}(X_1|x_2^n,y^n)|} \geq 2^{n(-R_1-2\varepsilon H(X_1,X_2))}\frac{|\mathcal{T}^n_{\varepsilon}(X_1|x_2^n)|}{|\mathcal{T}^n_{\varepsilon}(X_1|x_2^n,y^n)|}\ .
      \end{aligned}
    \end{equation}

     But$|\mathcal{T}^n_{\varepsilon}(X_1|x_2^n,y^n)| \leq 2^{n(1+\varepsilon)H(X_1|X_2Y)}$ and for a large enough $n$ we have $\forall x_2^n \in \mathcal{T}^n_{\varepsilon'}(X_2), |\mathcal{T}^n_{\varepsilon}(X_1|x_2^n)| \geq (1-\varepsilon)2^{n(1-\varepsilon)H(X_1|X_2)} \geq 2^{n\left((1-\varepsilon)H(X_1|X_2)-\frac{1}{n}\right)}$, so we get with $C_1 := 2 H(X_1,X_2) + H(X_1|X_2Y) + H(X_1|X_2)$ (symmetrically $C_2 := 2 H(X_1,X_2) + H(X_2|X_1Y) + H(X_2|X_1)$):

     \[ \frac{\sum_{x_1^n} p_{x_1^n,x_2^n}}{k_1\sum_{x_1^n} r_{x_1^n,x_2^n,y^n}} \geq 2^{n\left(H(X_1|X_2) - \frac{1}{n} -H(X_1|X_2Y)-R_1-C_1\varepsilon\right)} = 2^{n\left(I(X_1: Y | X_2)-\frac{1}{n}-C_1\varepsilon-R_1\right)} \geq 1 \ .\]
     
    For a large enough $n$, all constraints are satisfied, thus $(p_{x_1^n,x_2^n},r_{x_1^n,x_2^n,y^n})$ is a valid solution. Then:

    \begin{equation}
      \begin{aligned}
        &\mathrm{S}^{\overline{\mathrm{NS}}}(W^{\otimes n},2^{nR_1},2^{nR_2}) \geq \frac{1}{2^{n(R_1+R_2)}} \sum_{(x_1^n,x_2^n,y^n) \in \mathcal{T}^n_{\varepsilon'}(X_1,X_2,Y)} W^{\otimes n}(y^n|x_1^nx_2^n)r_{x_1^n,x_2^n,y^n}\\
        &= \frac{1}{2^{n(R_1+R_2)}} \sum_{(x_1^n,x_2^n,y^n) \in \mathcal{T}^n_{\varepsilon'}(X_1,X_2,Y)} \frac{P_{X_1^nX_2^nY^n}(x_1^n,x_2^n,y^n)}{P_{X_1^nX_2^n}(x_1^n,x_2^n)}\frac{2^{n(R_1+R_2)}P_{X_1^nX_2^n}(x_1^n,x_2^n)}{\sum_{(x_1^n,x_2^n) \in \mathcal{T}^n_{\varepsilon}(X_1,X_2)}P_{X_1^nX_2^n}(x_1^n,x_2^n)}\\
        &= \frac{\sum_{(x_1^n,x_2^n,y^n) \in \mathcal{T}^n_{\varepsilon'}(X_1,X_2,Y)} P_{X_1^nX_2^nY^n}(x_1^n,x_2^n,y^n)}{\sum_{(x_1^n,x_2^n) \in \mathcal{T}^n_{\varepsilon}(X_1,X_2)}P_{X_1^nX_2^n}(x_1^n,x_2^n) } \underset{n \rightarrow +\infty}{\rightarrow} 1\ ,
      \end{aligned}
    \end{equation}

    since typical sets cover asymptotically the whole probability mass. Thus, since $\mathrm{S}^{\overline{\mathrm{NS}}}(W^{\otimes n},2^{nR_1},2^{nR_2}) \leq 1$, we get that $\mathrm{S}^{\overline{\mathrm{NS}}}(W^{\otimes n},2^{nR_1},2^{nR_2}) \underset{n \rightarrow +\infty}{\rightarrow} 1$. Thus for sufficiently large $n$ we can achieve a rate pair arbitrarily close to the outer bound. Finally, since $\mathcal{C}^{\overline{\mathrm{NS}}}(W)$ is closed and convex, a rate pair that is in the closure of the convex hull of the initial region is also in $\mathcal{C}^{\overline{\mathrm{NS}}}(W)$, and thus the proof is completed.
  \end{proof}

  \section{Independent Non-Signaling Assisted Capacity Region}
  \label{section:NSsr}
The goal of this section is to show that independent non-signaling assistance does not change the capacity region of a MAC $W$, i.e. that $\mathcal{C}^{\mathrm{NS}_{\mathrm{SR}}}(W)=\mathcal{C}(W)$. In order to prove this result, we will need some properties in the one-sender one-receiver case from~\cite{BF18}. Specifically, let us first recall the definition of the maximum success probability $\mathrm{S}(W,k)$ of transmitting $k$ messages using the channel $W$:
\begin{equation}
  \begin{aligned}
    \mathrm{S}(W,k) := &&\underset{e,d}{\maxi} &&& \frac{1}{k} \sum_{i,x,y} W(y|x)e(x|i)d(i|y)\\
    &&\st &&& \sum_{x \in \mathcal{X}} e(x|i) = 1, \forall i \in [k]\\
    &&&&& \sum_{j \in [k]} d(j|y) = 1, \forall y \in \mathcal{Y}\\
    &&&&& e(x|i), d(j|y) \geq 0
  \end{aligned}
\end{equation}

Then, the following characterization of $\mathrm{S}(W,k)$ can be derived:

\begin{prop}[Proposition 2.1 of~\cite{BF18}]
  $\mathrm{S}(W,k) = \frac{1}{k} \underset{S \subseteq X: |S| \leq k}{\max} f_W(S)$ with $f_W(S) := \sum_{y \in Y} \max_{x \in S} W(y|x)$.
\end{prop}

As in the MAC scenario, one can consider non-signaling assistance shared between the sender and the receiver, which leads to the following maximum success probability $\mathrm{S}^{\mathrm{NS}}(W,k)$:

\begin{equation}
  \begin{aligned}
    \mathrm{S}^{\mathrm{NS}}(W,k) := &&\underset{P}{\maxi} &&& \frac{1}{k} \sum_{i,x,y} W(y|x)P(xi|iy)\\
    &&\st &&& \sum_{x} P(xj|iy) = \sum_{x} P(xj|i'y)\\
    &&&&& \sum_{j} P(xj|iy) = \sum_{j} P(xj|iy')\\
    &&&&& \sum_{x,j} P(xj|iy) = 1\\
    &&&&& P(xj|iy) \geq 0
  \end{aligned}
\end{equation}

A symmetrization can also be done to simplify the expression of the linear program defining $\mathrm{S}^{\mathrm{NS}}(W,k)$:

\begin{prop}[Appendix A of~\cite{BF18}]
  \label{prop:NSonewayLP}
  \begin{equation}
    \begin{aligned}
      \mathrm{S}^{\mathrm{NS}}(W,k) = &&\underset{r,p}{\maxi} &&& \frac{1}{k} \sum_{x,y} W(y|x)r_{x,y}\\
      &&\st &&& \sum_{x} r_{x,y} = 1\\
      &&&&& \sum_{x} p_{x} = k\\
      &&&&& 0 \leq r_{x,y} \leq p_{x}
    \end{aligned}
  \end{equation}
\end{prop}

Finally, the main tool we will use from~\cite{BF18} is the following random coding technique, which describes how to find a classical code with a success probability close to the non-signaling assisted one:
\begin{theo}[Theorem 3.1 of~\cite{BF18}]
  \label{theo:RandomCoding}
  Given a solution $r,p$ of the program computing $\mathrm{S}^{\mathrm{NS}}(W,k)$, we have that:
  \[ \mathbb{E}_S\left[\frac{f_W(S)}{\ell}\right] \geq \frac{k}{\ell}\left(1-\left(1-\frac{1}{k}\right)^{\ell}\right) \cdot \frac{1}{k} \sum_{x,y} W(y|x)r_{x,y} , \]
  for the multiset $S$ obtained by choosing $\ell$ elements of $\mathcal{X}$ independently according to the distribution $\left(\frac{p_{x}}{k}\right)_{x \in \mathcal{X}}$.
\end{theo}

We can now state our result on independent non-signaling assistance, which says that even in one-shot scenarios, the success probability with and without that assistance are close:

\begin{theo}
  \label{theo:NSsr}
  For any $\ell_1,k_1,\ell_2,k_2$:
  \[ \min\left(\frac{k_1}{\ell_1}\left(1-\left(1-\frac{1}{k_1}\right)^{\ell_1}\right),\frac{k_2}{\ell_2}\left(1-\left(1-\frac{1}{k_2}\right)^{\ell_2}\right)\right)\mathrm{S}_{\text{sum}}^{\mathrm{NS}_{\mathrm{SR}}}(W,k_1,k_2) \leq \mathrm{S}_{\text{sum}}(W,\ell_1,\ell_2)\ . \]
\end{theo}

In particular, this will imply that the capacity regions are the same:

\begin{cor}
  \label{cor:NSsr}
  $\mathcal{C}^{\mathrm{NS}_{\mathrm{SR}}}(W)=\mathcal{C}(W)$.
\end{cor}
\begin{proof}
  We will show that $\mathcal{C}_{\text{sum}}^{\mathrm{NS}_{\mathrm{SR}}}(W)=\mathcal{C}_{\text{sum}}(W)$, which is enough to conclude thanks to Proposition~\ref{prop:CapaSumJoint} and Proposition~\ref{prop:NSCapaSumJoint}. We apply Theorem~\ref{theo:NSsr} on the MAC $W^{\otimes n}$.

  Let us fix $k_1=2^{nR_1},k_2=2^{nR_2}$ and $\ell_1=\frac{2^{nR_1}}{n},\ell_2=\frac{2^{nR_2}}{n}$. Since:
  \[ \frac{k}{\ell}\left(1-\left(1-\frac{1}{k}\right)^{\ell}\right) \geq \frac{k}{\ell}\left(1-e^{-\frac{\ell}{k}}\right) \geq 1 - \frac{\ell}{2k} \ ,\]

  and $1 - \frac{\ell_1}{2k_1} = 1 - \frac{\ell_2}{2k_2} = 1 - \frac{1}{2n}$, we get:
\[ \left(1-\frac{1}{2n}\right)\mathrm{S}_{\text{sum}}^{\mathrm{NS}_{\mathrm{SR}}}(W^{\otimes n},2^{nR_1},2^{nR_1}) \leq \mathrm{S}_{\text{sum}}\left(W^{\otimes n},\frac{2^{nR_1}}{n},\frac{2^{nR_2}}{n}\right)\ . \]

  As $1-\frac{1}{2n}$ tends to $1$ when $n$ tends to infinity, we get that $\forall \varepsilon > 0, \exists N \in \mathbb{N}, \forall n\geq N$:

  \[(1-\varepsilon)\mathrm{S}_{\text{sum}}^{\mathrm{NS}_{\mathrm{SR}}}(W^{\otimes n},2^{nR_1},2^{nR_1}) \leq \mathrm{S}_{\text{sum}}(W^{\otimes n},2^{n(R_1-\frac{\log(n)}{n})},2^{n(R_2-\frac{\log(n)}{n})}) \ . \]

Thus, if $\underset{n \rightarrow +\infty}{\lim} \mathrm{S}_{\text{sum}}^{\mathrm{NS}_{\mathrm{SR}}}(W^{\otimes n},2^{nR_1},2^{nR_1})  = 1$, we have that for all $R_1'<R_1$ and $R_2'<R_2$:

\[ \underset{n \rightarrow +\infty}{\lim} \mathrm{S}_{\text{sum}}(W^{\otimes n},2^{nR_1'},2^{nR_1'}) \geq 1-\varepsilon \ . \]

Since this is true for all $\varepsilon > 0$, we get in fact that $\underset{n \rightarrow +\infty}{\lim} \mathrm{S}_{\text{sum}}(W^{\otimes n},2^{nR_1'},2^{nR_1'}) = 1$. This implies that $\mathcal{C}_{\text{sum}}^{\mathrm{NS}_{\mathrm{SR}}}(W) \subseteq \mathcal{C}_{\text{sum}}(W)$, and thus that the capacity regions are equal as the other inclusion is always satisfied.
\end{proof}

In order to prove Theorem~\ref{theo:NSsr}, we will need the following lemma:

\begin{lem}
  \label{lem:DoubleOneWay}
  If $S_1,S_2$ are classical codes (i.e. multisets with elements in $\mathcal{X}_1,\mathcal{X}_2$) of size $\ell_1,\ell_2$:
  \[ \mathrm{S}_{\text{sum}}(W,\ell_1,\ell_2) \geq \frac{1}{2}\left(\frac{f_{W^1_{S_2,\ell_2}}(S_1)}{\ell_1} + \frac{f_{W^2_{S_1,\ell_1}}(S_2)}{\ell_2}\right) \ ,\]
  where $W^1_{S_2, \ell_2}$ is the channel defined by $W^1_{S_2, \ell_2}(y|x_1) = \frac{1}{\ell_2} \sum_{i_2 = 1}^{\ell_2} W(y|x_1S_2^{i_2})$ and similarly for $W^2_{S_1, \ell_1}$.
\end{lem}
\begin{proof}
  Let us define $e_1(x_1|i_1) := \mathbbm{1}_{S_1^{i_1}=x_1}$ and $e_2(x_2|i_2) := \mathbbm{1}_{S_2^{i_2}=x_2}$. Then for fixed $y$, let us take $j_1^y \in \text{argmax}_{i_1} \{\sum_{i_2 = 1}^{\ell_2}W(y|S_1^{i_1}S_2^{i_2})\}$, $j_2^y \in \text{argmax}_{i_2}\{\sum_{i_1 = 1}^{\ell_1}W(y|S_1^{i_1}S_2^{i_2})\}$ and then define $d(j_1|y) := \mathbbm{1}_{j_1=j_1^y}$, $d(j_2|y) := \mathbbm{1}_{j_2=j_2^y}$. We have then:

  \begin{equation}
    \begin{aligned}
      &\mathrm{S}_{\text{sum}}(W,\ell_1,\ell_2) \geq \frac{1}{\ell_1\ell_2}\sum_{i_1,i_2,x_1,x_2,y}W(y|x_1x_2)\mathbbm{1}_{S_1^{i_1}=x_1}\mathbbm{1}_{S_2^{i_2}=x_2}\frac{\mathbbm{1}_{i_1=i_1^y}+\mathbbm{1}_{i_2=i_2^y}}{2}\\
      &= \frac{1}{\ell_1\ell_2}\sum_{i_1,i_2,y}W(y|S_1^{i_1}S_2^{i_2})\frac{\mathbbm{1}_{i_1=i_1^y}+\mathbbm{1}_{i_2=i_2^y}}{2}
      = \frac{1}{\ell_1\ell_2}\sum_{y}\frac{1}{2}\left(\sum_{i_2}W(y|S_1^{i_1^y}S_2^{i_2}) + \sum_{i_1}W(y|S_1^{i_1}S_2^{i_2^y})\right)\\
      &= \frac{1}{\ell_1\ell_2}\sum_{y}\frac{1}{2}\left(\max_{i_1}\sum_{i_2}W(y|S_1^{i_1}S_2^{i_2}) + \max_{i_2}\sum_{i_1}W(y|S_1^{i_1}S_2^{i_2})\right)\\
      &= \frac{1}{2}\left(\frac{\sum_{y}\max_{i_1}\left[\frac{1}{\ell_2}\sum_{i_2} W(y|S_1^{i_1}S_2^{i_2})\right]}{\ell_1} + \frac{\sum_{y}\max_{i_2}\left[\frac{1}{\ell_1}\sum_{i_1} W(y|S_1^{i_1}S_2^{i_2})\right]}{\ell_2} \right)\\
      &= \frac{1}{2}\left(\frac{\sum_{y}\max_{i_1}W^1_{S_2,\ell_2}(y|S_1^{i_1})}{\ell_1} + \frac{\sum_{y}\max_{i_2}W^2_{S_1,\ell_1}(y|S_2^{i_2})}{\ell_2} \right)
      = \frac{1}{2}\left(\frac{f_{W^1_{S_2,\ell_2}}(S_1)}{\ell_1} + \frac{f_{W^2_{S_1,\ell_1}}(S_2)}{\ell_2}\right) \ .
    \end{aligned}
  \end{equation}
\end{proof}

We have now all the tools to prove Theorem~\ref{theo:NSsr}:

\begin{proof}[Proof of Theorem~\ref{theo:NSsr}]
  Let us consider an optimal solution $r^1,r^2,p^1,p^2$ of the program of Proposition~\ref{prop:NSsrProgram} computing $\mathrm{S}_{\text{sum}}^{\mathrm{NS}_{\mathrm{SR}}}(W,k_1,k_2)$.

  Let us fix some multiset $S_2$ with elements in $\mathcal{X}_2$ of size $\ell_2$. Note that $r^1$ and $p^1$ are a feasible solution of the program of Proposition~\ref{prop:NSonewayLP} computing $\mathrm{S}^{\mathrm{NS}}(W^1_{S_2,\ell_2},k_1)$. As a result, we can apply Theorem \ref{theo:RandomCoding} and get the following statement. For the multiset $S_1$ obtained by choosing $\ell_1$ elements of $\mathcal{X}_1$ independently according to the distribution $\left(\frac{p^1_{x_1}}{k_1}\right)_{x_1 \in \mathcal{X}_1}$, we have:

  \[ \mathbb{E}_{S_1}\left[ \frac{f_{W^1_{S_2,\ell_2}}(S_1)}{\ell_1}\right] \geq \frac{k_1}{\ell_1}\left(1-\left(1-\frac{1}{k_1}\right)^{\ell_1}\right) \cdot \frac{1}{k_1}\sum_{x_1,y} W^1_{S_2,\ell_2}(y|x_1)r^1_{x_1,y} \ .\]

  Now, let $S_2$ be the multiset obtained by choosing $\ell_2$ elements of $\mathcal{X}_2$ independently according to the distribution $\left(\frac{p^2_{x_2}}{k_2}\right)_{x_2 \in \mathcal{X}_2}$. We have:
  \begin{equation}
    \begin{aligned}
      &\mathbb{E}_{S_2}\left[\frac{1}{k_1}\sum_{x_1,y}W^1_{S_2,\ell_2}(y|x_1)r^1_{x_1,y}\right] = \mathbb{E}_{S_2}\left[\frac{1}{k_1}\sum_{x_1,y}\frac{1}{\ell_2}\sum_{i_2=1}^{\ell_2}W(y|x_1S_2^{i_2})r^1_{x_1,y}\right]\\
      &= \frac{1}{\ell_2}\sum_{i_2=1}^{\ell_2}\mathbb{E}_{X^{i_2}_2 \sim \frac{p^2_{x_2}}{k_2}}\left[\frac{1}{k_1}\sum_{x_1,y}W(y|x_1X^{i_2}_2)r^1_{x_1,y}\right] = \mathbb{E}_{X_2 \sim \frac{p^2_{x_2}}{k_2}}\left[\frac{1}{k_1}\sum_{x_1,y}W(y|x_1X_2)r^1_{x_1,y}\right]\\
      &= \frac{1}{k_1}\sum_{x_1,x_2,y}\frac{p^2_{x_2}}{k_2}W(y|x_1x_2)r^1_{x_1,y} = \frac{1}{k_1}\sum_{x_1,y}W^1_{p^2,k_2}(y|x_1)r^1_{x_1,y} \ .
    \end{aligned}
  \end{equation}

  Thus in all, we have:

  \begin{equation}
    \begin{aligned}
      \mathbb{E}_{S_2}\left[\mathbb{E}_{S_1}\left[ \frac{f_{W^1_{S_2,\ell_2}}(S_1)}{\ell_1}\right]\right] &\geq  \mathbb{E}_{S_2}\left[\frac{k_1}{\ell_1}\left(1-\left(1-\frac{1}{k_1}\right)^{\ell_1}\right) \cdot \frac{1}{k_1}\sum_{x_1,y} W^1_{S_2,\ell_2}(y|x_1)r^1_{x_1,y}\right]\\
      &=\frac{k_1}{\ell_1}\left(1-\left(1-\frac{1}{k_1}\right)^{\ell_1}\right) \cdot \mathbb{E}_{S_2}\left[\frac{1}{k_1}\sum_{x_1,y} W^1_{S_2,\ell_2}(y|x_1)r^1_{x_1,y}\right]\\
      &\geq \frac{k_1}{\ell_1}\left(1-\left(1-\frac{1}{k_1}\right)^{\ell_1}\right) \cdot \frac{1}{k_1}\sum_{x_1,y} W^1_{p^2,k_2}(y|x_1)r^1_{x_1,y} \ ,
    \end{aligned}
  \end{equation}

  and symmetrically for $\mathbb{E}_{S_1}\left[\mathbb{E}_{S_2}\left[ \frac{f_{W^2_{S_1,\ell_1}}(S_2)}{\ell_2}\right]\right]$. Since there exists classical codes $S_1^*,S_2^*$ such that:

  \[ \frac{1}{2}\left(\frac{f_{W^1_{S_2^*,\ell_2}}(S_1^*)}{\ell_1} + \frac{f_{W^2_{S_1^*,\ell_1}}(S_2^*)}{\ell_2}\right) \geq \mathbb{E}_{S_1,S_2}\left[\frac{1}{2}\left(\frac{f_{W^1_{S_2,\ell_2}}(S_1)}{\ell_1} + \frac{f_{W^2_{S_1,\ell_1}}(S_2)}{\ell_2}\right)\right]\ , \]
  
  by applying Lemma~\ref{lem:DoubleOneWay}, we get:
  
  \begin{equation}
    \begin{aligned}
      &\mathrm{S}_{\text{sum}}(W,\ell_1,\ell_2) \geq \frac{1}{2}\left(\frac{f_{W^1_{S_2^*,\ell_2}}(S_1^*)}{\ell_1} + \frac{f_{W^2_{S_1^*,\ell_1}}(S_2^*)}{\ell_2}\right) \geq  \mathbb{E}_{S_1,S_2}\left[\frac{1}{2}\left(\frac{f_{W^1_{S_2,\ell_2}}(S_1)}{\ell_1} + \frac{f_{W^2_{S_1,\ell_1}}(S_2)}{\ell_2}\right)\right]\\
      &= \frac{1}{2} \left(\mathbb{E}_{S_2}\left[\mathbb{E}_{S_1}\left[ \frac{f_{W^1_{S_2,\ell_2}}(S_1)}{\ell_1}\right]\right] + \mathbb{E}_{S_1}\left[\mathbb{E}_{S_2}\left[ \frac{f_{W^2_{S_1,\ell_1}}(S_2)}{\ell_2}\right]\right]\right)\\
      &\geq \frac{1}{2}\left(\frac{k_1}{\ell_1}\left(1-\left(1-\frac{1}{k_1}\right)^{\ell_1}\right) \cdot \frac{1}{k_1}\sum_{x_1,y} W^1_{p^2}(y|x_1)r^1_{x_1,y} + \frac{k_2}{\ell_2}\left(1-\left(1-\frac{1}{k_2}\right)^{\ell_2}\right) \cdot \frac{1}{k_2}\sum_{x_2,y}W^2_{p^1}(y|x_2)r^2_{x_2,y} \right)\\
      &\geq  \min\left(\frac{k_1}{\ell_1}\left(1-\left(1-\frac{1}{k_1}\right)^{\ell_1}\right),\frac{k_2}{\ell_2}\left(1-\left(1-\frac{1}{k_2}\right)^{\ell_2}\right)\right)\mathrm{S}_{\text{sum}}^{\mathrm{NS}_{\mathrm{SR}}}(W,k_1,k_2) \ .
    \end{aligned}
  \end{equation}
\end{proof}

\begin{rk}
  In the whole proof of Theorem~\ref{theo:NSsr}, as well as the properties it depends on, we have never used the fact that the output of the channel $y$ was the same for both decoders $d_1$ and $d_2$. This implies that the result also holds for interference channels, i.e. two-sender two-receiver channels $W(y_1y_2|x_2x_2)$. Specifically, non-signaling assistance shared between the first sender and the first receiver and independently shared between the second sender and the second receiver does not change the capacity region of interference channels.
\end{rk}
  
\section{Conclusion}
In this work, we have studied the impact of non-signaling assistance on the capacity of multiple-access channels. We have developed an efficient linear program computing the success probability of the best non-signaling assisted code for a finite number of copies of a multiple-access channel. In particular, this gives lower bounds on the zero-error non-signaling assisted capacity of multiple-access channels.
Applied to the binary adder channel, these results were used to prove that a sum-rate of $\frac{\log_2(72)}{4} \simeq 1.5425$ can be reached with zero error, which beats the maximum classical sum-rate capacity of $\frac{3}{2}$. For noisy channels, we have developed a technique giving lower bounds through the use of concatenated codes. Applied to the noisy binary adder channel, this technique was used to show that non-signaling assistance still improves the sum-rate capacity. We have also found an outer bound on the non-signaling assisted capacity region through a relaxed notion of non-signaling assistance, whose capacity region was characterized by a single-letter formula. Finally, we have shown that independent non-signaling assistance does not change the capacity region.

Our results suggest that quantum entanglement may also increase the capacity of such channels. However, even for the binary adder channel, this question remains open. One could also ask if such efficient methods to compute the best non-signaling assisted codes can be extended to Gaussian multiple-access channels. Finally, establishing a single-letter formula for the non-signaling assisted capacity of multiple-access channels is the main open question left by this work. It remains open even for the binary adder channel. Proving that non-signaling assistance and relaxed non-signaling assistance coincide asymptotically would directly answer this question and show that the capacity region is described in Theorem~\ref{theo:CharaNSrelaxed}.

\section*{Acknowledgements}
We would like to thank Mario Berta and Andreas Winter for discussions about the use of non-signaling correlations for multiple access channels as well as the anonymous referees for their suggestions that improved the paper.
This work is funded by the European Research Council (ERC Grant AlgoQIP, Agreement No. 851716). We also acknowledge funding from the European Union’s Horizon 2020 research and innovation programme under Grant Agreement No 101017733 within the QuantERA II Programme.

\bibliographystyle{unsrturl}
\bibliography{MAC_NS}

\begin{thebibliography}{10}

\bibitem{FF22}
Omar Fawzi and Paul Ferm{\'{e}}.
\newblock Beating the sum-rate capacity of the binary adder channel with
  non-signaling correlations.
\newblock In {\em {IEEE} International Symposium on Information Theory, {ISIT}
  2022, Espoo, Finland, June 26 - July 1, 2022}, pages 2750--2755. {IEEE},
  2022.
\newblock \href {https://doi.org/10.1109/ISIT50566.2022.9834699}
  {\path{doi:10.1109/ISIT50566.2022.9834699}}.

\bibitem{Liao73}
Henry~Herng{-}Jiunn Liao.
\newblock Multiple access channels ({Ph.D. Thesis} abstr.).
\newblock {\em {IEEE} Trans. Inf. Theory}, 19(2):253, 1973.
\newblock \href {https://doi.org/10.1109/TIT.1973.1054960}
  {\path{doi:10.1109/TIT.1973.1054960}}.

\bibitem{Ahlswede73}
Rudolf Ahlswede.
\newblock Multi-way communication channels.
\newblock In {\em 2nd International Symposium on Information Theory,
  Tsahkadsor, Armenia, USSR, September 2-8, 1971}, pages 23--52. Publishing
  House of the Hungarian Academy of Science, 1973.

\bibitem{BCPSW14}
Nicolas Brunner, Daniel Cavalcanti, Stefano Pironio, Valerio Scarani, and
  Stephanie Wehner.
\newblock Bell nonlocality.
\newblock {\em Rev. Mod. Phys.}, 86:419--478, Apr 2014.
\newblock \href {https://doi.org/10.1103/RevModPhys.86.419}
  {\path{doi:10.1103/RevModPhys.86.419}}.

\bibitem{PLMK11}
R.~Prevedel, Y.~Lu, W.~Matthews, R.~Kaltenbaek, and K.~J. Resch.
\newblock Entanglement-enhanced classical communication over a noisy classical
  channel.
\newblock {\em Phys. Rev. Lett.}, 106:110505, Mar 2011.
\newblock \href {https://doi.org/10.1103/PhysRevLett.106.110505}
  {\path{doi:10.1103/PhysRevLett.106.110505}}.

\bibitem{BF18}
Siddharth Barman and Omar Fawzi.
\newblock Algorithmic aspects of optimal channel coding.
\newblock {\em {IEEE} Trans. Inf. Theory}, 64(2):1038--1045, 2018.
\newblock \href {https://doi.org/10.1109/TIT.2017.2696963}
  {\path{doi:10.1109/TIT.2017.2696963}}.

\bibitem{BSST99}
Charles~H Bennett, Peter~W Shor, John~A Smolin, and Ashish~V Thapliyal.
\newblock Entanglement-assisted classical capacity of noisy quantum channels.
\newblock {\em Physical Review Letters}, 83(15):3081, 1999.
\newblock \href {https://doi.org/10.1103/PhysRevLett.83.3081}
  {\path{doi:10.1103/PhysRevLett.83.3081}}.

\bibitem{Matthews12}
William Matthews.
\newblock A linear program for the finite block length converse of
  {Polyanskiy-Poor-Verd{\'{u}}} via nonsignaling codes.
\newblock {\em {IEEE} Trans. Inf. Theory}, 58(12):7036--7044, 2012.
\newblock \href {https://doi.org/10.1109/TIT.2012.2210695}
  {\path{doi:10.1109/TIT.2012.2210695}}.

\bibitem{QS17}
Yihui Quek and Peter~W. Shor.
\newblock Quantum and superquantum enhancements to two-sender, two-receiver
  channels.
\newblock {\em Phys. Rev. A}, 95:052329, May 2017.
\newblock \href {https://doi.org/10.1103/PhysRevA.95.052329}
  {\path{doi:10.1103/PhysRevA.95.052329}}.

\bibitem{LALS20}
Felix Leditzky, Mohammad~A Alhejji, Joshua Levin, and Graeme Smith.
\newblock Playing games with multiple access channels.
\newblock {\em Nature communications}, 11(1):1--5, 2020.
\newblock \href {https://doi.org/10.1038/s41467-020-15240-w}
  {\path{doi:10.1038/s41467-020-15240-w}}.

\bibitem{SLSS22}
Akshay Seshadri, Felix Leditzky, Vikesh Siddhu, and Graeme Smith.
\newblock On the separation of correlation-assisted sum capacities of multiple
  access channels.
\newblock In {\em {IEEE} International Symposium on Information Theory, {ISIT}
  2022, Espoo, Finland, June 26 - July 1, 2022}, pages 2756--2761. {IEEE},
  2022.
\newblock \href {https://doi.org/10.1109/ISIT50566.2022.9834459}
  {\path{doi:10.1109/ISIT50566.2022.9834459}}.

\bibitem{Mermin90}
N~David Mermin.
\newblock Simple unified form for the major no-hidden-variables theorems.
\newblock {\em Physical review letters}, 65(27):3373, 1990.
\newblock \href {https://doi.org/10.1103/PhysRevLett.65.3373}
  {\path{doi:10.1103/PhysRevLett.65.3373}}.

\bibitem{Peres90}
Asher Peres.
\newblock Incompatible results of quantum measurements.
\newblock {\em Physics Letters A}, 151(3-4):107--108, 1990.
\newblock \href {https://doi.org/10.1016/0375-9601(90)90172-K}
  {\path{doi:10.1016/0375-9601(90)90172-K}}.

\bibitem{Aravind02}
PK~Aravind.
\newblock A simple demonstration of {Bell}'s theorem involving two observers
  and no probabilities or inequalities.
\newblock {\em CoRR}, abs/quant-ph/0206070, 2002.
\newblock \href {https://arxiv.org/abs/quant-ph/0206070}
  {\path{arXiv:quant-ph/0206070}}.

\bibitem{BBT05}
Gilles Brassard, Anne Broadbent, and Alain Tapp.
\newblock Quantum pseudo-telepathy.
\newblock {\em Foundations of Physics}, 35(11):1877--1907, 2005.
\newblock \href {https://doi.org/10.1007/s10701-005-7353-4}
  {\path{doi:10.1007/s10701-005-7353-4}}.

\bibitem{Noetzel20}
Janis Noetzel.
\newblock Entanglement-enabled communication.
\newblock {\em {IEEE} J. Sel. Areas Inf. Theory}, 1(2):401--415, 2020.
\newblock \href {https://doi.org/10.1109/jsait.2020.3017121}
  {\path{doi:10.1109/jsait.2020.3017121}}.

\bibitem{ND20}
Janis N{\"{o}}tzel and Stephen Diadamo.
\newblock Entanglement-enabled communication for the internet of things.
\newblock In Mohammad~S. Obaidat, Kuei{-}Fang Hsiao, Petros Nicopolitidis, and
  Daniel~Cascado Caballero, editors, {\em International Conference on Computer,
  Information and Telecommunication Systems, {CITS} 2020, Hangzhou, China,
  October 5-7, 2020}, pages 1--6. {IEEE}, 2020.
\newblock \href {https://doi.org/10.1109/CITS49457.2020.9232550}
  {\path{doi:10.1109/CITS49457.2020.9232550}}.

\bibitem{JNWY20}
Zhengfeng Ji, Anand Natarajan, Thomas Vidick, John Wright, and Henry Yuen.
\newblock {MIP*=RE}.
\newblock {\em CoRR}, abs/2001.04383, 2020.
\newblock URL: \url{https://arxiv.org/abs/2001.04383}, \href
  {https://arxiv.org/abs/2001.04383} {\path{arXiv:2001.04383}}.

\bibitem{PDB23}
Uzi Pereg, Christian Deppe, and Holger Boche.
\newblock The multiple-access channel with entangled transmitters.
\newblock {\em CoRR}, abs/2303.10456, 2023.
\newblock \href {https://arxiv.org/abs/2303.10456} {\path{arXiv:2303.10456}},
  \href {https://doi.org/10.48550/arXiv.2303.10456}
  {\path{doi:10.48550/arXiv.2303.10456}}.

\bibitem{HDW08}
Min{-}Hsiu Hsieh, Igor Devetak, and Andreas~J. Winter.
\newblock Entanglement-assisted capacity of quantum multiple-access channels.
\newblock {\em {IEEE} Trans. Inf. Theory}, 54(7):3078--3090, 2008.
\newblock \href {https://doi.org/10.1109/TIT.2008.924726}
  {\path{doi:10.1109/TIT.2008.924726}}.

\bibitem{CT01}
Thomas~M. Cover and Joy~A. Thomas.
\newblock {\em Elements of Information Theory}.
\newblock Wiley, 2001.
\newblock \href {https://doi.org/10.1002/0471200611}
  {\path{doi:10.1002/0471200611}}.

\bibitem{LinearProgramming}
Bernd G{\"{a}}rtner and Jir{\'{\i}} Matousek.
\newblock {\em Understanding and using linear programming}.
\newblock Universitext. Springer, 2007.
\newblock \href {https://doi.org/10.1007/978-3-540-30717-4}
  {\path{doi:10.1007/978-3-540-30717-4}}.

\bibitem{Lindstrom69}
Bernt Lindstr{\"o}m.
\newblock Determination of two vectors from the sum.
\newblock {\em Journal of Combinatorial Theory}, 6(4):402--407, 1969.
\newblock \href {https://doi.org/10.1016/S0021-9800(69)80038-4}
  {\path{doi:10.1016/S0021-9800(69)80038-4}}.

\bibitem{Tilborg78}
Henk C.~A. van Tilborg.
\newblock An upper bound for codes in a two-access binary erasure channel
  {(Corresp.)}.
\newblock {\em {IEEE} Trans. Inf. Theory}, 24(1):112--116, 1978.
\newblock \href {https://doi.org/10.1109/TIT.1978.1055814}
  {\path{doi:10.1109/TIT.1978.1055814}}.

\bibitem{KL78}
Tadao Kasami and Shu Lin.
\newblock Bounds on the achievable rates of block coding for a memoryless
  multiple-access channel.
\newblock {\em {IEEE} Trans. Inf. Theory}, 24(2):187--197, 1978.
\newblock \href {https://doi.org/10.1109/TIT.1978.1055860}
  {\path{doi:10.1109/TIT.1978.1055860}}.

\bibitem{Weldon78}
E.~J.~Weldon Jr.
\newblock Coding for a multiple-access channel.
\newblock {\em Inf. Control.}, 36(3):256--274, 1978.
\newblock \href {https://doi.org/10.1016/S0019-9958(78)90312-1}
  {\path{doi:10.1016/S0019-9958(78)90312-1}}.

\bibitem{KLWY83}
Tadao Kasami, Shu Lin, Victor~K.{-}W. Wei, and Saburo Yamamura.
\newblock Graph theoretic approaches to the code construction for the two-user
  multiple-access binary adder channel.
\newblock {\em {IEEE} Trans. Inf. Theory}, 29(1):114--130, 1983.
\newblock \href {https://doi.org/10.1109/TIT.1983.1056614}
  {\path{doi:10.1109/TIT.1983.1056614}}.

\bibitem{BT85}
P.~A. B. M.~Coebergh van~den Braak and Henk C.~A. van Tilborg.
\newblock A family of good uniquely decodable code pairs for the two-access
  binary adder channel.
\newblock {\em {IEEE} Trans. Inf. Theory}, 31(1):3--9, 1985.
\newblock \href {https://doi.org/10.1109/TIT.1985.1057004}
  {\path{doi:10.1109/TIT.1985.1057004}}.

\bibitem{BB98}
Shraga~I. Bross and Ian~F. Blake.
\newblock Upper bound for uniquely decodable codes in a binary input {N}-user
  adder channel.
\newblock {\em {IEEE} Trans. Inf. Theory}, 44(1):334--340, 1998.
\newblock \href {https://doi.org/10.1109/18.651062}
  {\path{doi:10.1109/18.651062}}.

\bibitem{UL98}
R{\"u}diger Urbanke and Quinn Li.
\newblock The zero-error capacity region of the 2-user synchronous {BAC} is
  strictly smaller than its {Shannon} capacity region.
\newblock In {\em 1998 Information Theory Workshop (Cat. No. 98EX131)},
  page~61. IEEE, 1998.
\newblock \href {https://doi.org/10.1109/ITW.1998.706434}
  {\path{doi:10.1109/ITW.1998.706434}}.

\bibitem{AB99}
Rudolf Ahlswede and Vladimir~B. Balakirsky.
\newblock Construction of uniquely decodable codes for the two-user binary
  adder channel.
\newblock {\em {IEEE} Trans. Inf. Theory}, 45(1):326--330, 1999.
\newblock \href {https://doi.org/10.1109/18.746834}
  {\path{doi:10.1109/18.746834}}.

\bibitem{MO05}
M.~Mattas and Patric R.~J. {\"{O}}sterg{\aa}rd.
\newblock A new bound for the zero-error capacity region of the two-user binary
  adder channel.
\newblock {\em {IEEE} Trans. Inf. Theory}, 51(9):3289--3291, 2005.
\newblock \href {https://doi.org/10.1109/TIT.2005.853309}
  {\path{doi:10.1109/TIT.2005.853309}}.

\bibitem{OS15}
Or~Ordentlich and Ofer Shayevitz.
\newblock A {VC}-dimension-based outer bound on the zero-error capacity of the
  binary adder channel.
\newblock In {\em {IEEE} International Symposium on Information Theory, {ISIT}
  2015, Hong Kong, China, June 14-19, 2015}, pages 2366--2370. {IEEE}, 2015.
\newblock \href {https://doi.org/10.1109/ISIT.2015.7282879}
  {\path{doi:10.1109/ISIT.2015.7282879}}.

\bibitem{mosek}
Erling~D Andersen and Knud~D Andersen.
\newblock The {M}osek interior point optimizer for linear programming: an
  implementation of the homogeneous algorithm.
\newblock In {\em High performance optimization}, pages 197--232. Springer,
  2000.
\newblock \href {https://doi.org/10.1007/978-1-4757-3216-0_8}
  {\path{doi:10.1007/978-1-4757-3216-0_8}}.

\bibitem{PPV10}
Yury Polyanskiy, H.~Vincent Poor, and Sergio Verd{\'{u}}.
\newblock Channel coding rate in the finite blocklength regime.
\newblock {\em {IEEE} Trans. Inf. Theory}, 56(5):2307--2359, 2010.
\newblock \href {https://doi.org/10.1109/TIT.2010.2043769}
  {\path{doi:10.1109/TIT.2010.2043769}}.

\bibitem{GK11}
Abbas~El Gamal and Young{-}Han Kim.
\newblock {\em Network Information Theory}.
\newblock Cambridge University Press, 2011.
\newblock \href {https://doi.org/10.1017/CBO9781139030687}
  {\path{doi:10.1017/CBO9781139030687}}.

\end{thebibliography}
\end{document}